\documentclass[11pt]{article}

\usepackage[margin=1in]{geometry}
\usepackage{amsmath,amssymb,amsthm,mathtools}
\usepackage{booktabs}
\usepackage{needspace}
\usepackage{array}
\usepackage{enumitem}
\usepackage{setspace}
\usepackage{microtype}
\usepackage{natbib}
\usepackage{xcolor}
\usepackage{hyperref}
\hypersetup{hidelinks}

\numberwithin{equation}{section}

\newtheorem{proposition}{Proposition}[section]

\newtheorem{corollary}[proposition]{Corollary}
\newtheorem{lemma}[proposition]{Lemma}
\newtheorem{assumption}{Assumption}[section]
\newtheorem{remark}{Remark}[section]

\newcommand{\E}{\mathbb{E}}
\newcommand{\Var}{\operatorname{Var}}
\newcommand{\Cov}{\operatorname{Cov}}
\newcommand{\1}{\mathbf{1}}
\newcommand{\cU}{\mathcal U}
\newcommand{\cS}{\mathcal S}
\newcommand{\cC}{\mathcal C}
\newcommand{\R}{\mathbb R}

\newcommand{\sate}{\tau_{\cS}}

\newcommand{\anchor}[1]{}
\newcounter{linkageexample}

\title{Randomized Experiments with Ambiguous Record Linkage}
\author{Bruno Ferman \qquad Lucas Finamor\\
\small Sao Paulo School of Economics - FGV}
\date{First draft: October 9, 2026}

\begin{document}
\maketitle

\begin{abstract}
Randomized experiments are often linked to administrative data to measure outcomes, but coarse identifiers can leave multiple administrative records consistent with each experimental unit. We show that unbiased estimation of the sample average treatment effect is possible without resolving these ambiguous links or specifying an outcome model. Our key linkage requirement is that the true administrative outcome is contained among the candidate outcomes for each experimental unit. We characterize the precision cost of linkage ambiguity and study two complementary ways to mitigate it: regression adjustment and reweighting based on observed linkage characteristics. We establish finite-sample and asymptotic properties of these estimators and develop valid randomization-based and large-sample inference.
\end{abstract}

\section{Introduction}
\label{sec:intro}

Administrative records provide an attractive way to measure outcomes in randomized experiments. They can reduce survey costs, avoid attrition from follow-up surveys, and measure outcomes that are difficult to collect directly. A practical obstacle is that the experimental file and the administrative file may not share a unique identifier. If the available identifiers are coarse---for example, names shared by multiple individuals or only a subset of the digits of a national identification number---one experimental unit may be compatible with several administrative records. These compatible records form the unit's \emph{candidate set}.

We show that the sample average treatment effect (SATE) can be identified without resolving these ambiguous links. For candidate sets defined by fixed coarse identifiers, we require that the true administrative outcome of each experimental unit be contained among its candidate outcomes. Rather than selecting or imputing the true record, we construct an alternative outcome by summing the outcomes of all administrative records consistent with the observed identifier. This sum contains the true outcome together with nuisance components contributed by false candidates. We exploit the known treatment-assignment mechanism to account for these components and construct unbiased estimators of the SATE without identifying the true record or assigning probabilities to candidate identities.

Our setting is particularly relevant when the administrative data contain the outcome of interest for a population that is broader than the experimental sample, but individual records can only be linked to experimental units through coarse identifiers. Examples include earnings from tax records, hospitalizations from health records, and birth outcomes from vital-statistics records. When absence from the administrative outcome file implies that the true outcome is zero, the database need only contain records for units with nonzero outcomes. For mortality or participation in a social program, for example, absence from the administrative outcome file can then be interpreted as a zero outcome, provided all nonzero outcomes are recorded.

We establish both finite-sample and large-sample results in a finite-population framework. For expositional clarity, our main analysis focuses on independent Bernoulli treatment assignment. The baseline estimator applies the Horvitz--Thompson contrast to the candidate-sum outcomes and is exactly unbiased for the SATE. When the experimental sample is drawn representatively from a well-defined target subpopulation, it is also unbiased for the corresponding subpopulation average treatment effect.  We also characterize the precision cost associated with ambiguous linkage and develop valid randomization-based and large-sample inference that accounts for the dependence induced by shared coarse identifiers.

The precision loss from ambiguous linkage can be mitigated in two complementary ways. First, reweighting assigns less weight to linkage cells expected to contribute more to the estimator’s variance. For example, cells with more candidate records may generate noisier candidate sums. Under heterogeneous treatment effects, however, reweighting generally changes the target from the SATE to a different weighted average of treatment effects. Second, regression adjustment can remove predictable differences in candidate-sum levels. When larger linkage cells have larger expected candidate sums, adjusting for cell size can reduce variation caused by chance imbalances between treatment and control. Sample-based adjustment can introduce finite-sample bias but preserves consistency for the SATE. The two strategies can also be combined.

Finally, we extend the analysis to stratified randomization with fixed treatment counts, including complete randomization as a special case. Under these designs, dependence between assignments can generate bias when experimental units share both a coarse identifier and a randomization stratum. We characterize this bias and provide a design-based correction that restores exact unbiasedness under suitable conditions. We establish asymptotic theory and variance estimators both for a fixed number of growing strata and for a growing number of small strata. The analysis also accommodates reweighting and regression adjustment.

There is a large literature on the analysis of imperfectly linked data
(see \cite{kamatgutman2026} for a recent survey). However,  existing methods need not consistently estimate the
SATE in randomized experiments. For instance, methods that rely on
unambiguous matches can fail when those matches disproportionately
represent individuals with different treatment effects, while pooled
corrections for linkage error can fail when treatment effects vary with
linkage ambiguity. Appendix~\ref{app:alternative-methods} analyzes a
range of methods, providing examples of settings in which they may fail to recover the SATE
\citep{lahiri2005,chambers2009,kimchambers2012a,kimchambers2012b,hofzwinderman2012,goldstein2012,hofzwinderman2015,steorts2015,sadinle2017,wortman2018,hanlahiri2019,chipperfield2019,chamberssilva2020,shan2021,zhangtuoto2021,zhangtuoto2024,slawskietal2025,bukkeslawski2025,hureiter2027}.
Other approaches use joint Bayesian models for record linkage and relationships among variables in the
linked files, including models for causal inference
\citep{gutman2013,guha2022,kamatshan2023}. Their validity for SATE
estimation, though, depends on the specification and identification of the joint
model. In contrast, under the stated assignment and candidate-set
conditions, we propose estimators that recover the SATE without
specifying an outcome-distribution model, estimating correct-link
probabilities, or restricting how treatment effects vary with linkage
ambiguity. We also develop a theory for precision-motivated weighted
estimators, making explicit the trade-off between precision and retaining
the SATE as the target estimand in settings with ambiguous linkage.

Our candidate-sum approach also accommodates outcomes measured using administrative files containing only nonzero records, whenever absence from the file implies that the outcome is zero, a data structure not directly covered by the implementations compared above. \citet{tahamont2021} study linkage errors in this setting for binary outcomes. Correcting the attenuation characterized in their analysis, however, requires knowledge of true-positive and false-positive rates, which are not generally identified from the positive-record file and treatment assignments alone, even under random assignment. In contrast, our candidate-sum SATE estimators exploit positive-record counts without requiring knowledge of these rates. Finally, other approaches use additional information, such as verified links from clerical review or an audit sample \citep{dasylva2018,gutman2016}. Our candidate-sum SATE estimators do not require such information.

\section{Setup and oracle benchmark}
\label{sec:setup}

\subsection{Finite administrative population and experimental sample}

Consider a fixed finite administrative population
\[
\cU=\{1,\ldots,M\}.
\]

We assume a well-defined treatment and no interference across units in $\cU$ (SUTVA), so each unit's outcome depends only on its own treatment assignment. Unit $i\in\cU$ has fixed potential outcomes $Y_i(0)$ and $Y_i(1)$ and treatment effect
\[
\tau_i=Y_i(1)-Y_i(0).
\]

The realized experimental sample is a fixed subset $\cS\subset\cU$ of size $N$. We do not require $\cS$ to be a random sample from $\cU$. This reflects the setting in which the administrative file may cover a broad population while the experiment is conducted in a particular region or among a selected set of participants. We take the SATE as our target.
\begin{equation}
\tau_{\cS}=\frac{1}{N}\sum_{i\in\cS}\tau_i.
\label{eq:sample-ate}
\end{equation}

{If the experimental sample is randomly drawn from a well-defined subpopulation, the corresponding subpopulation average treatment effect can instead be taken as the target by incorporating the sampling design.}

\subsection{Treatment assignment and administrative outcomes}
\label{sec:assignment}

For simplicity, the main analysis uses independent Bernoulli treatment assignment. For every $i\in\cS$,
\begin{equation}
D_i\sim\operatorname{Bernoulli}(p),
\qquad 0<p<1,
\label{eq:bern-design}
\end{equation}
independently across experimental units. Let $\mathbf D=(D_i)_{i\in\cS}$. For notational convenience set $D_i=0$ for $i\notin\cS$, because treatment is not available outside the experiment. We consider complete and stratified randomization with fixed treatment counts in Section~\ref{sec:other-designs}.

In the full-register setting, the administrative file contains one outcome record per unit, reporting $Y_i^{\mathrm{adm}}$.
 The administrative outcome observed after assignment is
\begin{equation}
Y_i^{\mathrm{adm}}
=
Y_i(D_i)
=
Y_i(0)+D_i\tau_i,
\qquad i\in\cU.
\label{eq:admin-y}
\end{equation}

\subsection{Coarse identifiers and candidate outcomes}

The observed coarse identifier partitions the administrative population into linkage cells $\cC_1,\ldots,\cC_{G}$. Let
\[
M_g=|\cC_g|
\]
be the number of administrative records in linkage cell $g$. For each experimental unit $i$, let $g(i)$ denote the index of the administrative linkage cell to which the coarse identifier recorded in the experimental file maps that unit. Our coverage requirement is that $i\in\cC_{g(i)}$ for every $i\in\cS$, so that each unit's true administrative record is included among its candidates.

The partition and the linkage-cell sizes may be arbitrarily related to potential outcomes and treatment effects. We impose no restriction on the number, composition, or outcomes of the false candidates. However, we assume that both the administrative partition and the mapping $g(i)$ are unaffected by treatment assignment. 

For each linkage cell define the number of experimental units and treated experimental units,
\begin{equation}
 N_g^{\mathrm{cell}}=|\cC_g\cap\cS|,
 \qquad
 T_g^{\mathrm{cell}}=\sum_{i\in\cC_g\cap\cS}D_i,
\label{eq:ng-tg}
\end{equation}
and define the untreated administrative total in linkage cell $g$
\begin{equation}
B_g=\sum_{j\in\cC_g}Y_j(0).
\label{eq:bg}
\end{equation}

The total observed administrative outcome in linkage cell $g$ is
\begin{equation}
\widetilde Y_g(\mathbf D)
=
\sum_{j\in\cC_g}Y_j^{\mathrm{adm}}
=
B_g+
\sum_{j\in\cC_g\cap\cS}D_j\tau_j.
\label{eq:linked-cell-outcome}
\end{equation}

Every experimental unit $i\in\cC_g$ is assigned the same candidate-sum outcome,
\begin{equation}
\widetilde Y_i(\mathbf D)=\widetilde Y_{g(i)}(\mathbf D),
\label{eq:linked-unit-outcome}
\end{equation}
where this observed outcome may depend on the treatment assignments of other experimental units sharing the same coarse identifier.

\begin{remark}[Administrative files omitting zero outcomes]
Suppose the administrative outcome file contains only records for which $Y_i^{\mathrm{adm}}\neq 0$, with absence from the file implying $Y_i^{\mathrm{adm}}=0$. The sum in \eqref{eq:linked-cell-outcome} is still observed because records with zero outcomes contribute nothing to it. The relevant coverage condition is that every realized nonzero outcome is recorded and retained by the candidate rule.
The set of observed nonzero-outcome records may change with assignment, but the underlying linkage partition and the mapping $g(i)$ remain fixed.
\end{remark}
\subsection{Oracle benchmark}

If exact linkage were available, we would observe each experimental unit's own outcome $Y_i(D_i)$. As a benchmark, consider the infeasible Horvitz--Thompson estimator \citep{horvitz1952} under Bernoulli assignment:\footnote{Compared with the usual difference-in-means estimator, the Horvitz--Thompson estimator uses the known assignment probabilities $p$ and $1-p$, rather than the realized treatment shares $N_1/N$ and $1-N_1/N$, where $N_1=\sum_{i\in\cS}D_i$. Under Bernoulli assignment, both estimators are unbiased for the oracle estimand, conditional on nondegenerate realized treatment counts in the case of difference-in-means. We use the Horvitz--Thompson form here because its extension to the candidate-sum outcomes introduced below remains exactly unbiased in the presence of ambiguous linkage, which simplifies the analysis. Section~\ref{sec:adjustment} studies a broader class of candidate-sum estimators that includes the corresponding difference-in-means estimator as a special case.}

\begin{equation}
\widehat\tau^{\mathrm O}
=
\frac{1}{N}\sum_{i\in\cS}
\left\{
\frac{D_iY_i(D_i)}{p}
-
\frac{(1-D_i)Y_i(D_i)}{1-p}
\right\}.
\label{eq:oracle}
\end{equation}

This estimator satisfies
\begin{equation}
\E_D[\widehat\tau^{\mathrm O}]=\tau_{\cS},
\label{eq:oracle-unbiased}
\end{equation}
and independence across units gives
\begin{equation}
\Var_D(\widehat\tau^{\mathrm O})
=
\frac{1}{N^2}\sum_{i\in\cS}
\left\{
\frac{Y_i(1)^2}{p}
+
\frac{Y_i(0)^2}{1-p}
-
\tau_i^2
\right\},
\label{eq:oracle-var}
\end{equation}
where expectations and variances indexed by $D$ are taken over the treatment assignment, conditional on the fixed administrative population, experimental sample, potential outcomes, and linkage cells.

The oracle will serve as a precision benchmark. Our objective is to preserve the causal identification that exact linkage would deliver while quantifying the additional randomization variance created by ambiguous linkage and mitigating the resulting loss of precision.

\section{Finite-sample results under Bernoulli assignment}
\label{sec:finite}

\subsection{Baseline estimator}
\label{sec:baseline}

We now construct an estimator that is exactly unbiased for the SATE under independent Bernoulli assignment, even when individual links are unresolved and several experimental units share the same candidate set. Define the baseline estimator
\begin{align}
\widehat\tau^{\mathrm B}
&=
\frac{1}{N}\sum_{i\in\cS}
\left\{
\frac{D_i\widetilde Y_i(\mathbf D)}{p}
-
\frac{(1-D_i)\widetilde Y_i(\mathbf D)}{1-p}
\right\}
\notag\\
&=
\frac{1}{N}\sum_{g=1}^{G}
\frac{T_g^{\mathrm{cell}}-pN_g^{\mathrm{cell}}}{p(1-p)}\,\widetilde Y_g(\mathbf D).
\label{eq:baseline}
\end{align}

The first expression compares candidate-sum outcomes across treatment arms, dividing the treated and control totals by the expected numbers of treated and control units, $Np$ and $N(1-p)$, respectively. The second expression writes the same estimator in terms of linkage-cell contributions. Each linkage cell's administrative outcome total is weighted by $T_g^{\mathrm{cell}}-pN_g^{\mathrm{cell}}$, the difference between its realized and expected numbers of treated experimental units, and the sum is normalized by $Np(1-p)$. Thus, linkage cells with more treated units than expected receive positive weights, whereas linkage cells with fewer treated units than expected receive negative weights. Note that linkage cells containing no experimental units have $N_g^{\mathrm{cell}}=T_g^{\mathrm{cell}}=0$ and therefore make no contribution to the estimator.

We next characterize the expectation and variance of the baseline estimator.

\begin{proposition}[Exact unbiasedness and variance]
\label{prop:baseline}
Under independent Bernoulli assignment,
\begin{equation}
\E_D[\widehat\tau^{\mathrm B}]=\tau_{\cS}.
\label{eq:baseline-unbiased}
\end{equation}

Moreover,
\begin{equation}
\Var_D(\widehat\tau^{\mathrm B})
=
\frac{1}{N^2}\sum_{g=1}^{G}V_g,
\label{eq:baseline-var}
\end{equation}
where
\begin{equation}
V_g
=
\theta_g^2
+
\frac{
 N_g^{\mathrm{cell}}(B_g+p\theta_g)^2
 +2(1-2p)(B_g+p\theta_g)\theta_g
 +\{1+(N_g^{\mathrm{cell}}-6)p(1-p)\}Q_g
}{p(1-p)},
\label{eq:vg}
\end{equation}
with
\[
\theta_g=\sum_{i\in\cC_g\cap\cS}\tau_i,
\qquad
Q_g=\sum_{i\in\cC_g\cap\cS}\tau_i^2.
\]
\end{proposition}

Proofs of all propositions and corollaries are provided in Appendix~\ref{app:proofs}.

The expectation result is the central design-based property. Linkage-cell size can be highly informative about potential outcomes, false candidates can have systematically different outcomes from true candidates, and ambiguity can vary arbitrarily across experimental units. None of these features creates bias under independent treatment assignment.

A key potential cost of ambiguous linkage is reduced precision. To make the comparison with exact linkage transparent,  consider the case in which no linkage cell contains more than one experimental unit. 

\begin{corollary}[Variance relative to exact linkage when $N_g^{\mathrm{cell}}\leq1$]
\label{cor:variance-gap}
Suppose $N_g^{\mathrm{cell}}\leq1$ for every linkage cell. Then
\begin{align}
&\Var_D(\widehat\tau^{\mathrm B})
-
\Var_D(\widehat\tau^{\mathrm O})
\notag\\
&\quad=
\frac{1}{N^2p(1-p)}\sum_{i\in\cS}
\left[
\{B_{g(i)}+(1-p)\tau_i\}^2
-
\{Y_i(0)+(1-p)\tau_i\}^2
\right].
\label{eq:variance-gap}
\end{align}

If \(Y_i(0),Y_i(1)\geq0\) for all units, then the variance difference in \eqref{eq:variance-gap} is nonnegative.

\end{corollary}

For unrestricted real-valued outcomes there is no universal pointwise variance ordering: false-candidate outcomes could, by chance, generate sufficiently strong negative cancellation with the true outcome. Such a variance reduction requires systematic cancellation within linkage cells. In the absence of such cancelation, ambiguous linkage increases variance, which motivates the efficiency methods below.

\subsection{Precision weighting}
\label{sec:weighting}

The precision loss from ambiguous linkage need not be uniform across experimental units. The number and composition of administrative candidates can help predict how much each linkage cell contributes to the variance of the baseline estimator. We investigate whether precision can be improved by reweighting units based on these pre-treatment linkage characteristics. Because reweighting can also change the treatment-effect target, we characterize both the variance and the estimand of the weighted estimator. Let $w_g>0$ be a weight determined only by pre-treatment linkage information, such as $M_g$, name prevalence, or other features of the coarse identifier, and define
\[
W_{\cS}=\sum_g N_g^{\mathrm{cell}}w_g.
\]

The weighted estimator is
\begin{equation}
\widehat\tau^{\mathrm W}
=
\frac{1}{W_{\cS}}
\sum_g
w_g\frac{T_g^{\mathrm{cell}}-pN_g^{\mathrm{cell}}}{p(1-p)}\widetilde Y_g(\mathbf D).
\label{eq:weighted}
\end{equation}

Its estimand is the weighted sample effect
\begin{equation}
\tau_{w,\cS}
=
\frac{1}{W_{\cS}}
\sum_g w_g\theta_g.
\label{eq:weighted-target}
\end{equation}

\begin{proposition}[Fixed precision weights]
\label{prop:weighted}
If the weights are fixed with respect to treatment assignment, then
\begin{equation}
\E_D[\widehat\tau^{\mathrm W}]=\tau_{w,\cS}
\label{eq:weighted-unbiased}
\end{equation}
and
\begin{equation}
\Var_D(\widehat\tau^{\mathrm W})
=
\frac{1}{W_{\cS}^2}\sum_gw_g^2V_g,
\label{eq:weighted-var}
\end{equation}
with $V_g$ given by \eqref{eq:vg}. If treatment effects are homogeneous, $\tau_i=\tau$ for every $i\in\cS$, then $\tau_{w,\cS}=\tau_{\cS}=\tau$ for every positive weighting scheme.
\end{proposition}

The proposition makes the tradeoff explicit. Under heterogeneous effects, precision weighting changes the target. Under homogeneous effects, all positive weighting schemes have the same target, so weights can be chosen solely to reduce variance.

A simple outcome-free choice is
\begin{equation}
w_g=\frac{1}{M_g}.
\label{eq:inverse-k}
\end{equation}

Since candidate-sum outcomes aggregate $M_g$ administrative
outcomes, larger linkage cells may contribute more to the
estimator's variance. Inverse-size weighting downweights these
cells without using observed outcomes.
Section~\ref{sec:efficiency} characterizes asymptotically optimal
weights and discusses conditions under which inverse-size
weighting is optimal when combined with regression adjustment.

\begin{remark}[Linkage-cell-size weighting when zero outcomes are omitted]
When the administrative outcome file omits records with $Y=0$, the full linkage-cell size $M_g$ is generally not observed. Using instead the number of candidates with $Y\neq0$ may be problematic, since this quantity may itself be affected by treatment. An alternative in this case is to use other pre-treatment proxies for $M_g$ or for the degree of linkage ambiguity. For example, one could use the prevalence of the relevant coarse identifier in an auxiliary population register or other pre-treatment data.
\end{remark}

\subsection{Regression adjustment}
\label{sec:adjustment}

We can also seek to improve precision by adjusting for pre-treatment characteristics that predict the level of the candidate-sum outcome. Because this outcome aggregates administrative records, characteristics such as $M_g$ can generate substantial differences in its level across linkage cells. Random treatment assignment may then produce chance imbalances in the distribution of such characteristics between treated and control units, adding variance to the estimator. Regression adjustment can reduce this source of variation by removing the component of the candidate-sum outcome that is predictable from pre-treatment characteristics. Let $x_g\in\R^L$ denote the characteristics used for adjustment. For example, $x_g$ may include $M_g$ and $N_g^{\mathrm{cell}}$. Other pre-treatment covariates can also be included, provided they are defined at, or aggregated to, the linkage-cell level.

For the theory it is useful first to consider a fixed adjustment $h_g$. Define the general weighted-and-adjusted estimator
\begin{equation}
\widehat\tau(w,h)
=
\frac{1}{W_{\cS}}
\sum_g
w_g\frac{T_g^{\mathrm{cell}}-pN_g^{\mathrm{cell}}}{p(1-p)}
\{\widetilde Y_g(\mathbf D)-h_g\}.
\label{eq:generic-estimator}
\end{equation}

The baseline estimator is $\widehat\tau(1,0)$ and the weighted estimator is $\widehat\tau(w,0)$. Let $V_g(h)$ denote the expression in \eqref{eq:vg} with $B_g$ replaced by $B_g-h_g$. That is,
\[
V_g(h)
=
\theta_g^2
+
\frac{
 N_g^{\mathrm{cell}}(B_g-h_g+p\theta_g)^2
 +2(1-2p)(B_g-h_g+p\theta_g)\theta_g
 +\{1+(N_g^{\mathrm{cell}}-6)p(1-p)\}Q_g
}{p(1-p)}.
\]

\begin{proposition}[Fixed adjustment]
\label{prop:fixed-adjustment}
For any fixed adjustment $h_g$,
\begin{equation}
\E_D[\widehat\tau(w,h)]=\tau_{w,\cS},
\qquad
\Var_D\{\widehat\tau(w,h)\}
=
\frac{1}{W_{\cS}^2}\sum_gw_g^2V_g(h).
\label{eq:fixed-adjustment-result}
\end{equation}
\end{proposition}

In particular, with equal weights, $\widehat\tau(1,h)$ remains exactly unbiased for the SATE for any fixed adjustment $h_g$.

In practice, the adjustment function must be estimated.
For linear adjustment, include an intercept in $x_g$
and write $x_g=(1,z_g')'$.
Let $\widehat\beta_1$ and $\widehat\beta_0$ denote the
coefficients from separate weighted least-squares
regressions of the candidate-sum outcome on $x_{g(i)}$
in the treatment and control arms, respectively,
using weights $w_{g(i)}$. We then set
\[
\widehat h_g
=
(1-p)x_g'\widehat\beta_1
+
p x_g'\widehat\beta_0.
\]

The resulting estimator can equivalently be computed
from a single fully interacted weighted regression.
Define the weighted covariate mean and centered covariates
\[
\overline z_w
=
\frac{1}{W_{\cS}}
\sum_{i\in\cS}w_{g(i)}z_{g(i)},
\qquad z_g^c=z_g-\overline z_w,
\]
and let
\begin{equation}
\begin{aligned}
(\widehat a,\widehat t,\widehat b,\widehat c)
&=
\underset{a,t,b,c}{\arg\min}
\sum_{i\in\cS}w_{g(i)}\\
&\quad\times
\left[
\widetilde Y_{g(i)}(\mathbf D)
-a-tD_i
-(z_{g(i)}^c)'b
-D_i(z_{g(i)}^c)'c
\right]^2.
\end{aligned}
\label{eq:interacted-wls}
\end{equation}
Whenever these regressions are identified, the coefficient
on treatment satisfies
$\widehat t=\widehat\tau(w,\widehat h)$ exactly.
Thus, the adjusted estimator is the treatment coefficient
from weighted least squares of the candidate-sum outcome
on an intercept, treatment, centered covariates, and
their interactions with treatment. With equal weights,
this is the fully interacted regression adjustment of
\citet{lin2013}, applied to the candidate-sum outcome.

If \(x_g\) contains only an intercept, the adjusted estimator reduces exactly to the difference in weighted mean candidate-sum outcomes between treated and control units, using weights $w_{g(i)}$. With equal weights, this is the usual difference in means. More generally, \(x_g\) can include pre-treatment characteristics that predict the level of the candidate-sum outcome, such as \(M_g\), and flexible functions of these characteristics. For example, indicators for different values or ranges of \(M_g\) allow the relationship between linkage-cell size and the candidate-sum outcome to be nonlinear. Sparsely represented categories can be pooled or modeled more parsimoniously. Other pre-treatment covariates can also be included in \(x_g\).
Because \(\widehat h_g\) depends on the realized treatment assignment, the exact finite-sample unbiasedness result in Proposition~\ref{prop:fixed-adjustment} no longer applies. The asymptotic analysis below shows that estimating the adjustment does not affect consistency or the first-order limiting distribution.

\subsection{Randomization inference}
\label{sec:ri}

Consider the sharp null
\begin{equation}
H_0^{\mathrm{sharp}}:
Y_i(1)=Y_i(0)
\quad\text{for every }i\in\cS,
\label{eq:sharp-null}
\end{equation}
and a test statistic
$\mathcal T(\widetilde{\mathbf Y},\mathbf D)$, where
$\widetilde{\mathbf Y}
=(\widetilde Y_g(\mathbf D))_{g=1}^G$
is the vector of observed candidate-sum outcomes and
$\mathbf D=(D_i)_{i\in\cS}$ is the treatment-assignment vector.
We adopt the convention that larger values of $\mathcal T$
provide stronger evidence against the null. 

Under \eqref{eq:sharp-null}, every candidate-sum outcome is
fixed with respect to treatment assignment:
\[
\widetilde Y_g(\mathbf d)=B_g
\quad\text{for every }g
\text{ and every }\mathbf d\in\{0,1\}^N.
\]
Consequently, the null distribution of the statistic can be
obtained by holding $\widetilde{\mathbf Y}$ fixed and redrawing
treatment assignments according to the known design.

For two-sided tests, possible choices include
$|\widehat\tau(w,\widehat h)|$ or its studentized version
using the asymptotic variance estimator developed in
Section~\ref{sec:variance-estimation}.
These choices encompass the baseline, weighted, adjusted,
and weighted-and-adjusted estimators.

\begin{proposition}[Exact Fisher randomization inference]
\label{prop:ri}
Let
\[
\mathcal A(\widetilde{\mathbf Y})
=
\left\{
\mathbf d\in\{0,1\}^N:
\mathcal T(\widetilde{\mathbf Y},\mathbf d)
\text{ is defined}
\right\},
\]
and suppose this set has positive probability under the
Bernoulli design in \eqref{eq:bern-design}.
For an observed assignment
$\mathbf D\in\mathcal A(\widetilde{\mathbf Y})$,
define the randomization $p$-value
\begin{equation}
p_{\mathrm{RI}}
=
\Pr_{\mathbf D^*}\!\left\{
\mathcal T(\widetilde{\mathbf Y},\mathbf D^*)
\geq
\mathcal T(\widetilde{\mathbf Y},\mathbf D)
\,\middle|\,
\widetilde{\mathbf Y},\mathbf D,\,
\mathbf D^*\in\mathcal A(\widetilde{\mathbf Y})
\right\},
\label{eq:ri-pvalue}
\end{equation}
where $\mathbf D^*$ is an independent draw from the
Bernoulli design. Under \eqref{eq:sharp-null},
\begin{equation} \label{eq:ri-level}
\Pr_D\!\left\{
p_{\mathrm{RI}}\leq\alpha
\,\middle|\,
\mathbf D\in\mathcal A(\widetilde{\mathbf Y})
\right\}
\leq\alpha
\qquad\text{for every }\alpha\in(0,1).
\end{equation}
Thus, rejecting when $p_{\mathrm{RI}}\leq\alpha$
gives a finite-sample level-$\alpha$ test conditional
on the statistic being defined. Conditional rejection
probability exactly equal to $\alpha$ can be obtained
by randomization at the boundary.
\end{proposition}

The conditioning on
$\mathbf D\in\mathcal A(\widetilde{\mathbf Y})$
in \eqref{eq:ri-level} means that the level guarantee
applies to realizations of the treatment assignment
for which the test statistic can be computed.
For example, for statistics that require observations
in both treatment arms, this excludes assignments
in which all units are treated or all units are controls.
To compute the randomization $p$-value in
\eqref{eq:ri-pvalue}, draws of $\mathbf D^*$ outside
$\mathcal A(\widetilde{\mathbf Y})$ are discarded
and redrawn. When the statistic uses
$\widehat\tau(w,\widehat h)$ with an estimated adjustment,
$\widehat h$ must be re-estimated under each reassignment.
Any variance estimate used for studentization must
likewise be recomputed.

Randomization inference conditions on the fixed outcomes
under the sharp null and reproduces the known assignment
mechanism. The same argument applies to the stratified and
complete randomization designs considered in
Appendix~\ref{app:fixed-count}, provided treatment assignments
are redrawn according to the corresponding randomization
design. Its limitation is that it tests the sharp null that
every individual treatment effect is zero. In many
applications the object of interest is instead a hypothesis
about the SATE while allowing treatment effects to be
heterogeneous. The large-sample theory below provides
inference for such weak null hypotheses.

\section{Asymptotic theory}
\label{sec:asymptotics}

We next study the large-sample behavior of the estimators along sequences of fixed administrative populations and experimental samples, with the experimental sample size $N$ tending to infinity. The asymptotic analysis provides inference for weak null hypotheses, characterizes the first-order behavior of the baseline, weighted, adjusted, and weighted-and-adjusted estimators, and clarifies the role of ambiguous linkage in precision.

\subsection{Asymptotic framework}

For each $N$, the administrative population, experimental sample $\cS$, potential outcomes, and linkage structure are treated as fixed, and randomness arises only from the treatment-assignment mechanism. The coarse identifiers partition the administrative population into linkage cells $\cC_g$ of size $M_g$, while
\[
N_g^{\mathrm{cell}}=|\cC_g\cap\cS|
\]
denotes the number of experimental units in linkage cell $g$. Let
\[
\mathcal G_N=\{g:N_g^{\mathrm{cell}}>0\},
\qquad
G_N=|\mathcal G_N|
\]
denote the represented linkage cells and their number.

\begin{assumption}[Bounded linkage cells]
\label{ass:bounded-cells}
There exists a finite constant $\overline M$ such that
\[
\max_g M_g\leq \overline M
\]
along the sequence.
\end{assumption}

Assumption~\ref{ass:bounded-cells} requires linkage cells to remain uniformly bounded in size as the administrative population grows. It therefore rules out coarse identifiers whose number of matches grows without bound, as could occur, for example, with very common names.\footnote{If some linkage cells grow unusually large, their candidate sums can also grow with linkage-cell size and may dominate the asymptotic variance. In applications, such cases would typically need to be handled separately, for example by excluding extremely large linkage cells or assigning them substantially smaller weights.} The assumption still allows persistent linkage ambiguity: a linkage cell may contain multiple candidate records even asymptotically.

Because $N_g^{\mathrm{cell}}\leq M_g$, Assumption~\ref{ass:bounded-cells} also implies that the number of experimental units in any linkage cell is uniformly bounded. Moreover,
\[
G_N\geq \frac{N}{\overline M},
\]
so the number of linkage cells represented in the experimental sample necessarily diverges as $N\to\infty$. The large-sample approximation is therefore driven by an increasing number of non-dominating linkage cells.

We impose no sampling model for the experimental sample and no additional restriction on the realized degree of overlap among experimental units within linkage cells. Some sequences may have $N_g^{\mathrm{cell}}\leq1$ for every represented linkage cell, while in others multiple experimental units may share linkage cells even asymptotically. The results below cover both cases.

\subsection{A general estimator class}

Let $w_g=w(x_g)>0$ be a bounded pre-treatment weight and let $h_g=h(x_g)$ be a fixed adjustment. Define the linkage-cell contribution
\begin{equation}
H_g(w,h)
=
w_g\frac{T_g^{\mathrm{cell}}-pN_g^{\mathrm{cell}}}{p(1-p)}
\{\widetilde Y_g(\mathbf D)-h_g\}.
\label{eq:Hg}
\end{equation}

Then
\begin{equation}
\widehat\tau(w,h)
=
\frac{\sum_gH_g(w,h)}{W_{\cS}},
\qquad
W_{\cS}=\sum_gw_gN_g^{\mathrm{cell}}.
\label{eq:generic-again}
\end{equation}

The corresponding target is $\tau_{w,\cS}$ in \eqref{eq:weighted-target}. The four estimators considered in the paper are nested as follows:
\begin{center}
\begin{tabular}{lcc}
\toprule
Estimator & Weight & Adjustment \\
\midrule
Baseline $\widehat\tau^{\mathrm B}$ & $w_g=1$ & $h_g=0$ \\
Weighted $\widehat\tau^{\mathrm W}$ & $w_g=w(x_g)$ & $h_g=0$ \\
Adjusted $\widehat\tau^{\mathrm A}$ & $w_g=1$ & $h_g=\widehat h_g$ \\
Weighted + adjusted $\widehat\tau^{\mathrm{WA}}$ & $w_g=w(x_g)$ & $h_g=\widehat h_g$ \\
\bottomrule
\end{tabular}
\end{center}

\begin{assumption}[Regularity]
\label{ass:cells}
Along the sequence:
\begin{enumerate}[label=(\roman*)]
\item $W_{\cS}/N\to w^\ast\in(0,\infty)$ and $0<\underline w\leq w_g\leq\overline w<\infty$;
\item the candidate sums and fixed adjustments have bounded fourth moments on average,
\[
\frac1N\sum_{g\in\mathcal G_N}
\left\{
\E_D\left[\left|\widetilde Y_g(\mathbf D)\right|^4\right]
+|h_g|^4
\right\}=O(1);
\]
\item the exact design variance satisfies
\begin{equation}
V_N(w,h)
\equiv
\frac{N}{W_{\cS}^2}\sum_{g\in\mathcal G_N}w_g^2V_g(h)
\longrightarrow
V(w,h)\in(0,\infty).
\label{eq:VN}
\end{equation}
\end{enumerate}
\end{assumption}

Assumption~\ref{ass:cells} bounds weights and the average fourth moments of candidate sums and fixed adjustments, and requires the root-$N$ design variance to converge to a nondegenerate limit. It does not require uniformly bounded outcomes or a uniform moment bound for every linkage cell. With bounded linkage-cell sizes, part (ii) also bounds the average fourth moments of $H_g(w,h)$, both centered and noncentered, and the average second moments of the candidate sums. Together with Assumption~\ref{ass:bounded-cells}, it gives a triangular array of independent, non-dominating linkage-cell contributions under Bernoulli assignment.

For the adjusted estimators, let
\[
\widehat h_g=h(x_g;\widehat\beta),
\qquad
h_g^*=h(x_g;\beta^*),
\]
where $\beta^*$ is a nonrandom finite-population target. Dependence of
$\beta^*$ on $N$ is left implicit.

\begin{assumption}[Estimated adjustment]
\label{ass:adjustment}
The dimension of $\beta$ is fixed, and:
\begin{enumerate}[label=(\roman*)]
\item
\[
\widehat\beta-\beta^*=O_p(N^{-1/2});
\]
\item $h(x;\beta)$ is twice continuously differentiable in $\beta$ in a
neighborhood of $\beta^*$, with first and second derivatives uniformly
bounded over the realized values of $x_g$ along the sequence.
\end{enumerate}
\end{assumption}

Assumption~\ref{ass:adjustment} accommodates both interacted and
non-interacted regression adjustment. The separate treatment- and
control-arm regressions described above correspond to one particular
parameterization of $h(x_g;\widehat\beta)$. 

We can now state a common asymptotic result for this estimator class, allowing both weighting and estimated adjustment.

\begin{proposition}[General design-based central limit theorem]
\label{thm:general-clt}
Suppose Assumptions~\ref{ass:bounded-cells}, \ref{ass:cells}, and
\ref{ass:adjustment} hold, with Assumption~\ref{ass:cells} evaluated at
$h=h^*$. Then
\[
\sqrt N
\{\widehat\tau(w,\widehat h)-\tau_{w,\cS}\}
\Rightarrow
\mathcal N\{0,V(w,h^*)\}.
\]
\end{proposition}

The proposition covers all four estimators introduced above. The baseline and adjusted estimators are consistent and asymptotically normal for the SATE. This includes the usual difference in means of candidate-sum outcomes between treated and control experimental units, which is obtained with equal weights and intercept-only treatment- and control-arm regressions for the feasible adjustment in Section~\ref{sec:adjustment}. In particular, the finite-sample departure from exact unbiasedness induced by estimating the adjustment does not affect consistency or the first-order limiting distribution. The weighted and weighted-and-adjusted estimators are consistent and asymptotically normal for the weighted target in \eqref{eq:weighted-target}. Under homogeneous treatment effects this target coincides with the SATE.

The result does not depend on the degree of experimental overlap within linkage cells or on the limiting behavior of \(N/M\). In particular, it applies both when \(N_g^{\mathrm{cell}}\leq1\) for every represented linkage cell and when multiple experimental units may share a linkage cell.

\subsection{Adjustment and precision weighting}
\label{sec:efficiency}

\subsubsection{Optimal regression adjustment based on $x_g$}

We first characterize the variance-minimizing adjustment for a given fixed weighting rule. Suppose $x_g$ has finite support and let $w_g>0$ denote fixed pre-treatment weights. For each value $x$, define
\[
W_{2,x}
=
\sum_{i\in\cS:x_{g(i)}=x} w_{g(i)}^2
=
\sum_{g:x_g=x} w_g^2 N_g^{\mathrm{cell}}.
\]
For $d\in\{0,1\}$, define the finite-population weighted arm-specific mean
\begin{equation}
\mu_{d,w,N}(x)
=
\frac{1}{W_{2,x}}
\sum_{i\in\cS:x_{g(i)}=x}
w_{g(i)}^2
\E_D\!\left[
\widetilde Y_{g(i)}(\mathbf D)\mid D_i=d
\right].
\label{eq:oracle-arm-mean-weighted}
\end{equation}
The dependence on $N$ emphasizes that these are finite-population quantities and need not converge to fixed limits.

Define
\begin{equation}
h_{w,N}^*(x)
=
(1-p)\mu_{1,w,N}(x)
+
p\mu_{0,w,N}(x).
\label{eq:optimal-h-weighted}
\end{equation}

\begin{proposition}[Optimal adjustment based on $x_g$]
\label{prop:optimal-adjustment}
Fix any positive pre-treatment weighting rule $\{w_g\}$. Among all fixed adjustments of the form $h_g=h(x_g)$, the choice
\[
h_g=h_{w,N}^*(x_g)
\]
minimizes the exact design variance of $\widehat\tau(w,h)$.
\end{proposition}

The oracle adjustment in \eqref{eq:optimal-h-weighted} is generally infeasible because the conditional expectations in \eqref{eq:oracle-arm-mean-weighted} are not observed. A natural saturated sample analogue is obtained by estimating the weighted arm-specific means separately within each category of $x_g$. Define
\begin{equation}
\widehat\mu_{d,w}(x)
=
\frac{
\sum_{i\in\cS}
w_{g(i)}^2
\1\{x_{g(i)}=x,D_i=d\}
\widetilde Y_{g(i)}(\mathbf D)
}{
\sum_{i\in\cS}
w_{g(i)}^2
\1\{x_{g(i)}=x,D_i=d\}
},
\qquad d\in\{0,1\},
\label{eq:muhat-saturated-weighted}
\end{equation}
whenever the denominator is positive, and set
\begin{equation}
\widehat h_w(x)
=
(1-p)\widehat\mu_{1,w}(x)
+
p\widehat\mu_{0,w}(x).
\label{eq:hhat-saturated-weighted}
\end{equation}
Equivalently, this adjustment can be obtained from separate saturated treatment- and control-arm regressions using $w_g^2$ as regression weights.

The following result verifies that this feasible saturated adjustment satisfies Assumption~\ref{ass:adjustment}.

\begin{corollary}[Feasible saturated adjustment]
\label{cor:saturated-adjustment}
Suppose Assumption~\ref{ass:bounded-cells} holds and the weight bounds and
moment condition in Assumption~\ref{ass:cells}(i)--(ii) hold with $h=0$.
Suppose $x_g$ has fixed finite support and, for every retained category $x$,
\begin{equation}
\frac{N_x}{N}\longrightarrow\pi_x>0,
\qquad
N_x=\sum_{i\in\cS}\1\{x_{g(i)}=x\}.
\label{eq:positive-category-share}
\end{equation}
Then
\[
\widehat\mu_{d,w}(x)-\mu_{d,w,N}(x)
=
O_p(N^{-1/2}),
\qquad d\in\{0,1\},
\]
for every retained category $x$. Hence
\[
\widehat h_w(x)-h_{w,N}^*(x)
=
O_p(N^{-1/2}),
\]
and the saturated adjustment satisfies Assumption~\ref{ass:adjustment} with finite-population target $h_{w,N}^*$. Also, Assumption~\ref{ass:cells}(ii) holds with
$h_g=h_{w,N}^*(x_g)$.
\end{corollary}

Condition~\eqref{eq:positive-category-share} also implies that each retained category is represented in both treatment arms with probability approaching one under Bernoulli assignment. Categories that are too sparse for a saturated specification can instead be pooled or modeled more parsimoniously.

Combining Proposition~\ref{prop:optimal-adjustment}, Corollary~\ref{cor:saturated-adjustment}, and Proposition~\ref{thm:general-clt}, the feasible saturated adjustment has the same first-order asymptotic distribution as the infeasible variance-minimizing adjustment $h_{w,N}^*$. Thus, for any given fixed weighting rule, saturated regression adjustment weakly reduces first-order asymptotic variance relative to using the same weights without adjustment.

When the weights are themselves functions only of $x_g$, $w_g^2$ is constant within each category and cancels from \eqref{eq:oracle-arm-mean-weighted} and \eqref{eq:muhat-saturated-weighted}. In that case, the oracle and feasible adjustments reduce to the corresponding unweighted arm-specific means within each category.

\subsubsection{Optimal weights}

Weighting addresses a different source of imprecision. Suppose $x_g$ has finite support and, for each value $x$ in its support,
\begin{equation}
\frac{N_x}{N}
\longrightarrow
\pi_x>0.
\label{eq:Nx-pix}
\end{equation}

For a given fixed adjustment $h$, define the asymptotic design-variance contribution per experimental unit in category $x$ by
\begin{equation}
\overline V_h(x)
=
\lim
\frac{1}{N_x}
\sum_{g\in\mathcal G_N:x_g=x}V_g(h),
\label{eq:Vbarx}
\end{equation}
whenever the limit exists and lies in $(0,\infty)$.

\begin{proposition}[First-order precision weighting]
\label{prop:efficiency}
Suppose the conditions of Proposition~\ref{thm:general-clt} hold and the limits in \eqref{eq:Nx-pix}--\eqref{eq:Vbarx} exist. For a given fixed adjustment $h$, among positive weights that are functions of $x_g$, asymptotic variance is minimized by
\begin{equation}
w^*(x)
\propto
\frac{1}{\overline V_h(x)}.
\label{eq:optimal-weight}
\end{equation}

In particular, let $x_g=M_g$ and suppose that there exists a constant $c>0$
such that
\begin{equation}
\frac{1}{N_m}
\sum_{g\in\mathcal G_N:M_g=m}V_g(h)
\longrightarrow
cm
\label{eq:homoskedastic-benchmark}
\end{equation}
for each $m$ in the support of $M_g$, where
\[
N_m=\sum_{g:M_g=m}N_g^{\mathrm{cell}}.
\]
Then
\[
w^*(m)\propto\frac{1}{m}.
\]
\end{proposition}

Appendix~\ref{app:weighting-primitives} provides primitive stochastic conditions under which the asymptotic proportionality condition in \eqref{eq:homoskedastic-benchmark} holds almost surely. While those conditions are relatively strong, they illustrate the intuition that linkage cells with larger \(M_g\) may contribute more to the design-based variance, providing a rationale for downweighting such cells.

Under homogeneous treatment effects, reweighting does not
change the causal target, so the variance gains described
above correspond to efficiency improvements for the same
target. With heterogeneous treatment effects, reweighting
may involve a trade-off: downweighting cells to improve
precision also changes their contribution to the causal
target, yielding a weighted average of treatment effects
that need not coincide with the SATE.

\subsection{Feasible asymptotic variance estimation}
\label{sec:variance-estimation}

Proposition~\ref{thm:general-clt} establishes root-\(N\) asymptotic normality of the estimators considered above. Therefore, a feasible estimator of the asymptotic variance, or an asymptotically conservative upper bound, is sufficient for asymptotic inference on \(\tau_{w,\cS}\).

Because experimental units that share a linkage cell also share the same candidate-sum outcome, the natural approach is to cluster the variance calculation at the linkage-cell level. Let
\begin{equation}
\widehat H_g
=
w_g\frac{T_g^{\mathrm{cell}}-pN_g^{\mathrm{cell}}}{p(1-p)}
\{\widetilde Y_g(\mathbf D)-\widehat h_g\},
\label{eq:Hhat}
\end{equation}
where $\widehat h_g=0$ for estimators without adjustment.
Let $\widehat\delta_g$ be a centering term for the treatment effect
in cell $g$. A simple choice is $\widehat\delta_g=\widehat\tau$
for every cell; cell-specific values can instead account for
predictable variation in treatment effects. Define
\begin{equation}
\widehat V_{\mathrm{cl}}
=
\frac{N}{W_{\cS}^2}
\sum_{g\in\mathcal G_N}
\left(\widehat H_g-w_gN_g^{\mathrm{cell}}\widehat\delta_g\right)^2.
\label{eq:cluster-vhat}
\end{equation}

For fixed centering targets $\delta_g^*$, define
\begin{equation}
\Delta_N(w,\delta^*)
=
\frac{N}{W_{\cS}^2}
\sum_{g\in\mathcal G_N}
w_g^2
\{\theta_g-N_g^{\mathrm{cell}}\delta_g^*\}^2.
\label{eq:cluster-gap-N}
\end{equation}

\begin{proposition}[Cluster variance and treatment-effect heterogeneity]
\label{prop:cluster-se}
Suppose the conditions of Proposition~\ref{thm:general-clt} hold.
Suppose there exist fixed targets $\delta_g^*$ such that
\[
\frac1N\sum_{g\in\mathcal G_N}|\delta_g^*|^4=O(1),
\qquad
\frac1N\sum_{g\in\mathcal G_N}
(\widehat\delta_g-\delta_g^*)^2=o_p(1),
\]
and $\Delta_N(w,\delta^*)\longrightarrow\Delta(w,\delta^*)$.
Then
\begin{equation}
\widehat V_{\mathrm{cl}}
\overset{p}{\longrightarrow}
V(w,h^*)+\Delta(w,\delta^*),
\label{eq:cluster-limit}
\end{equation}
where $\Delta(w,\delta^*)\geq0$. Thus
$\widehat V_{\mathrm{cl}}$ is asymptotically conservative for the
design variance and is consistent whenever $\Delta(w,\delta^*)=0$.
Under homogeneous treatment effects, the gap vanishes if
$\delta_g^*$ equals the common effect for every represented cell.
\end{proposition}

The common centering $\widehat\delta_g=\widehat\tau$ satisfies these
conditions with $\delta_g^*=\tau_{w,\cS}$, provided the corresponding
gap has a limit. It therefore gives a valid variance estimator for
all four estimators considered above, including the baseline and
weighted estimators with $\widehat h_g=0$, and is consistent under
homogeneous treatment effects. More generally, the gap reflects
variation in linkage-cell average treatment effects around
$\delta_g^*$ and vanishes whenever
$\theta_g=N_g^{\mathrm{cell}}\delta_g^*$ for every represented cell.

For the WLS implementation in \eqref{eq:interacted-wls}, take
\[
\widehat\delta_g=x_g'(\widehat\beta_1-\widehat\beta_0).
\]
Let $\beta_d^*$ denote the finite-population weighted projection
coefficients of $\E_D[\widetilde Y_{g(i)}(\mathbf D)\mid D_i=d]$
on $x_{g(i)}$ over all $i\in\cS$, using weights $w_{g(i)}$.
The corresponding targets are
$h_g^*=(1-p)x_g'\beta_1^*+p x_g'\beta_0^*$ and
$\delta_g^*=x_g'(\beta_1^*-\beta_0^*)$.
Suppose the conditions of Proposition~\ref{thm:general-clt}
hold, with Assumption~\ref{ass:adjustment} applied jointly to the
coefficients of the two arm-specific regressions and their
finite-population targets defined above. In particular,
$\widehat\beta_d-\beta_d^*=O_p(N^{-1/2})$ for $d\in\{0,1\}$.
Suppose also that $x_g$ includes an intercept, has fixed dimension,
and is uniformly bounded, and
\[
\mathsf A_N
=\frac1N\sum_{g\in\mathcal G_N}
w_gN_g^{\mathrm{cell}}x_gx_g'
\longrightarrow\mathsf A\succ0.
\]
Then this centering satisfies the moment and approximation conditions
of Proposition~\ref{prop:cluster-se}. Moreover,
\[
N\widehat{\operatorname{Var}}_{\mathrm{CR0}}(\widehat t)
-\widehat V_{\mathrm{cl}}=o_p(1),
\]
where $\widehat{\operatorname{Var}}_{\mathrm{CR0}}(\widehat t)$ is the
usual cluster-robust variance estimator for the treatment coefficient
in \eqref{eq:interacted-wls}, clustering by linkage cell and without
a degrees-of-freedom correction. Thus, when the corresponding gap
converges, the regression CRVE is also asymptotically conservative.
With an intercept, homogeneous treatment effects imply
$\delta_g^*=\tau_{w,\cS}$, so both variance estimators are consistent
on their respective scales.

This equivalence provides a direct implementation:
run the regression in \eqref{eq:interacted-wls} at the experimental-unit
level and use the standard error clustered by linkage cell.
The result covers both intercept-only WLS adjustment and the fully
interacted adjustment with additional covariates. In the first case,
the regression includes an intercept and treatment; in the second,
it also includes the centered covariates and their interactions with
treatment. 
When $N_g^{\mathrm{cell}}\leq1$ for every represented linkage cell,
CR0 coincides with the individual-level heteroskedasticity-robust
estimator HC0 applied to the same regression. When a cell contains
multiple experimental units, clustering accounts for the dependence
induced by their shared candidate-sum outcome.
Usual degrees-of-freedom corrections that tend to one, such as CR1,
leave these asymptotic conclusions unchanged.

Proposition~\ref{prop:cluster-se} conditions on the realized
experimental sample and therefore concerns inference for the SATE,
or for the corresponding weighted treatment-effect target. If the
experimental sample is itself treated as random and the target is
an average treatment effect in a broader population or subpopulation,
sampling variation would provide an additional source of uncertainty.

\section{Other assignment mechanisms}
\label{sec:other-designs}

Many experiments fix the number of treated units, either overall or within randomization strata. Appendix~\ref{app:fixed-count} extends our finite-sample and asymptotic results to these designs.

We show that a generalization of the baseline estimator allowing treatment probabilities to vary across strata, denoted $\widehat\tau^{\mathrm{SR}}$, can be biased for the SATE when experimental units share both a linkage cell and a randomization stratum.\footnote{When treatment probabilities are constant across strata, $\widehat\tau^{\mathrm{SR}}$ coincides numerically with $\widehat\tau^{\mathrm B}$.} Under fixed treatment counts, treating one unit reduces the probability that another unit in the same stratum is treated. When these units share a linkage cell, this dependence prevents the other unit's treatment-effect contribution from canceling in expectation. We derive a corrected estimator, $\widehat\tau^{\mathrm{SR,corr}}$, that is exactly unbiased provided no stratum is entirely contained in one linkage cell.  Fisher randomization inference for the sharp null of no individual treatment effects remains exact when assignments follow the actual design and the full statistic is recomputed after each reassignment.

We also study two asymptotic regimes. With a fixed number of growing strata and bounded linkage-cell sizes, the bias of $\widehat\tau^{\mathrm{SR}}$ vanishes and the corrected and uncorrected estimators are first-order asymptotically equivalent. With a growing number of small strata, the uncorrected estimator can remain inconsistent for the SATE, whereas the corrected estimator is consistent and asymptotically normal under the stated moment and variance conditions. 
For each regime, we provide a nonnegative variance estimator that is asymptotically conservative in general and consistent under homogeneous treatment effects. With many small strata, linkage cells crossing stratum boundaries can generate first-order cross-stratum covariances, which our variance estimator explicitly accounts for.

Regression adjustment and weighting extend to these designs, subject to the same caveats discussed under Bernoulli assignment.

\section{Discussion}
\label{sec:discussion}

Ambiguous record linkage is usually treated as a problem of recovering individual identities or individual outcomes. For treatment-effect estimation in a randomized experiment, that objective can be unnecessarily demanding. We show that the SATE can be recovered under much weaker conditions on the linkage process: the true outcome must be retained among the candidate outcomes, but the relative probabilities of candidate identities can remain unknown and the ambiguity can be arbitrarily related to potential outcomes.

The baseline estimator trades an identification problem for an efficiency problem. Retaining all candidate outcomes avoids attenuation from selecting or averaging ambiguous links, but can add substantial noise. The linkage structure itself provides useful information about that noise. Regression adjustment can remove predictable level differences while retaining the SATE as the first-order target. Precision-motivated weighting can further reduce variance, but generally changes the target from the SATE to a weighted average when treatment effects are heterogeneous. Together, these approaches characterize how precision can be recovered without returning to a probabilistic model of individual identity and when precision gains come at the cost of changing the estimand.

Several extensions are natural. We focus on candidate sets generated by a coarse identifier and therefore forming a partition. More general overlapping candidate sets can be represented by a known incidence matrix, suggesting extensions of the same design-based logic beyond partitions. A second direction is to study whether candidate-sum outcomes can be used with nonexperimental methods, such as matching, instrumental variables, difference-in-differences, or regression discontinuity.

\appendix
\numberwithin{table}{section}

\begin{center}
{\Large\bfseries Online Appendix}
\end{center}

\section{Fixed-count and stratified randomization: additional results}
\label{app:fixed-count}

Fixing treatment counts changes both identification and inference with candidate-sum outcomes. First, treatment assignments are no longer independent within a randomization stratum. If experimental units share both a stratum and a linkage cell, the other participants' treatment-effect contributions need not cancel from the expected treatment contrast. Second, even though strata are randomized independently, linkage cells crossing strata can make the estimator's contributions dependent across strata.

We first characterize the resulting bias, remove it with a correction based only on experimental counts, and identify the covariance terms entering the finite-sample variance. We then develop inference under two asymptotic regimes. The distinction matters because the correction becomes negligible with a fixed number of growing strata but can remain essential with many small strata. Throughout, linkage cells may cross randomization strata. 

For ease of exposition, the main development below considers the estimators without weighting or adjustment. Appendix~\ref{app:strat-general-results} states the corresponding assumptions and results when we consider weighting and regression adjustment. Appendix~\ref{app:proofs} provides the proofs of these general results. The unweighted, unadjusted results follow as the special case $w_g=1$ and $h_g=0$.

\subsection{Finite-sample results}

Partition the experimental sample into strata $\cS_1,\ldots,\cS_R$. Stratum $r$ has $N_r^{\mathrm{str}}=|\cS_r|$ units, of which exactly $T_r^{\mathrm{str}}$ are treated. 
Every subset of this size is equally likely, and assignments in distinct strata are independent. Write $p_r=T_r^{\mathrm{str}}/N_r^{\mathrm{str}}$ and assume $1\le T_r^{\mathrm{str}}\le N_r^{\mathrm{str}}-1$. Complete randomization is the special case $R=1$.

The relevant overlap is the intersection of a linkage cell and a randomization stratum. For each intersection, write
\begin{equation}
N_{gr}=|\cC_g\cap\cS_r|,\qquad
T_{gr}=\sum_{i\in\cC_g\cap\cS_r}D_i,\qquad
\theta_{gr}=\sum_{i\in\cC_g\cap\cS_r}\tau_i.
\label{eq:strata-cell-defs}
\end{equation}
Thus $\sum_rN_{gr}=N_g^{\mathrm{cell}}$, $\sum_gN_{gr}=N_r^{\mathrm{str}}$, and $\sum_rT_{gr}=T_g^{\mathrm{cell}}$. These counts concern experimental participants rather than all administrative candidates. Unless stated otherwise, sums over intersection contributions retain only nonempty intersections.

\subsubsection{Why fixed treatment counts create bias, and how to remove it}

Allowing the assignment probability to vary across strata gives the analogue of the baseline estimator. Define the uncorrected intersection contribution
\begin{equation}
H_{gr}
=
\frac{T_{gr}-p_rN_{gr}}{p_r(1-p_r)}
\widetilde Y_g(\mathbf D),
\qquad
\widehat\tau^{\mathrm{SR}}
=\frac1N\sum_{g,r}H_{gr}.
\label{eq:strat-naive}
\end{equation}

When $p_r=p$ for all $r$, this equals $\widehat\tau^{\mathrm B}$ numerically. Under Bernoulli assignment, the contribution of another experimental participant to \(\widetilde Y_g(\mathbf D)\) is orthogonal to the recipient's centered treatment indicator because their treatment assignments are independent. Under fixed treatment counts, treating one participant lowers the probability of treating another participant in the same stratum. This matters for the expected contrast only when the two participants also share a candidate sum.

\begin{proposition}[Bias under stratified fixed-count randomization]
\label{prop:strat-main}
Under the design above,
\begin{equation}
\E_D[\widehat\tau^{\mathrm{SR}}]
=\tau_{\cS}-\frac1N\sum_{g,r}
\frac{N_{gr}-1}{N_r^{\mathrm{str}}-1}\theta_{gr}.
\label{eq:strat-bias-main}
\end{equation}
\end{proposition}

The bias depends on experimental overlap within linkage-cell--stratum intersections, not on the number of nonexperimental candidates. Intersections with $N_{gr}=1$ contribute no bias. More generally, the expected contribution from an intersection is its effect total multiplied by $(N_r^{\mathrm{str}}-N_{gr})/(N_r^{\mathrm{str}}-1)$. With heterogeneous effects of either sign, the aggregate bias need not be toward zero. The mechanism requires both overlaps: fixed treatment counts create assignment dependence only within a stratum, and the candidate sum transmits one unit's assignment to another unit's contribution only when they also share a linkage cell.

To remove this bias, define
\begin{equation}
\begin{aligned}
\widetilde H_{gr}
=\frac{N_r^{\mathrm{str}}-1}{N_r^{\mathrm{str}}-N_{gr}}
\frac{T_{gr}-p_rN_{gr}}{p_r(1-p_r)}
\widetilde Y_g(\mathbf D), \qquad
\widehat\tau^{\mathrm{SR,corr}}
&=\frac1N\sum_{g,r}\widetilde H_{gr}.
\end{aligned}
\label{eq:strat-corrected-main}
\end{equation}
The correction therefore divides each uncorrected contribution \(H_{gr}\) by the known factor multiplying \(\theta_{gr}\) in its expectation.  It equals one when $N_{gr}=1$ and is well defined provided
\begin{equation}
N_{gr}<N_r^{\mathrm{str}}\quad\text{whenever }N_{gr}>0.
\label{eq:strat-nonempty}
\end{equation}
This restriction excludes a stratum entirely contained in a single linkage cell. In that case $T_{gr}=T_r^{\mathrm{str}}$ is fixed, so the assignment multiplier for that intersection is identically zero and cannot be rescaled to recover its effect total. Set $\widetilde H_{gr}=\theta_{gr}=0$ for empty intersections when they appear in algebraic expressions.

\begin{proposition}[Exact correction]
\label{prop:strat-correction}
Under \eqref{eq:strat-nonempty},
\[
\E_D\widetilde H_{gr}=\theta_{gr},\qquad
\E_D[\widehat\tau^{\mathrm{SR,corr}}]=\tau_{\cS}.
\]
\end{proposition}

The correction restores exact unbiasedness without resolving any individual link. It uses only the stratum sizes, treatment counts, experimental linkage-cell memberships, and observed candidate sums. In particular, it does not require observing the full administrative linkage-cell size $M_g$. Consequently, the correction remains available when zero-outcome administrative records are omitted, provided all nonzero outcomes are retained and absence from the file implies a zero outcome.

\subsubsection{Which covariances enter the variance?}

Unlike the Bernoulli case, the variance is not generally a sum of independent linkage-cell variances. Contributions can share the dependence induced by fixed treatment counts within a stratum, share a candidate sum within a linkage cell, or share both sources indirectly when linkage cells cross strata. Write
\begin{equation}
V_{gr}=\Var_D(\widetilde H_{gr}),\qquad
C_{gr,hs}=\Cov_D(\widetilde H_{gr},\widetilde H_{hs}),
\qquad C_{gr,gr}=V_{gr}.
\label{eq:strat-variance-notation}
\end{equation}
All covariance sums below use ordered pairs.

For contributions from distinct strata $r$ and $s$, a potentially nonzero covariance requires both reciprocal intersections to be nonempty: linkage cell $g$ must also contain participants in stratum $s$, and linkage cell $h$ must also contain participants in stratum $r$. If either intersection is empty, the covariance vanishes. The next result uses this restriction to separate the variance into four interpretable components.

\begin{proposition}[Finite-sample variance and covariance support]
\label{prop:strat-covariance}
If $r\ne s$ and $N_{gs}N_{hr}=0$, then $C_{gr,hs}=0$. Consequently,
\begin{equation}
\begin{aligned}
\Var_D(\widehat\tau^{\mathrm{SR,corr}})=\frac1{N^2}\Bigg[
&\sum_{g,r}V_{gr}+\sum_r\sum_{g\ne h}C_{gr,hr}\\
&+\sum_g\sum_{r\ne s}C_{gr,gs}
+\sum_{\substack{g\ne h,\ r\ne s\\N_{gs}N_{hr}>0}}C_{gr,hs}\Bigg].
\end{aligned}
\label{eq:strat-finite-variance}
\end{equation}
\end{proposition}

The first component collects intersection variances, while the second and third correspond to the two natural clustering dimensions---randomization strata and linkage cells. This resembles two-way clustering along these dimensions \citep{camerongelbachmiller2011,thompson2011}, which allows nonzero covariance between contributions that share either a randomization stratum or a linkage cell. Here, however, crossing linkage cells and strata can also generate dependence between contributions that share neither dimension. The fourth component collects exactly these additional covariance terms, which arise for distinct linkage cells and distinct strata with reciprocal intersections.

\subsection{Asymptotic analysis}

We consider two asymptotic regimes with $N\to\infty$, retaining the finite-population interpretation of the main text. With a fixed number of growing strata, each bounded linkage cell occupies a vanishing fraction of its strata, so the bias correction approaches one. With a growing number of bounded-size strata, this fraction need not vanish and the correction can remain important. We treat these regimes separately because they also require different ways of estimating the variance.

For root-$N$ inference, define
\[
V_N^{\mathrm{SR}}=N\Var_D(\widehat\tau^{\mathrm{SR,corr}}),
\]
and let $V_N^{(1)},\ldots,V_N^{(4)}$ denote the four sums inside the brackets in \eqref{eq:strat-finite-variance}, each divided by $N$. Thus $V_N^{\mathrm{SR}}=\sum_{j=1}^4V_N^{(j)}$. For a nondegenerate normal limit, we additionally impose the same type of variance-limit condition used under Bernoulli assignment.

\begin{assumption}[Stratified variance limit]
\label{ass:strat-variance}
The normalized design variance satisfies
\begin{equation}
V_N^{\mathrm{SR}}\longrightarrow V^{\mathrm{SR}}\in(0,\infty).
\label{eq:strat-variance-limit}
\end{equation}
\end{assumption}

This condition requires the combined variance and covariance contributions to stabilize at a positive value. The size and moment conditions imposed below control the magnitude of this normalized variance. However, they do not by themselves ensure that it converges or stays away from zero.

\subsubsection{A fixed number of growing strata}

We first consider sequences with a fixed number of strata, each containing an increasing number of participants in both treatment arms. Linkage cells remain bounded. In this regime, the dependence between any fixed set of participants within a stratum weakens as the stratum grows, although its aggregate effect on the variance need not disappear.

\begin{assumption}[Growing strata]
\label{ass:strat-large}
Assumption~\ref{ass:bounded-cells} and \eqref{eq:strat-nonempty} hold. The number $R$ is fixed, $N_r^{\mathrm{str}}/N\to\lambda_r>0$, and $p_r\to p_r^\infty\in(0,1)$. Along the sequence, the corrected contributions satisfy
\begin{equation}
\frac1N\sum_{g,r:N_{gr}>0}
\E_D\bigl[|\widetilde H_{gr}|^4\bigr]=O(1).
\label{eq:strat-moment}
\end{equation}
\end{assumption}

The moment condition is stated directly in terms of the estimator's contributions. Because it bounds their noncentered fourth moments on average, it also controls their means and centered moments on average, without requiring uniformly bounded potential outcomes or uniformly bounded moments in every intersection.

\paragraph{Asymptotic distribution.}
We first consider the asymptotic distributions of $\widehat\tau^{\mathrm{SR,corr}}$ and $\widehat\tau^{\mathrm{SR}}$ under this framework.

\begin{proposition}[Growing strata: equivalence and normality]
\label{prop:fixedcount-asymptotics}
Under Assumption~\ref{ass:strat-large},
\begin{equation}
\sqrt N\{\widehat\tau^{\mathrm{SR,corr}}-\widehat\tau^{\mathrm{SR}}\}=o_p(1).
\label{eq:corr-first-order}
\end{equation}
If Assumption~\ref{ass:strat-variance} also holds, both estimators satisfy
\[
\sqrt N(\widehat\tau-\tau_{\cS})\Rightarrow\mathcal N(0,V^{\mathrm{SR}}).
\]
\end{proposition}

The estimators are asymptotically equivalent because, under this asymptotic framework, the correction factor satisfies
\[
\frac{N_r^{\mathrm{str}}-1}{N_r^{\mathrm{str}}-N_{gr}}
=1+O(N^{-1})
\]
uniformly over nonempty intersections. Under the stated moment condition, its effect on the estimator is therefore negligible at the root-$N$ scale. Relative to $\widehat\tau^{\mathrm{SR}}$ under independent Bernoulli assignment with the same probabilities $p_r$, fixing the treatment counts removes a nonnegative projection term from the first-order variance. Thus, fixing treatment counts cannot increase the variance at first order.

\paragraph{Variance estimation.}
We consider the following estimator of the asymptotic variance \(V^{\mathrm{SR}}\):
\begin{equation}
\widehat V_a=\frac1N\sum_g
\left[
\widetilde H_g-\sum_rN_{gr}\widehat\tau_r
-\sum_r\widehat b_r(T_{gr}-p_rN_{gr})
\right]^2.
\label{eq:strat-vhat-a}
\end{equation}
where
\begin{equation}
\widetilde H_g=\sum_r\widetilde H_{gr},\qquad
\widehat\tau_r=\frac1{N_r^{\mathrm{str}}}\sum_g\widetilde H_{gr},
\label{eq:strat-large-centers}
\end{equation}
and
\begin{equation}
\widehat b_r=
\frac{\sum_g(T_{gr}-p_rN_{gr})\widetilde H_g}
{N_r^{\mathrm{str}}p_r(1-p_r)}.
\label{eq:strat-large-projection-coefficient}
\end{equation}

\begin{proposition}[Variance estimation with growing strata]
\label{prop:strat-variance-a}
Under Assumptions~\ref{ass:strat-large} and \ref{ass:strat-variance}, let
$\tau_{\cS_r}=(N_r^{\mathrm{str}})^{-1}\sum_g\theta_{gr}$. Then
\begin{equation}
\widehat V_a=V_N^{\mathrm{SR}}+
\frac1N\sum_g\left(\theta_g-\sum_rN_{gr}\tau_{\cS_r}\right)^2+o_p(1).
\label{eq:strat-vhat-a-limit}
\end{equation}
Thus $\widehat V_a$ is nonnegative in every sample, asymptotically conservative, and consistent for $V^{\mathrm{SR}}$ under homogeneous treatment effects.
\end{proposition}

The nonnegative gap compares each linkage cell's effect total with the total predicted by its stratum composition and the stratum average effects. It vanishes under homogeneous effects, and also when effects are constant within each stratum even if they differ across strata.

With perfect linkage, this construction reduces asymptotically to the familiar stratified Neyman variance estimate. In fact, $\widehat V_a$ is then exactly
\[
\sum_r\frac{N_r^{\mathrm{str}}}{N}
\left\{\frac{\widehat\sigma_{1r}^2}{p_r}
+\frac{\widehat\sigma_{0r}^2}{1-p_r}\right\},
\]
where $\widehat\sigma_{dr}^2$ is the empirical variance of observed outcomes in arm $d$ of stratum $r$, calculated with the arm size as divisor. 

\subsubsection{A growing number of small strata}
\label{app:many-small-strata}

We next consider sequences in which the number of strata grows while stratum sizes remain bounded. Fixed treatment counts then continue to induce substantial dependence within a stratum. The bias correction need not approach one, and candidate sums can transmit assignment dependence across strata when linkage cells cross their boundaries. The key restriction is that both strata and linkage cells have bounded size.

\begin{assumption}[Small strata and bounded contribution moments]
\label{ass:strat-small}
Assumption~\ref{ass:bounded-cells} and \eqref{eq:strat-nonempty} hold. There is a fixed $\overline N^{\mathrm{str}}<\infty$ such that $2\le N_r^{\mathrm{str}}\le\overline N^{\mathrm{str}}$, and
\begin{equation}
\frac1N\sum_{g,r:N_{gr}>0}\E_D|\widetilde H_{gr}|^4=O(1).
\label{eq:strat-small-moments}
\end{equation}
\end{assumption}

The number of strata and the number of nonempty linkage-cell--stratum intersections, $\#\{(g,r):N_{gr}>0\}$, are both of order $N$. The noncentered moment bound controls both the means and centered moments of the contributions on average. Proposition~\ref{prop:strat-covariance} leaves at most $\overline N^{\mathrm{str}}\overline M$ potentially nonzero covariances per intersection, including its own variance. Thus only order $N$ entries of an order-$N^2$ covariance matrix can be nonzero. The average moment bound and Cauchy--Schwarz bound their total absolute magnitude by $O(N)$, implying $\Var_D(\widehat\tau^{\mathrm{SR,corr}})=O(N^{-1})$. The surviving covariances can nevertheless contribute in first order to $V_N^{\mathrm{SR}}$.

\paragraph{Asymptotic distribution.}
We next consider the asymptotic distributions of $\widehat\tau^{\mathrm{SR,corr}}$ and $\widehat\tau^{\mathrm{SR}}$ under this alternative asymptotic framework.

\begin{proposition}[Small strata: bias, consistency, and normality]
\label{prop:small-strata}
Under Assumption~\ref{ass:strat-small},
\begin{align}
\widehat\tau^{\mathrm{SR}}-\tau_{\cS}
&=-\frac1N\sum_{g,r}\frac{N_{gr}-1}{N_r^{\mathrm{str}}-1}\theta_{gr}
+O_p(N^{-1/2}),\label{eq:small-strata-bias}\\
\widehat\tau^{\mathrm{SR,corr}}-\tau_{\cS}
&=O_p(N^{-1/2}).\label{eq:small-strata-rate}
\end{align}
If Assumption~\ref{ass:strat-variance} also holds,
\begin{equation}
\sqrt N(\widehat\tau^{\mathrm{SR,corr}}-\tau_{\cS})
\Rightarrow\mathcal N(0,V^{\mathrm{SR}}).
\label{eq:small-strata-clt}
\end{equation}
\end{proposition}

In addition to the exact unbiasedness established in Proposition~\ref{prop:strat-correction}, the corrected estimator has root-$N$ convergence and, under the variance-limit condition, is asymptotically normal despite cross-stratum linkage. In contrast, the uncorrected estimator can converge away from the SATE because its average bias term need not vanish, as illustrated in Example~\ref{ex:alt-stratified} of Appendix~\ref{app:alternative-methods}. The correction is essential for covering this regime without additional restrictions on overlap or treatment effects.

\paragraph{Variance estimation.}
We consider the following estimator of the asymptotic variance $V^{\mathrm{SR}}$, constructed from estimates of the four components in \eqref{eq:strat-finite-variance}.

Define the centered contribution
\begin{equation}
\widehat R_{gr}=\widetilde H_{gr}-N_{gr}\widehat\tau^{\mathrm{SR,corr}}.
\label{eq:strat-residual}
\end{equation}
For the cross-stratum components below, call a pair of strata $(r,s)$ \emph{singular} when both strata are covered exclusively by the same two linkage cells, with both linkage cells represented in each stratum.

\emph{Intersection variances.} The first component is estimated by squared residuals,
\begin{equation}
\widehat V^{(1)}=\frac1N\sum_{g,r}\widehat R_{gr}^{\,2}.
\label{eq:strat-vhat-one}
\end{equation}
\emph{Covariances within a stratum.} For different linkage cells in the same stratum, use
\begin{equation}
\widehat V^{(2)}=\frac1N\sum_r\sum_{g\ne h}\widehat R_{gr}\widehat R_{hr}.
\label{eq:strat-vhat-two}
\end{equation}

\emph{Covariances within a linkage cell across strata.} Use
\begin{equation}
\widehat V^{(3)}=\frac1N\sum_g\sum_{r\ne s}
\begin{cases}
\widehat R_{gr}\widehat R_{gs},&(r,s)\text{ singular},\\
\tfrac12\widetilde H_{gr}\widetilde H_{gs},&\text{otherwise}.
\end{cases}
\label{eq:strat-vhat-three}
\end{equation}
\emph{Covariances between distinct linkage cells and strata.} For $g\ne h$ and $r\ne s$ with all four intersections nonempty, define the observable coefficients
\begin{equation}
a_{gr,hs}=\frac{N_{gr}N_{hs}}
{(N_r^{\mathrm{str}}-N_{gr})(N_s^{\mathrm{str}}-N_{hs})},\qquad
b_{gr,hs}=\frac{N_{hr}N_{gs}}
{(N_r^{\mathrm{str}}-N_{hr})(N_s^{\mathrm{str}}-N_{gs})}.
\label{eq:strat-ab}
\end{equation}
They satisfy $0<a_{gr,hs}b_{gr,hs}\le1$, with equality precisely for singular pairs. Use
\begin{equation}
\widehat V^{(4)}=\frac1N\sum_{\substack{g\ne h,\ r\ne s\\N_{gs}N_{hr}>0}}
\begin{cases}
\widehat R_{gr}\widehat R_{hs},&(r,s)\text{ singular},\\[3pt]
\displaystyle\frac{a_{gr,hs}
(\widetilde H_{hr}\widetilde H_{gs}-b_{gr,hs}\widetilde H_{gr}\widetilde H_{hs})}
{1-a_{gr,hs}b_{gr,hs}},&\text{otherwise}.
\end{cases}
\label{eq:strat-vhat-four}
\end{equation}
The different formulas for singular and nonsingular pairs in $\widehat V^{(3)}$ and $\widehat V^{(4)}$ ensure the asymptotic conservativeness of the resulting variance estimator.

Let $\widehat V_{b,0}=\sum_{j=1}^4\widehat V^{(j)}$. To guarantee nonnegativity in finite samples, use
\begin{equation}
\begin{aligned}
\widehat V_b
&=\max\{\widehat V_{b,0},\widehat V^{(1)}+\widehat V^{(2)}\}\\
&=\frac1N\sum_r\left(\sum_g\widehat R_{gr}\right)^2
+\max\{0,\widehat V^{(3)}+\widehat V^{(4)}\}.
\end{aligned}
\label{eq:strat-vhat-b}
\end{equation}

\begin{proposition}[Variance estimation with small strata]
\label{prop:strat-variance-b}
Under Assumption~\ref{ass:strat-small},
\begin{equation}
\widehat V_{b,0}=V_N^{\mathrm{SR}}+\Delta_{b,N}+O_p(N^{-1/2}),
\qquad \Delta_{b,N}\ge0.
\label{eq:strat-vhat-b-limit}
\end{equation}
The gap vanishes under homogeneous treatment effects. The estimator $\widehat V_b$ is nonnegative in every sample and asymptotically conservative. Under homogeneous effects,
\[
\widehat V_b-V_N^{\mathrm{SR}}=O_p(N^{-1/2}),
\]
so it consistently estimates $V^{\mathrm{SR}}$ when Assumption~\ref{ass:strat-variance} holds.
\end{proposition}

The gap $\Delta_{b,N}$ is a nonnegative treatment-effect-heterogeneity term, analogous to the usual conservatism in design-based variance estimation. Its expression is given in the proof. The nonnegative floor may add conservatism under heterogeneous effects but has no first-order effect under homogeneous effects.

\subsection{Inference and intermediate designs}
\label{app:strat-variance}

The two variance estimators are tailored to different asymptotic regimes and need not be valid outside their respective regimes. This distinction matters even with perfect linkage and homogeneous zero treatment effects. We consider two simple examples to illustrate this distinction.

First, consider a single stratum with $N/2$ treated units, where half the units have $Y_i(0)=Y_i(1)=0$ and half have $Y_i(0)=Y_i(1)=1$. Then
\[
V_N^{\mathrm{SR}}=\frac{N}{N-1}\longrightarrow1,
\]
while $\widehat V_b=0$ for every assignment.

Conversely, consider $R$ strata of size two, each with one treated unit and one control. In every stratum, one unit has $Y_i(0)=Y_i(1)=0$ and the other has $Y_i(0)=Y_i(1)=1$. With $N=2R$,
\[
V_N^{\mathrm{SR}}=2,
\]
while $\widehat V_a=0$ for every assignment.

In practice, a design may have neither a small number of large strata nor a large number of uniformly small strata. One possibility in this case is to report the maximum $\max\{\widehat V_a,\widehat V_b\}$. This is asymptotically conservative under either set of stated conditions because it is no smaller than the estimate valid in that regime. However, it need not be consistent under homogeneous effects, since the estimate used outside its regime can retain upward fluctuations. For intermediate sequences satisfying neither set of conditions, the maximum is a pragmatic precaution rather than a procedure with an established coverage guarantee.

\subsection{General results with weighting and adjustment}
\label{app:strat-general-results}

This subsection states the general versions of the preceding results. Let $w_g>0$ be fixed pre-treatment weights and let $h_g=h(x_g)$ be fixed with respect to assignment. For every nonempty intersection define the uncorrected contribution
\begin{equation}
\begin{aligned}
\widetilde H_{gr}^{\,0}(w,h)
&=w_g\frac{T_{gr}-p_rN_{gr}}{p_r(1-p_r)}
\{\widetilde Y_g-h_g\},\\
\widetilde H_{gr}(w,h)
&=\frac{N_r^{\mathrm{str}}-1}{N_r^{\mathrm{str}}-N_{gr}}
\widetilde H_{gr}^{\,0}(w,h).
\end{aligned}
\label{eq:strat-general-contribution}
\end{equation}
The corrected contribution in the second line is defined under \eqref{eq:strat-nonempty}. Write $\widetilde H_{gr}(h)$ for the corrected contribution with $w_g=1$, and set all contributions to zero for empty intersections. With
\begin{equation}
W_{\cS}=\sum_gw_gN_g^{\mathrm{cell}},
\qquad
\tau_{w,\cS}=\frac1{W_{\cS}}\sum_gw_g\theta_g,
\label{eq:strat-general-target}
\end{equation}
the corresponding estimators are
\begin{equation}
\widehat\tau^{\mathrm{SR,corr}}(w,h)
=\frac1{W_{\cS}}\sum_{g,r}\widetilde H_{gr}(w,h),
\qquad
\widehat\tau^{\mathrm{SR}}(w,h)
=\frac1{W_{\cS}}\sum_{g,r}\widetilde H_{gr}^{\,0}(w,h).
\label{eq:strat-general-estimators}
\end{equation}
Thus the estimators and results developed above are obtained as the special case $w_g=1$, $h_g=0$, and $W_{\cS}=N$.

\begin{proposition}[Fixed weights and adjustment: generalization of Propositions~\ref{prop:strat-main} and \ref{prop:strat-correction}]
\label{prop:strat-general-finite}
For every nonempty intersection,
\begin{align}
\E_D\{\widetilde H_{gr}^{\,0}(w,h)\}
&=w_g\frac{N_r^{\mathrm{str}}-N_{gr}}{N_r^{\mathrm{str}}-1}\theta_{gr},\notag\\
\E_D\{\widehat\tau^{\mathrm{SR}}(w,h)\}
&=\tau_{w,\cS}
-\frac1{W_{\cS}}\sum_{g,r}w_g
\frac{N_{gr}-1}{N_r^{\mathrm{str}}-1}\theta_{gr}.
\label{eq:strat-general-bias}
\end{align}
If \eqref{eq:strat-nonempty} also holds, then
\begin{equation}
\E_D\widetilde H_{gr}(w,h)=w_g\theta_{gr},
\qquad
\E_D\{\widehat\tau^{\mathrm{SR,corr}}(w,h)\}=\tau_{w,\cS}.
\label{eq:strat-general-unbiased}
\end{equation}
The fixed adjustment affects neither expectation.
\end{proposition}

For the variance, let
\[
V_{gr}(w,h)=\Var_D\{\widetilde H_{gr}(w,h)\},
\qquad
C_{gr,ks}(w,h)=\Cov_D\{\widetilde H_{gr}(w,h),\widetilde H_{ks}(w,h)\}.
\]

\begin{proposition}[Covariance support: generalization of Proposition~\ref{prop:strat-covariance}]
\label{prop:strat-general-covariance}
If $r\ne s$ and $N_{gs}N_{kr}=0$, then $C_{gr,ks}(w,h)=0$. Hence, with all covariance sums taken over ordered pairs,
\begin{equation}
\begin{aligned}
\Var_D\{\widehat\tau^{\mathrm{SR,corr}}(w,h)\}
=\frac1{W_{\cS}^2}\Bigg[&
\sum_{g,r}V_{gr}(w,h)
+\sum_r\sum_{g\ne k}C_{gr,kr}(w,h)\\
&+\sum_g\sum_{r\ne s}C_{gr,gs}(w,h)
+\sum_{\substack{g\ne k,\ r\ne s\\N_{gs}N_{kr}>0}}C_{gr,ks}(w,h)
\Bigg].
\end{aligned}
\label{eq:strat-general-variance}
\end{equation}
\end{proposition}

For asymptotic statements define
\begin{equation}
V_N^{\mathrm{SR}}(w,h)
=N\Var_D\{\widehat\tau^{\mathrm{SR,corr}}(w,h)\}.
\label{eq:strat-general-VN}
\end{equation}
For all asymptotic results below, assume
\[
0<\underline w\le w_g\le\overline w<\infty,
\qquad W_{\cS}/N\longrightarrow w^\ast\in(0,\infty).
\]

\begin{assumption}[Variance limit: generalization of Assumption~\ref{ass:strat-variance}]
\label{ass:strat-general-variance}
The normalized design variance satisfies
\begin{equation}
V_N^{\mathrm{SR}}(w,h)\longrightarrow
V^{\mathrm{SR}}(w,h)\in(0,\infty).
\label{eq:strat-general-variance-limit}
\end{equation}
\end{assumption}

\begin{assumption}[Growing strata: generalization of Assumption~\ref{ass:strat-large}]
\label{ass:strat-general-large}
Assumption~\ref{ass:bounded-cells} and \eqref{eq:strat-nonempty} hold. The number $R$ is fixed, $N_r^{\mathrm{str}}/N\to\lambda_r>0$, and $p_r\to p_r^\infty\in(0,1)$. Along the sequence,
\begin{equation}
\frac1N\sum_{g,r:N_{gr}>0}
\E_D\bigl[|\widetilde H_{gr}(w,h)|^4\bigr]=O(1).
\label{eq:strat-general-large-moment}
\end{equation}
\end{assumption}

\begin{proposition}[Growing strata: generalization of Proposition~\ref{prop:fixedcount-asymptotics}]
\label{prop:strat-general-large}
Under Assumption~\ref{ass:strat-general-large},
\[
\widehat\tau^{\mathrm{SR,corr}}(w,h)
-\widehat\tau^{\mathrm{SR}}(w,h)=O_p(N^{-1}).
\]
If Assumption~\ref{ass:strat-general-variance} also holds, then, for $q\in\{\mathrm{corr},0\}$,
\begin{equation}
\sqrt N\{\widehat\tau^{\mathrm{SR},q}(w,h)-\tau_{w,\cS}\}
\Rightarrow\mathcal N\{0,V^{\mathrm{SR}}(w,h)\},
\label{eq:strat-general-large-clt}
\end{equation}
where $q=0$ denotes the uncorrected estimator.
\end{proposition}

For variance estimation in this regime, define
\begin{align}
\widetilde H_g(w,h)&=\sum_r\widetilde H_{gr}(w,h),
&\widehat\tau_r(h)&=\frac1{N_r^{\mathrm{str}}}\sum_g\widetilde H_{gr}(h),
\label{eq:strat-general-large-centers}\\
\widehat b_r(w,h)&=
\frac{\sum_g(T_{gr}-p_rN_{gr})\widetilde H_g(w,h)}
{N_r^{\mathrm{str}}p_r(1-p_r)}.
\label{eq:strat-general-large-projection}
\end{align}
Set
\begin{equation}
\widehat V_a(w,h)
=\frac{N}{W_{\cS}^2}\sum_g
\left[
\widetilde H_g(w,h)
-w_g\sum_rN_{gr}\widehat\tau_r(h)
-\sum_r\widehat b_r(w,h)(T_{gr}-p_rN_{gr})
\right]^2.
\label{eq:strat-general-vhat-a}
\end{equation}

\begin{proposition}[Growing-strata variance estimator: generalization of Proposition~\ref{prop:strat-variance-a}]
\label{prop:strat-general-variance-a}
Under Assumptions~\ref{ass:strat-general-large} and \ref{ass:strat-general-variance}, with $\tau_{\cS_r}=(N_r^{\mathrm{str}})^{-1}\sum_g\theta_{gr}$,
\begin{equation}
\widehat V_a(w,h)
=V_N^{\mathrm{SR}}(w,h)
+\frac{N}{W_{\cS}^2}\sum_gw_g^2
\left(\theta_g-\sum_rN_{gr}\tau_{\cS_r}\right)^2
+o_p(1).
\label{eq:strat-general-vhat-a-limit}
\end{equation}
The gap is nonnegative and vanishes when effects are constant within each stratum, including under homogeneous treatment effects. Setting $w_g=1$ and $h_g=0$ gives Proposition~\ref{prop:strat-variance-a}.
\end{proposition}

\begin{assumption}[Small strata: generalization of Assumption~\ref{ass:strat-small}]
\label{ass:strat-general-small}
Assumption~\ref{ass:bounded-cells} and \eqref{eq:strat-nonempty} hold. There is a fixed $\overline N^{\mathrm{str}}<\infty$ such that $2\le N_r^{\mathrm{str}}\le\overline N^{\mathrm{str}}$, and
\begin{equation}
\frac1N\sum_{g,r:N_{gr}>0}
\E_D|\widetilde H_{gr}(w,h)|^4=O(1).
\label{eq:strat-general-small-moments}
\end{equation}
\end{assumption}

\begin{proposition}[Small strata: generalization of Proposition~\ref{prop:small-strata}]
\label{prop:strat-general-small}
Under Assumption~\ref{ass:strat-general-small},
\begin{align}
\widehat\tau^{\mathrm{SR}}(w,h)-\tau_{w,\cS}
&=-\frac1{W_{\cS}}\sum_{g,r}w_g
\frac{N_{gr}-1}{N_r^{\mathrm{str}}-1}\theta_{gr}
+O_p(N^{-1/2}),\notag\\
\widehat\tau^{\mathrm{SR,corr}}(w,h)-\tau_{w,\cS}
&=O_p(N^{-1/2}).
\label{eq:strat-general-small-rate}
\end{align}
If Assumption~\ref{ass:strat-general-variance} also holds, then
\begin{equation}
\sqrt N\{\widehat\tau^{\mathrm{SR,corr}}(w,h)-\tau_{w,\cS}\}
\Rightarrow\mathcal N\{0,V^{\mathrm{SR}}(w,h)\}.
\label{eq:strat-general-small-clt}
\end{equation}
\end{proposition}

The small-strata variance estimator uses
\begin{equation}
\widehat R_{gr}(w,h)
=\widetilde H_{gr}(w,h)
-w_gN_{gr}\widehat\tau^{\mathrm{SR,corr}}(w,h),
\qquad c_N(w)=\frac{N}{W_{\cS}^2}.
\label{eq:strat-general-residual}
\end{equation}
Using the singular-pair definition and the coefficients $a_{gr,ks}$ and $b_{gr,ks}$ in \eqref{eq:strat-ab}, define
\begin{align}
\widehat V^{(1)}(w,h)
&=c_N(w)\sum_{g,r}\widehat R_{gr}(w,h)^2,\notag\\
\widehat V^{(2)}(w,h)
&=c_N(w)\sum_r\sum_{g\ne k}\widehat R_{gr}(w,h)\widehat R_{kr}(w,h),\notag\\
\widehat V^{(3)}(w,h)
&=c_N(w)\sum_g\sum_{r\ne s}
\begin{cases}
\widehat R_{gr}(w,h)\widehat R_{gs}(w,h),&(r,s)\text{ singular},\\
\tfrac12\widetilde H_{gr}(w,h)\widetilde H_{gs}(w,h),&\text{otherwise},
\end{cases}
\label{eq:strat-general-vhat-123}
\end{align}
and
\begin{equation}
\widehat V^{(4)}(w,h)
=c_N(w)\sum_{\substack{g\ne k,\ r\ne s\\N_{gs}N_{kr}>0}}
\begin{cases}
\widehat R_{gr}(w,h)\widehat R_{ks}(w,h),&(r,s)\text{ singular},\\[3pt]
\displaystyle\frac{a_{gr,ks}
\{\widetilde H_{kr}(w,h)\widetilde H_{gs}(w,h)
-b_{gr,ks}\widetilde H_{gr}(w,h)\widetilde H_{ks}(w,h)\}}
{1-a_{gr,ks}b_{gr,ks}},&\text{otherwise}.
\end{cases}
\label{eq:strat-general-vhat-four}
\end{equation}
Let $\widehat V_{b,0}(w,h)=\sum_{j=1}^4\widehat V^{(j)}(w,h)$ and
\begin{equation}
\widehat V_b(w,h)
=\max\{\widehat V_{b,0}(w,h),\widehat V^{(1)}(w,h)+\widehat V^{(2)}(w,h)\}.
\label{eq:strat-general-vhat-b}
\end{equation}

\begin{proposition}[Small-strata variance estimator: generalization of Proposition~\ref{prop:strat-variance-b}]
\label{prop:strat-general-variance-b}
Under Assumption~\ref{ass:strat-general-small},
\begin{equation}
\widehat V_{b,0}(w,h)
=V_N^{\mathrm{SR}}(w,h)+\Delta_{b,N}(w)+O_p(N^{-1/2}),
\label{eq:strat-general-vhat-b-limit}
\end{equation}
where
\begin{equation}
\Delta_{b,N}(w)
=\frac{N}{W_{\cS}^2}\sum_b
\left[
\sum_{r\in b}\sum_gw_g\{\theta_{gr}-N_{gr}\tau_{w,\cS}\}
\right]^2\ge0.
\label{eq:strat-general-gap-b}
\end{equation}
Here each group $b$ collects strata covered exclusively by the same two linkage cells, while every remaining stratum forms its own group. The final estimator $\widehat V_b(w,h)$ is nonnegative and asymptotically conservative. Under homogeneous effects, the gap vanishes and
\[
\widehat V_b(w,h)-V_N^{\mathrm{SR}}(w,h)=O_p(N^{-1/2}).
\]
\end{proposition}

Finally, consider an estimated adjustment $\widehat h_g=h(x_g;\widehat\beta)$ and its fixed limit $h_g^*=h(x_g;\beta^*)$.

\begin{proposition}[Estimated adjustment under stratified randomization]
\label{prop:strat-estimated-adjustment}
Suppose Assumption~\ref{ass:adjustment} holds under the actual stratified design, and that Assumption~\ref{ass:strat-general-large} or \ref{ass:strat-general-small}, as appropriate, holds with $h=h^*$. For $q\in\{\mathrm{corr},0\}$,
\begin{equation}
\widehat\tau^{\mathrm{SR},q}(w,\widehat h)
-\widehat\tau^{\mathrm{SR},q}(w,h^*)=O_p(N^{-1}).
\label{eq:strat-estimated-adjustment-rate}
\end{equation}
Consequently, whenever the fixed-$h^*$ statistic has a root-$N$ limiting distribution, estimated and fixed adjustment have the same limit. Moreover, in the growing-strata regime and the small-strata regime, respectively,
\begin{equation}
\widehat V_a(w,\widehat h)-\widehat V_a(w,h^*)=O_p(N^{-1/2}),
\qquad
\widehat V_b(w,\widehat h)-\widehat V_b(w,h^*)=O_p(N^{-1/2}).
\label{eq:strat-estimated-adjustment-variance}
\end{equation}
These are asymptotic equivalences; because $\widehat h$ depends on assignment, neither exact unbiasedness nor the fixed-$h$ covariance-support identities are asserted for the estimated-adjustment contributions.
\end{proposition}

\section{Proofs}
\label{app:proofs}

\subsection{Proof of Proposition~\ref{prop:baseline}}

\begin{proof}
We establish unbiasedness and the variance formula by calculating the expectation and variance of each linkage cell's contribution. Independence across linkage cells then allows us to add these quantities. For each linkage cell $g$, define its contribution to $N\widehat\tau^{\mathrm B}$ as
\[
H_g
=
\frac{(T_g^{\mathrm{cell}}-pN_g^{\mathrm{cell}})\widetilde Y_g}{p(1-p)}.
\]
Then
\[
N\widehat\tau^{\mathrm B}=\sum_g H_g.
\]

To prove unbiasedness, it is enough to show that $\E_D[H_g]=\theta_g$, since $\sum_g\theta_g=N\tau_{\cS}$. Using
\[
T_g^{\mathrm{cell}}-pN_g^{\mathrm{cell}}
=
\sum_{i\in\cC_g\cap\cS}(D_i-p)
\]
and
\[
\widetilde Y_g
=
B_g+p\theta_g
+
\sum_{j\in\cC_g\cap\cS}\tau_j(D_j-p),
\]
we have
\[
\begin{aligned}
\E\!\left[(T_g^{\mathrm{cell}}-pN_g^{\mathrm{cell}})\widetilde Y_g\right]
&=
(B_g+p\theta_g)
\sum_{i\in\cC_g\cap\cS}\E[D_i-p] \\
&\quad+
\sum_{i\in\cC_g\cap\cS}\sum_{j\in\cC_g\cap\cS}
\tau_j
\E\!\left[(D_i-p)(D_j-p)\right].
\end{aligned}
\]
The first term is zero. For $i\neq j$, independence and $\E[D_i-p]=0$ imply
\[
\E\!\left[(D_i-p)(D_j-p)\right]=0,
\]
whereas for $i=j$,
\[
\E[(D_i-p)^2]=p(1-p).
\]
Therefore,
\[
\E\!\left[(T_g^{\mathrm{cell}}-pN_g^{\mathrm{cell}})\widetilde Y_g\right]
=
p(1-p)\sum_{i\in\cC_g\cap\cS}\tau_i
=
p(1-p)\theta_g,
\]
and hence
\[
\E[H_g]=\theta_g.
\]
Summing across linkage cells and dividing by $N$ proves \eqref{eq:baseline-unbiased}.

For the variance formula, the expectation just obtained gives
\[
\Var_D(H_g)
=
\frac{\E_D[(T_g^{\mathrm{cell}}-pN_g^{\mathrm{cell}})^2\widetilde Y_g^2]}{p^2(1-p)^2}
-\theta_g^2.
\]
It therefore remains to evaluate the second moment in the numerator. Expanding the candidate sum separates it into three terms:
\begin{align*}
\E\!\left[(T_g^{\mathrm{cell}}-pN_g^{\mathrm{cell}})^2\widetilde Y_g^2\right]
={}&
(B_g+p\theta_g)^2
\E\!\left[(T_g^{\mathrm{cell}}-pN_g^{\mathrm{cell}})^2\right]\\
&+
2(B_g+p\theta_g)
\E\!\left[
(T_g^{\mathrm{cell}}-pN_g^{\mathrm{cell}})^2
\sum_{i\in\cC_g\cap\cS}\tau_i(D_i-p)
\right]\\
&+
\E\!\left[
(T_g^{\mathrm{cell}}-pN_g^{\mathrm{cell}})^2
\left\{
\sum_{i\in\cC_g\cap\cS}\tau_i(D_i-p)
\right\}^2
\right].
\end{align*}

To evaluate these three terms, we use independence and the second through fourth centered Bernoulli moments:
\[
\E(D_i-p)^2=p(1-p),\qquad
\E(D_i-p)^3=p(1-p)(1-2p),
\]
and
\[
\E(D_i-p)^4
=
p(1-p)\{1-3p(1-p)\}.
\]
Using these moments and independence across units gives
\[
\E\!\left[(T_g^{\mathrm{cell}}-pN_g^{\mathrm{cell}})^2\right]
=
N_g^{\mathrm{cell}}p(1-p),
\]
\[
\E\!\left[
(T_g^{\mathrm{cell}}-pN_g^{\mathrm{cell}})^2
\sum_{i\in\cC_g\cap\cS}\tau_i(D_i-p)
\right]
=
p(1-p)(1-2p)\theta_g,
\]
and
\[
\E\!\left[
(T_g^{\mathrm{cell}}-pN_g^{\mathrm{cell}})^2
\left\{
\sum_{i\in\cC_g\cap\cS}\tau_i(D_i-p)
\right\}^2
\right]
=
p(1-p)\{1+(N_g^{\mathrm{cell}}-6)p(1-p)\}Q_g
+
2p^2(1-p)^2\theta_g^2.
\]
Therefore,
\begin{align*}
\Var(H_g)
&=
\frac{1}{p^2(1-p)^2}
\Big[
N_g^{\mathrm{cell}}p(1-p)(B_g+p\theta_g)^2\\
&\qquad\qquad
+2p(1-p)(1-2p)(B_g+p\theta_g)\theta_g\\
&\qquad\qquad
+p(1-p)\{1+(N_g^{\mathrm{cell}}-6)p(1-p)\}Q_g
+2p^2(1-p)^2\theta_g^2
\Big]
-\theta_g^2\\
&=
\theta_g^2
+
\frac{1}{p(1-p)}
\left[
\begin{aligned}
&N_g^{\mathrm{cell}}(B_g+p\theta_g)^2
+2(1-2p)(B_g+p\theta_g)\theta_g
\\[-2pt]
&\quad+\{1+(N_g^{\mathrm{cell}}-6)p(1-p)\}Q_g
\end{aligned}
\right].
\end{align*}

Treatment assignments are independent across linkage cells, so the variance of their sum is the sum of the linkage-cell-specific variances. This proves \eqref{eq:baseline-var}--\eqref{eq:vg}.
\end{proof}

\subsection{Proof of Corollary~\ref{cor:variance-gap}}

\begin{proof}
To obtain the variance gap, we compare the candidate-sum and oracle contributions unit by unit. Under $N_g^{\mathrm{cell}}\leq1$, each contribution depends only on that unit's assignment, so the aggregate gap is the sum of the individual variance differences, divided by $N^2$.

When every represented linkage cell contains at most one experimental unit, the baseline estimator can be written as
\[
\widehat\tau^{\mathrm B}
=
\frac{1}{N}\sum_{i\in\cS}
\frac{D_i-p}{p(1-p)}
\widetilde Y_i(\mathbf D),
\]
where
\[
\widetilde Y_i(\mathbf D)
=
B_{g(i)}+D_i\tau_i.
\]

The Oracle estimator can similarly be written as
\[
\widehat\tau^{\mathrm O}
=
\frac{1}{N}\sum_{i\in\cS}
\frac{D_i-p}{p(1-p)}
 Y_i(D_i),
\]
where
\[
Y_i(D_i)
=
Y_i(0)+D_i\tau_i.
\]

The two individual contributions differ only in their fixed baseline term. We can therefore evaluate both variances with one identity: for any fixed scalar $b$, the Horvitz--Thompson contribution based on $b+D_i\tau_i$ satisfies
\begin{align*}
\Var_D\!\left[
\frac{D_i-p}{p(1-p)}\{b+D_i\tau_i\}
\right]
&=
\frac{(b+\tau_i)^2}{p}
+
\frac{b^2}{1-p}
-
\tau_i^2\\
&=
\frac{\{b+(1-p)\tau_i\}^2}{p(1-p)}.
\end{align*}

Apply this identity with $b=B_{g(i)}$ for the candidate-sum contribution and with $b=Y_i(0)$ for the oracle contribution. Subtracting the latter variance from the former, summing over $i\in\cS$, and dividing by $N^2$ gives
\[
\Var_D(\widehat\tau^{\mathrm B})
-
\Var_D(\widehat\tau^{\mathrm O})
=
\frac{1}{N^2p(1-p)}
\sum_{i\in\cS}
\left[
\{B_{g(i)}+(1-p)\tau_i\}^2
-
\{Y_i(0)+(1-p)\tau_i\}^2
\right],
\]
which is \eqref{eq:variance-gap}.
\end{proof}

\subsection{Proofs of Propositions~\ref{prop:weighted} and \ref{prop:fixed-adjustment}}

Both results follow by applying the linkage-cell-level expectation and variance calculation from Proposition~\ref{prop:baseline} to the modified contribution.

For weighting, multiplying the linkage-cell contribution $H_g$ in the proof of Proposition~\ref{prop:baseline} by $w_g$ changes its expectation to $w_g\theta_g$ and its variance to $w_g^2V_g$. Summing across linkage cells and normalizing by $W_{\cS}$ proves Proposition~\ref{prop:weighted}.

For fixed adjustment $h_g$, the same calculation applies after replacing $B_g$ by $B_g-h_g$. The expectation is unchanged because
\[
\E_D[(T_g^{\mathrm{cell}}-pN_g^{\mathrm{cell}})h_g]=0,
\]
while the variance expression is obtained by replacing $B_g+p\theta_g$ with $B_g-h_g+p\theta_g$ in $V_g$. This proves Proposition~\ref{prop:fixed-adjustment}.

\subsection{Proof of Proposition~\ref{prop:ri}}

\begin{proof}
We first show that the rejection probability is no greater than the
nominal level conditional on the observed assignment being admissible.
We then show that randomization at the boundary gives equality.
Under \eqref{eq:sharp-null}, equation
\eqref{eq:linked-cell-outcome} implies
$\widetilde Y_g(\mathbf d)=B_g$ for every linkage cell $g$
and every assignment $\mathbf d$. Thus
$\widetilde{\mathbf Y}=\mathbf B=(B_g)_{g=1}^G$ is fixed,
and so is the set $\mathcal A=\mathcal A(\mathbf B)$
of assignments for which the statistic is defined.
By assumption, $\Pr_D\{\mathbf D\in\mathcal A\}>0$.
Conditional on $\mathbf D\in\mathcal A$, the observed
assignment follows the same distribution as an independent
Bernoulli draw $\mathbf D^*$ conditional on
$\mathbf D^*\in\mathcal A$.

Let $Z=\mathcal T(\mathbf B,\mathbf D)$, and let
$z_1>\cdots>z_K$ be the distinct values in its support
under this conditional distribution. Define
\[
q_j=\Pr_D\{Z=z_j\mid\mathbf D\in\mathcal A\},
\qquad Q_0=0,
\qquad Q_j=\sum_{\ell=1}^j q_\ell.
\]
On the event $Z=z_j$, the randomization $p$-value in
\eqref{eq:ri-pvalue} equals $Q_j$. Therefore,
\[
\Pr_D\{p_{\mathrm{RI}}\leq\alpha
\mid\mathbf D\in\mathcal A\}
=\sum_{j:Q_j\leq\alpha}q_j
\leq\alpha,
\]
which establishes \eqref{eq:ri-level}.

For equality, let $U$ be independent of the treatment
assignment and uniformly distributed on $(0,1)$.
On the event $Z=z_j$, define the randomized $p$-value
\[
p_{\mathrm{RI}}^{\mathrm{rand}}=Q_{j-1}+Uq_j.
\]
Conditional on $\mathbf D\in\mathcal A$, this $p$-value
is uniform on $(0,1)$: each interval $(Q_{j-1},Q_j)$
has length $q_j$, is selected with probability $q_j$,
and receives a uniform draw. Rejecting when
$p_{\mathrm{RI}}^{\mathrm{rand}}\leq\alpha$ therefore
gives conditional rejection probability exactly $\alpha$.
\end{proof}

\subsection{Proof of Proposition~\ref{thm:general-clt}}

\begin{proof}
We prove the result in two steps: establish the central limit theorem with the fixed adjustment $h^*$, and then show that estimating the adjustment is negligible at the root-$N$ scale. These two claims imply the proposition by Slutsky's theorem.

For the fixed-adjustment claim, the estimator has the representation
\begin{align*}
\sqrt N\{\widehat\tau(w,h^*)-\tau_{w,\cS}\}
&=
\frac{\sqrt N}{W_{\cS}}
\sum_g\{H_g(w,h^*)-w_g\theta_g\}\\
&=
\sqrt{V_N(w,h^*)}\,
\frac{\sum_g\{H_g(w,h^*)-w_g\theta_g\}}
{\{\sum_gw_g^2V_g(h^*)\}^{1/2}}.
\end{align*}
Since $V_N(w,h^*)\to V(w,h^*)\in(0,\infty)$, it suffices to prove that the standardized sum in the second line converges to $\mathcal N(0,1)$. We do so by verifying the independent-array Lyapunov condition.

The linkage-cell-level calculation in Proposition~\ref{prop:fixed-adjustment} gives
\[
\E_D[H_g(w,h^*)]=w_g\theta_g,
\qquad
\Var_D\{H_g(w,h^*)\}=w_g^2V_g(h^*).
\]
The summands are therefore centered, and they are independent because each depends only on assignments in its own linkage cell. To verify Lyapunov's condition, we bound the sum of fourth moments and the squared total variance separately. Bounded linkage-cell sizes and weights bound the multiplier of $\widetilde Y_g-h_g^*$ in \eqref{eq:Hg}. Thus, using Jensen's inequality to center the contributions and Assumption~\ref{ass:cells}(ii) evaluated at $h=h^*$,
\[
\sum_{g\in\mathcal G_N}
\E_D\!\left[|H_g(w,h^*)-w_g\theta_g|^4\right]
\leq C\sum_{g\in\mathcal G_N}
\left\{\E_D[|\widetilde Y_g(\mathbf D)|^4]+|h_g^*|^4\right\}
=O(N).
\]
The same bound holds for the sum of noncentered fourth moments of $H_g(w,h^*)$.
For the total variance, the same assumption gives
\[
\frac{N}{W_{\cS}^2}\sum_gw_g^2V_g(h^*)
=V_N(w,h^*)\longrightarrow V(w,h^*)\in(0,\infty),
\]
and $W_{\cS}/N$ converges to a positive finite limit. Thus $\sum_gw_g^2V_g(h^*)\asymp N$. Combining the two bounds,
\[
\frac{\sum_g\E_D[|H_g(w,h^*)-w_g\theta_g|^4]}
{\{\sum_gw_g^2V_g(h^*)\}^2}
=O(N^{-1})\longrightarrow0.
\]
Lyapunov's central limit theorem applies to the standardized sum. The representation above and Slutsky's theorem then yield
\[
\sqrt N\{\widehat\tau(w,h^*)-\tau_{w,\cS}\}
\Rightarrow\mathcal N\{0,V(w,h^*)\}.
\]

For the second claim, we need to show
\[
\sqrt N\{\widehat\tau(w,\widehat h)-\widehat\tau(w,h^*)\}=o_p(1).
\]
We expand the adjustment to separate a linear term from the Taylor remainder. Write
\[
\dot h_g^*
=
\left.\frac{\partial h(x_g;\beta)}{\partial\beta}\right|_{\beta=\beta^*}.
\]
A second-order Taylor expansion around $\beta^*$, together with
Assumption~\ref{ass:adjustment}, gives, uniformly in $g$,
\[
h(x_g;\widehat\beta)-h(x_g;\beta^*)
=
\dot h_g^{*\prime}(\widehat\beta-\beta^*)
+
O_p\!\left(\|\widehat\beta-\beta^*\|^2\right).
\]
Substituting into the estimator, using bounded weights and
$\sum_g|T_g^{\mathrm{cell}}-pN_g^{\mathrm{cell}}|\leq\sum_gN_g^{\mathrm{cell}}=N$, gives
\begin{align*}
&\sqrt N\{\widehat\tau(w,\widehat h)-\widehat\tau(w,h^*)\}\\
&\quad=
-\frac{N}{W_{\cS}p(1-p)}
\left\{\frac{1}{\sqrt N}\sum_gw_g(T_g^{\mathrm{cell}}-pN_g^{\mathrm{cell}})\dot h_g^*\right\}'
(\widehat\beta-\beta^*)\\
&\qquad\quad+
O_p\!\left(\frac{N}{W_{\cS}}\sqrt N\,
\|\widehat\beta-\beta^*\|^2\right).
\end{align*}
Because $N/W_{\cS}=O(1)$ and $\widehat\beta-\beta^*=O_p(N^{-1/2})$, it suffices to show that the normalized sum in braces is $O_p(1)$; the displayed remainder is already $O_p(N^{-1/2})=o_p(1)$.

To establish that bound on the sum, we use its zero mean and bounded second moment. Independence of Bernoulli assignments gives
\[
\E_D\!\left[\frac{1}{\sqrt N}
\sum_gw_g(T_g^{\mathrm{cell}}-pN_g^{\mathrm{cell}})\dot h_g^*\right]=0
\]
and
\[
\E_D\left\|\frac{1}{\sqrt N}
\sum_gw_g(T_g^{\mathrm{cell}}-pN_g^{\mathrm{cell}})\dot h_g^*\right\|^2
=
\frac{p(1-p)}{N}
\sum_gw_g^2N_g^{\mathrm{cell}}\|\dot h_g^*\|^2
=O(1).
\]
The last equality follows from bounded weights and derivatives and $\sum_gN_g^{\mathrm{cell}}=N$. Hence the normalized sum is $O_p(1)$, the linear term in the expansion is $O_p(N^{-1/2})$, and
\[
\sqrt N\{\widehat\tau(w,\widehat h)-\widehat\tau(w,h^*)\}=o_p(1).
\]
Combining this result with the fixed-adjustment central limit theorem proves the proposition.
\end{proof}

\subsection{Proof of Proposition~\ref{prop:optimal-adjustment}}

\begin{proof}
We first find the variance-minimizing adjustment within each category and then show that it equals $h_{w,N}^*(x)$ as defined in the proposition. Because $W_{\cS}$ does not depend on the adjustment, minimizing the variance of $\widehat\tau(w,h)$ is equivalent to minimizing
\[
\sum_g w_g^2V_g(h_g).
\]
If $h_g=h(x_g)$, the minimization separates across values of $x$. For a fixed category $x$, the terms in $V_g\{h(x)\}$ that depend on $h(x)$ are, up to the common factor $1/[p(1-p)]$,
\[
\sum_{g:x_g=x}w_g^2
\left[
N_g^{\mathrm{cell}}\{B_g-h(x)+p\theta_g\}^2
+2(1-2p)\{B_g-h(x)+p\theta_g\}\theta_g
\right].
\]
This objective is a strictly convex quadratic whenever $W_{2,x}>0$, so its unique minimizer is characterized by the first-order condition. Differentiating with respect to $h(x)$ gives
\[
\sum_{g:x_g=x}w_g^2
\left[
N_g^{\mathrm{cell}}\{B_g-h(x)+p\theta_g\}
+(1-2p)\theta_g
\right]
=0.
\]
Since
\[
W_{2,x}
=
\sum_{g:x_g=x}w_g^2N_g^{\mathrm{cell}},
\]
the first-order condition can be solved for $h(x)$ as
\begin{equation}
h(x)
=
\frac{1}{W_{2,x}}
\sum_{g:x_g=x}w_g^2
\left[
N_g^{\mathrm{cell}}(B_g+p\theta_g)
+(1-2p)\theta_g
\right].
\label{eq:optimal-adjustment-foc-proof}
\end{equation}
It remains to identify this minimizer with $(1-p)\mu_{1,w,N}(x)+p\mu_{0,w,N}(x)$. To evaluate that combination, first note that, for $i\in\cC_g\cap\cS$,
\[
\E_D\!\left[\widetilde Y_g(\mathbf D)\mid D_i=d\right]
=
B_g+p\theta_g+(d-p)\tau_i.
\]
Summing over experimental units in linkage cell $g$ yields
\begin{align*}
\sum_{i\in\cC_g\cap\cS}
\E_D\!\left[\widetilde Y_g(\mathbf D)\mid D_i=1\right]
&=
N_g^{\mathrm{cell}}(B_g+p\theta_g)+(1-p)\theta_g,\\
\sum_{i\in\cC_g\cap\cS}
\E_D\!\left[\widetilde Y_g(\mathbf D)\mid D_i=0\right]
&=
N_g^{\mathrm{cell}}(B_g+p\theta_g)-p\theta_g.
\end{align*}
Therefore, by the definition of $\mu_{d,w,N}(x)$,
\begin{align*}
&(1-p)\mu_{1,w,N}(x)+p\mu_{0,w,N}(x)
\\
&\qquad=
\frac{1}{W_{2,x}}
\sum_{g:x_g=x}w_g^2
\left[
N_g^{\mathrm{cell}}(B_g+p\theta_g)
+(1-2p)\theta_g
\right].
\end{align*}
Comparing this expression with \eqref{eq:optimal-adjustment-foc-proof} gives
\[
h(x)
=
(1-p)\mu_{1,w,N}(x)+p\mu_{0,w,N}(x)
=
h_{w,N}^*(x),
\]
which proves the proposition.
\end{proof}

\subsection{Proof of Corollary~\ref{cor:saturated-adjustment}}

\begin{proof}
We first establish the root-$N$ rate for each weighted arm mean. Because there are finitely many retained categories and the adjustment is a fixed linear combination of the two arm means, these rates will give the rate for $\widehat h_w$ and verify Assumption~\ref{ass:adjustment}.

Fix a retained category $x$ and consider $d=1$. On the event that the denominator defining $\widehat\mu_{1,w}(x)$ is positive,
\begin{align*}
&\sqrt N\{\widehat\mu_{1,w}(x)-\mu_{1,w,N}(x)\}\\
&\quad=
\frac{
N^{-1/2}\sum_{i\in\cS:x_{g(i)}=x}
w_{g(i)}^2D_i
\{\widetilde Y_{g(i)}(\mathbf D)-\mu_{1,w,N}(x)\}
}{
N^{-1}\sum_{i\in\cS:x_{g(i)}=x}w_{g(i)}^2D_i
}.
\end{align*}
It suffices to show that the numerator is $O_p(1)$ and the denominator is bounded away from zero in probability. The latter also ensures that the ratio is defined with probability approaching one.

For the denominator, it suffices to bound it below by a quantity converging in probability to a positive constant. Since $w_g\geq\underline w>0$,
\[
\frac{1}{N}\sum_{i\in\cS:x_{g(i)}=x}w_{g(i)}^2D_i
\geq
\underline w^2\frac{N_x}{N}
\left(\frac{1}{N_x}\sum_{i\in\cS:x_{g(i)}=x}D_i\right).
\]
Because $N_x/N\to\pi_x>0$, we have $N_x\to\infty$. The law of large numbers for independent Bernoulli assignments implies that the average in parentheses converges in probability to $p$. The right-hand side therefore converges in probability to $\underline w^2\pi_xp>0$, proving the required denominator bound.

For the numerator, we establish a zero mean and a bounded second moment. By the definition of $\mu_{1,w,N}(x)$,
\begin{align*}
&\E_D\!\left[
\sum_{i\in\cS:x_{g(i)}=x}w_{g(i)}^2D_i
\{\widetilde Y_{g(i)}(\mathbf D)-\mu_{1,w,N}(x)\}
\right]\\
&\quad=
p\sum_{i\in\cS:x_{g(i)}=x}w_{g(i)}^2
\left\{\E_D[\widetilde Y_{g(i)}(\mathbf D)\mid D_i=1]
-\mu_{1,w,N}(x)\right\}
=0.
\end{align*}
To bound its second moment, first note that Assumption~\ref{ass:cells}(ii) implies the required candidate-sum second-moment bound. Since $G_N\leq N$, Cauchy--Schwarz gives
\[
\frac1N\sum_g\E_D[\widetilde Y_g(\mathbf D)^2]
\leq
\left\{\frac{G_N}{N}\frac1N\sum_g
\E_D[|\widetilde Y_g(\mathbf D)|^4]\right\}^{1/2}
=O(1).
\]
Positive category shares and bounded weights then imply $\mu_{1,w,N}(x)^2=O(1)$ by Cauchy--Schwarz applied to \eqref{eq:oracle-arm-mean-weighted}.

Now group the numerator by linkage cell. Contributions from different linkage cells are independent, but their individual expectations need not be zero. Subtracting each linkage-cell contribution's own expectation leaves the total sum unchanged, since those expectations sum to zero. We can therefore add the linkage-cell-specific variances. Together with bounded linkage-cell sizes, the preceding moment bounds give
\begin{align*}
&\E_D\!\left[
\left|
\frac{1}{\sqrt N}
\sum_{i\in\cS:x_{g(i)}=x}w_{g(i)}^2D_i
\{\widetilde Y_{g(i)}(\mathbf D)-\mu_{1,w,N}(x)\}
\right|^2
\right]\\
&\quad=
\frac{1}{N}\sum_{g:x_g=x}
\Var_D\!\left[
w_g^2\sum_{i\in\cC_g\cap\cS}D_i
\{\widetilde Y_g(\mathbf D)-\mu_{1,w,N}(x)\}
\right]\\
&\quad\leq\frac{C}{N}\sum_{g:x_g=x}
\left\{\E_D[\widetilde Y_g(\mathbf D)^2]+\mu_{1,w,N}(x)^2\right\}
=O(1).
\end{align*}
Thus the numerator of the ratio is $O_p(1)$. Combining this with the denominator bound yields
\[
\widehat\mu_{1,w}(x)-\mu_{1,w,N}(x)=O_p(N^{-1/2}).
\]
Replacing $D_i$ by $1-D_i$ and $p$ by $1-p$ proves the same rate for $d=0$.

It remains to transfer these rates to the adjustment and check its parameterization. The support of $x_g$ is fixed and finite, so the arm-mean rates hold jointly over all retained categories. By \eqref{eq:hhat-saturated-weighted} and \eqref{eq:optimal-h-weighted},
\begin{align*}
\widehat h_w(x)-h_{w,N}^*(x)
&=(1-p)\{\widehat\mu_{1,w}(x)-\mu_{1,w,N}(x)\}\\
&\quad+p\{\widehat\mu_{0,w}(x)-\mu_{0,w,N}(x)\}\\
&=O_p(N^{-1/2}).
\end{align*}
Parameterizing the saturated adjustment by its finite vector of category-specific values makes $h(x;\beta)$ linear in a fixed-dimensional parameter. Its first derivatives are bounded and its second derivatives are zero. Thus both the rate and the smoothness requirements in Assumption~\ref{ass:adjustment} hold with finite-population target $h_{w,N}^*$.

Finally, to verify that this target satisfies the adjustment-moment bound in Assumption~\ref{ass:cells}(ii), apply Jensen's inequality to \eqref{eq:oracle-arm-mean-weighted}. Bounded weights and linkage-cell sizes, fixed $p\in(0,1)$, and positive category shares give
\[
|\mu_{d,w,N}(x)|^4
\leq\frac{C}{N_x}\sum_{g:x_g=x}
\E_D[|\widetilde Y_g(\mathbf D)|^4]
=O(1),\qquad d\in\{0,1\}.
\]
There are finitely many categories, so \eqref{eq:optimal-h-weighted} implies $\max_x|h_{w,N}^*(x)|^4=O(1)$. Since $G_N\leq N$, it follows that $N^{-1}\sum_g|h_{w,N}^*(x_g)|^4=O(1)$.
\end{proof}

\subsection{Proof of Proposition~\ref{prop:efficiency}}

To identify the optimal weights, we express the limiting variance as a function of $w(x)$, derive a common lower bound, and verify that inverse-variance weights attain it. Because the weights under consideration are functions of $x_g$,
\begin{align*}
V_N(w,h)
&=
\frac{N}{W_{\cS}^2}
\sum_x w(x)^2
\sum_{g\in\mathcal G_N:x_g=x}V_g(h),\\
\frac{W_{\cS}}{N}
&=
\sum_x\frac{N_x}{N}w(x).
\end{align*}
Under \eqref{eq:Nx-pix}--\eqref{eq:Vbarx},
\begin{equation}
V(w,h)
=
\frac{
\sum_x\pi_x w(x)^2\overline V_h(x)
}{
\left\{\sum_x\pi_xw(x)\right\}^2
}.
\label{eq:weight-objective-proof}
\end{equation}

We now bound \eqref{eq:weight-objective-proof} using Cauchy--Schwarz. The relevant factorization of its denominator is
\[
\sum_x\pi_xw(x)
=
\sum_x
\left\{\sqrt{\pi_x\overline V_h(x)}\,w(x)\right\}
\sqrt{\frac{\pi_x}{\overline V_h(x)}},
\]
Cauchy--Schwarz implies
\[
\left\{\sum_x\pi_xw(x)\right\}^2
\le
\left\{\sum_x\pi_xw(x)^2\overline V_h(x)\right\}
\left\{\sum_x\frac{\pi_x}{\overline V_h(x)}\right\}.
\]
Thus the asymptotic variance is bounded below by
\[
\left\{\sum_x\frac{\pi_x}{\overline V_h(x)}\right\}^{-1},
\]
and this lower bound is attained by $w(x)\propto1/\overline V_h(x)$. This proves \eqref{eq:optimal-weight}. Under \eqref{eq:homoskedastic-benchmark}, $\overline V_h(m)=cm$, so the optimal weight is proportional to $1/m$. \qed

\subsection{Proof of Proposition~\ref{prop:cluster-se}}

\begin{proof}

We first establish the general variance-estimator
limit, then verify the stated centering choices and the connection
with the WLS implementation. Let
\[
\varepsilon_g=H_g(w,h^*)-w_g\theta_g,
\qquad
b_g=w_g(\theta_g-N_g^{\mathrm{cell}}\delta_g^*),
\]
so that
\[
S_g^*\equiv H_g(w,h^*)-w_gN_g^{\mathrm{cell}}\delta_g^*
=\varepsilon_g+b_g,
\qquad \E_D\varepsilon_g=0.
\]
The fourth-moment bound for $H_g(w,h^*)$ established in the proof of
Proposition~\ref{thm:general-clt}, together with Jensen's inequality,
gives
\[
\sum_g\E_D|\varepsilon_g|^4=O(N),
\qquad
\sum_g|w_g\theta_g|^4=O(N).
\]
Bounded $N_g^{\mathrm{cell}}$ and $w_g$, together with the assumed average
fourth-moment bound on $\delta_g^*$, give $\sum_g|b_g|^4=O(N)$.
Independence across linkage cells and
Cauchy--Schwarz then give
\[
\Var_D\!\left(\sum_g\{\varepsilon_g^2-\E_D\varepsilon_g^2\}\right)
=O(N),
\qquad
\Var_D\!\left(\sum_g b_g\varepsilon_g\right)=O(N).
\]
Because $N/W_{\cS}^2=O(N^{-1})$, expanding the square yields
\begin{align*}
\frac{N}{W_{\cS}^2}\sum_g(S_g^*)^2
&=\frac{N}{W_{\cS}^2}
\sum_g\{\E_D\varepsilon_g^2+b_g^2\}+o_p(1)\\
&=V_N(w,h^*)+\Delta_N(w,\delta^*)+o_p(1).
\end{align*}
Here the deterministic term equals the gap defined in
\eqref{eq:cluster-gap-N}, since
\[
\frac{N}{W_{\cS}^2}\sum_g b_g^2
=\frac{N}{W_{\cS}^2}\sum_g w_g^2
\{\theta_g-N_g^{\mathrm{cell}}\delta_g^*\}^2
=\Delta_N(w,\delta^*)\geq0.
\]

For the feasible statistic, define
$\widehat S_g=\widehat H_g-w_gN_g^{\mathrm{cell}}\widehat\delta_g$.
The Taylor argument in the proof of
Proposition~\ref{thm:general-clt}, bounded weights, and bounded
linkage-cell sizes imply
\[
\frac{N}{W_{\cS}^2}
\sum_g\{\widehat H_g-H_g(w,h^*)\}^2=o_p(1).
\]
The assumed mean-square approximation of $\widehat\delta_g$
similarly gives
\[
\frac{N}{W_{\cS}^2}
\sum_g(w_gN_g^{\mathrm{cell}})^2(\widehat\delta_g-\delta_g^*)^2=o_p(1).
\]
Combining these bounds yields
$N W_{\cS}^{-2}\sum_g(\widehat S_g-S_g^*)^2=o_p(1)$.
The scaled sum of $(S_g^*)^2$ is $O_p(1)$ by the preceding argument.
Cauchy--Schwarz therefore shows that replacing $S_g^*$ by
$\widehat S_g$ changes the scaled sum of squares by $o_p(1)$.
It follows that
\[
\widehat V_{\mathrm{cl}}
=V_N(w,h^*)+\Delta_N(w,\delta^*)+o_p(1)
\overset p\longrightarrow V(w,h^*)+\Delta(w,\delta^*).
\]
The gap is nonnegative because it is a sum of squares. If all treatment
effects equal $\tau$ and $\delta_g^*=\tau$, then
$\theta_g=N_g^{\mathrm{cell}}\delta_g^*$ and the gap is zero.

For the common centering $\widehat\delta_g=\widehat\tau$, take
$\delta_g^*=\tau_{w,\cS}$. Proposition~\ref{thm:general-clt} gives
$\widehat\tau-\tau_{w,\cS}=O_p(N^{-1/2})$.
Jensen's inequality applied to the cell contributions and
$W_{\cS}\asymp N$ give $\tau_{w,\cS}=O(1)$.
Thus both centering conditions hold, establishing the stated
application to all four estimators.

We next consider the WLS adjustment and centering described after
the proposition. Under Bernoulli assignment, for
$i\in\cC_g\cap\cS$,
\[
\E_D[\widetilde Y_g(\mathbf D)\mid D_i=d]
=B_g+p\theta_g+(d-p)\tau_i.
\]
The finite-population WLS normal equations therefore imply
\[
\beta_1^*-\beta_0^*
=(N\mathsf A_N)^{-1}\sum_g w_gx_g\theta_g.
\]
The moment conditions, bounded covariates, and bounded inverse of
$\mathsf A_N$ imply $\|\beta_d^*\|=O(1)$, so $\delta_g^*$ is
uniformly bounded. Assumption~\ref{ass:adjustment} for the stacked
arm-specific coefficients implies
$\max_g|\widehat\delta_g-\delta_g^*|=O_p(N^{-1/2})$.
Hence this centering also satisfies the proposition's moment and
approximation conditions. Under homogeneous effects, the intercept
implies $\beta_1^*-\beta_0^*=\tau e_1$, where $e_1$ selects the
intercept. Consequently $\delta_g^*=\tau$ and the gap vanishes.

For the CRVE equivalence, let $p_1=p$, $p_0=1-p$, and define
\[
\widehat{\mathsf A}_{d,N}
=\frac1N\sum_{i\in\cS}
\1\{D_i=d\}w_{g(i)}x_{g(i)}x_{g(i)}',
\qquad
\overline x_w=\frac1{W_{\cS}}\sum_g w_gN_g^{\mathrm{cell}}x_g.
\]
Independence across linkage cells and bounded covariates imply
$\widehat{\mathsf A}_{d,N}=p_d\mathsf A_N+O_p(N^{-1/2})$.
Since $\mathsf A\succ0$, both arm-specific regressions are identified
with probability approaching one. The intercept gives the identity
\[
\overline x_w'\mathsf A_N^{-1}=\frac{N}{W_{\cS}}e_1',
\]
and therefore
\[
\overline x_w'\widehat{\mathsf A}_{d,N}^{-1}
=\frac{N}{p_dW_{\cS}}e_1'+O_p(N^{-1/2}).
\]
Let $\widehat u_{dg}=\widetilde Y_g(\mathbf D)-x_g'\widehat\beta_d$,
and write the arm-specific cluster scores as
\[
s_{1g}=w_gT_g^{\mathrm{cell}}x_g\widehat u_{1g},
\qquad
s_{0g}=w_g(N_g^{\mathrm{cell}}-T_g^{\mathrm{cell}})x_g\widehat u_{0g}.
\]
The centered specification in \eqref{eq:interacted-wls} has treatment
coefficient $\widehat t=\overline x_w'(\widehat\beta_1-\widehat\beta_0)$.
Expressing its CR0 variance in this equivalent arm-specific
parameterization gives the exact identity
\[
N\widehat{\operatorname{Var}}_{\mathrm{CR0}}(\widehat t)
=\frac1N\sum_g
\left\{
\overline x_w'\widehat{\mathsf A}_{1,N}^{-1}s_{1g}
-\overline x_w'\widehat{\mathsf A}_{0,N}^{-1}s_{0g}
\right\}^2.
\]
The moment conditions and bounded coefficients imply
$N^{-1}\sum_g(\|s_{1g}\|^2+\|s_{0g}\|^2)=O_p(1)$.
Thus, replacing the inverse-matrix factors by their preceding
approximations changes this expression by $o_p(1)$, yielding
\[
N\widehat{\operatorname{Var}}_{\mathrm{CR0}}(\widehat t)
=\frac{N}{W_{\cS}^2}\sum_g
\left\{
\frac{w_gT_g^{\mathrm{cell}}}{p}\widehat u_{1g}
-\frac{w_g(N_g^{\mathrm{cell}}-T_g^{\mathrm{cell}})}{1-p}\widehat u_{0g}
\right\}^2+o_p(1).
\]
Direct substitution of $\widehat h_g$ and $\widehat\delta_g$ shows that
\[
\frac{w_gT_g^{\mathrm{cell}}}{p}\widehat u_{1g}
-\frac{w_g(N_g^{\mathrm{cell}}-T_g^{\mathrm{cell}})}{1-p}\widehat u_{0g}
=\widehat H_g-w_gN_g^{\mathrm{cell}}\widehat\delta_g
=\widehat S_g.
\]
This proves the claimed equivalence. Since $G_N\to\infty$ and the
regression dimension is fixed, the usual CR1 correction tends to one
and has the same limit. If $N_g^{\mathrm{cell}}=1$ for every represented cell, each
cluster score consists of one observation, so CR0 and HC0 coincide.

\end{proof}

\subsection{Proofs of Propositions~\ref{prop:strat-main}, \ref{prop:strat-correction}, and \ref{prop:strat-general-finite}}

We first prove the weighted-and-adjusted formulas in Proposition~\ref{prop:strat-general-finite}. The calculation for the uncorrected contribution uses only the stratified design; \eqref{eq:strat-nonempty} is needed only for the corrected contribution. We then obtain Propositions~\ref{prop:strat-main} and \ref{prop:strat-correction} as special cases.

Fix $(g,r)$ with $N_{gr}>0$. Under fixed-count assignment within stratum $r$,
\begin{equation}
\Var_D(D_i)=p_r(1-p_r),\qquad
\Cov_D(D_i,D_j)
=-\frac{p_r(1-p_r)}{N_r^{\mathrm{str}}-1}
\qquad(i\ne j,\ i,j\in\cS_r).
\label{eq:fixedcount-cov}
\end{equation}
Hence, for $i\in\cC_g\cap\cS_r$,
\begin{equation}
\begin{aligned}
\Cov_D(T_{gr},D_i)
&=\Var_D(D_i)
+\sum_{\substack{j\in\cC_g\cap\cS_r\\j\ne i}}
\Cov_D(D_j,D_i)\\
&=p_r(1-p_r)\frac{N_r^{\mathrm{str}}-N_{gr}}
{N_r^{\mathrm{str}}-1}.
\end{aligned}
\label{eq:strat-cov}
\end{equation}
The sum in the first line is empty when $N_{gr}=1$, so the formula includes singleton intersections. Because $T_{gr}-p_rN_{gr}$ has mean zero, fixed terms in $\widetilde Y_g-h_g$ do not contribute. Assignment independence across strata and \eqref{eq:strat-cov} therefore give
\[
\begin{aligned}
&\E_D\!\left[
\frac{T_{gr}-p_rN_{gr}}{p_r(1-p_r)}
\{\widetilde Y_g-h_g\}
\right]\\
&\qquad=\frac1{p_r(1-p_r)}
\sum_{i\in\cC_g\cap\cS_r}\tau_i\Cov_D(T_{gr},D_i)
=\frac{N_r^{\mathrm{str}}-N_{gr}}
{N_r^{\mathrm{str}}-1}\theta_{gr}.
\end{aligned}
\]
Multiplying by $w_g$ proves the first line of \eqref{eq:strat-general-bias}. Summing and using $(N_r^{\mathrm{str}}-N_{gr})/(N_r^{\mathrm{str}}-1)=1-(N_{gr}-1)/(N_r^{\mathrm{str}}-1)$ proves its second line. Under \eqref{eq:strat-nonempty}, the correction factor in \eqref{eq:strat-general-contribution} is well defined and cancels the expectation factor above, proving \eqref{eq:strat-general-unbiased}. Setting $w_g=1$ and $h_g=0$ in the unrestricted calculation proves Proposition~\ref{prop:strat-main}; the same specialization of the corrected calculation proves Proposition~\ref{prop:strat-correction}. \qed

\subsection{Proofs of Propositions~\ref{prop:strat-covariance} and \ref{prop:strat-general-covariance}}

Because the corrected contributions are defined under \eqref{eq:strat-nonempty}, that condition is in force throughout this proof. We prove the general result in three steps: factor cross-stratum products, use that factorization to identify the covariance support, and then partition the variance. The unweighted result follows at the end. Separate the centered assignment multiplier and the effect contribution:
\[
Q_{gr}=\frac{N_r^{\mathrm{str}}-1}{N_r^{\mathrm{str}}-N_{gr}}
\frac{T_{gr}-p_rN_{gr}}{p_r(1-p_r)},\qquad
F_{gr}=\sum_{i\in\cC_g\cap\cS_r}\tau_i(D_i-p_r).
\]
Both have mean zero and depend only on assignments in stratum $r$. Set them to zero for empty intersections. Then
\[
\widetilde H_{gr}(w,h)=w_gQ_{gr}\{\widetilde Y_g-h_g\},\qquad
\widetilde Y_g-h_g=\E_D(\widetilde Y_g-h_g)+\sum_tF_{gt},\qquad
\E_D[Q_{gr}F_{gr}]=\theta_{gr}.
\]
Thus the fixed adjustment changes only the constant term, and fixed weights multiply products by the corresponding cell weights.

For $r\ne s$, independence across strata implies that a product survives only if each centered multiplier is paired with an outcome contribution from its own stratum. The direct and reciprocal pairings give
\begin{equation}
\begin{aligned}
\E_D[\widetilde H_{gr}(w,h)\widetilde H_{ks}(w,h)]
&=w_gw_k\{\theta_{gr}\theta_{ks}
+\E_D[Q_{gr}F_{kr}]\E_D[Q_{ks}F_{gs}]\},\\
C_{gr,ks}(w,h)
&=w_gw_k\E_D[Q_{gr}F_{kr}]\E_D[Q_{ks}F_{gs}].
\end{aligned}
\label{eq:strat-cross-factorization}
\end{equation}
If either reciprocal intersection is empty, the corresponding $F$ is zero, proving the covariance-support claim. Finally, expand the variance of $W_{\cS}^{-1}\sum_{g,r}\widetilde H_{gr}(w,h)$ over ordered pairs and partition them into the same intersection, distinct cells in the same stratum, distinct strata in the same cell, and distinct cells and strata. In the last class, only pairs with $N_{gs}N_{kr}>0$ remain. This gives \eqref{eq:strat-general-variance}. Setting $w_g=1$, $h_g=0$, and $W_{\cS}=N$ gives \eqref{eq:strat-finite-variance} and proves Proposition~\ref{prop:strat-covariance}. \qed

\subsection{Proofs of Propositions~\ref{prop:fixedcount-asymptotics} and \ref{prop:strat-general-large}}

We prove the general result in four steps: construct a contribution envelope, establish first-order equivalence, couple the fixed-count design to independent assignments, and apply an independent-array CLT. The unweighted, unadjusted result follows by setting $w_g=1$ and $h_g=0$.

\paragraph{Bounds from the moment condition.}

The growing-strata arguments require bounds that hold for every within-cell assignment pattern, whereas Assumption~\ref{ass:strat-general-large} bounds moments under the actual design. We first show that the stated moment condition supplies those bounds. For this proof only, let $\mathbf d_g$ denote a binary assignment vector for the experimental participants in linkage cell $g$, and write
\[
L_g=\max_{\mathbf d_g\in\{0,1\}^{N_g^{\mathrm{cell}}}}
\sum_{r:N_{gr}>0}|\widetilde H_{gr}(w,h;\mathbf d_g)|.
\]
We will establish $\sum_gL_g^4=O(N)$ and then use $L_g$ only as an envelope.

Every linkage cell contains at most $\overline M$ experimental participants. Since both arms of every stratum grow proportionally to $N$, all such within-cell assignment patterns are eventually feasible, with probabilities bounded below by a constant $c>0$ uniformly over linkage cells and patterns. Indeed, for a pattern with $t_{gr}$ treated participants in each intersection, its probability is
\[
\prod_{r:N_{gr}>0}
\frac{(T_r^{\mathrm{str}})_{t_{gr}}
(N_r^{\mathrm{str}}-T_r^{\mathrm{str}})_{N_{gr}-t_{gr}}}
{(N_r^{\mathrm{str}})_{N_{gr}}},
\]
where $(a)_k=a(a-1)\cdots(a-k+1)$. Bounded linkage-cell size and treatment probabilities bounded away from zero and one give the claimed lower bound. Therefore
\[
L_g^4\le
\frac{\overline M^3}{c}
\sum_{r:N_{gr}>0}\E_D|\widetilde H_{gr}(w,h)|^4,
\qquad
\sum_gL_g^4=O(N)
\]
by \eqref{eq:strat-general-large-moment}. There are at most $N$ represented linkage cells, so H\"older's inequality also gives $\sum_gL_g^j=O(N)$ for $j=1,2,3$.

The same bounds control corrected and uncorrected linkage-cell totals, under either the fixed-count design or auxiliary independent assignments: the uncorrected intersection contribution is $(N_r^{\mathrm{str}}-N_{gr})/(N_r^{\mathrm{str}}-1)$ times its corrected counterpart, with multiplier in $(0,1]$. In addition,
\[
1\le\frac{N_r^{\mathrm{str}}-1}{N_r^{\mathrm{str}}-N_{gr}}\le N_{gr}\le\overline M.
\]
Thus each relevant linkage-cell total is bounded by $CL_g$ for every within-cell assignment pattern. Constants in the following bounds may change from line to line but do not depend on $N$.

\paragraph{First-order equivalence.}
We seek a bound smaller than $N^{-1/2}$ for the difference between the estimators. Since $R$ is fixed, $N_r^{\mathrm{str}}/N\to\lambda_r>0$, and $N_{gr}$ is bounded,
\[
\max_{g,r:N_{gr}>0}
\left|\frac{N_r^{\mathrm{str}}-1}{N_r^{\mathrm{str}}-N_{gr}}-1\right|=O(N^{-1}).
\]
The moment bounds above therefore imply, uniformly over assignments,
\[
\left|
\widehat\tau^{\mathrm{SR,corr}}(w,h)
-\widehat\tau^{\mathrm{SR}}(w,h)
\right|
\le
\frac{C}{NW_{\cS}}\sum_gL_g
=O(N^{-1}).
\]
This proves the first claim of Proposition~\ref{prop:strat-general-large}. It remains to prove normality for the uncorrected statistic.

\paragraph{A coupling with independent assignments.}
To obtain the normal limit, we approximate the fixed-count estimator
by a sum of independent linkage-cell contributions under Bernoulli
assignment, after subtracting the projection on stratum treatment-count
fluctuations. We construct the two assignment vectors jointly so that
the approximation error is negligible at the root-$N$ scale.
Let $Z_i$ be independent Bernoulli variables with probability $p_r$ for $i\in\cS_r$, and let
\[
K_r=\sum_{i\in\cS_r}(Z_i-p_r).
\]
If $K_r>0$, choose $K_r$ of the treated units uniformly and switch them to control. If $K_r<0$, choose $-K_r$ controls uniformly and switch them to treatment. Make these choices independently across strata, conditional on $\mathbf Z$. The resulting vector $\mathbf D$ has exactly the required fixed-count law: its treated set has size $T_r^{\mathrm{str}}$ and its distribution is invariant to permutations within stratum $r$.

To identify the projection term, write
\[
F(\mathbf d)=W_{\cS}\widehat\tau^{\mathrm{SR}}(w,h;\mathbf d).
\]
Expanding the candidate sums and using $d_i^2=d_i$ gives the Bernoulli-centered representation
\begin{equation}
F(\mathbf Z)-W_{\cS}\tau_{w,\cS}
=
\sum_i\ell_i(Z_i-p_{r(i)})
+
\sum_{\substack{i<j\\g(i)=g(j)}}b_{ij}
(Z_i-p_{r(i)})(Z_j-p_{r(j)}),
\label{eq:strat-polynomial}
\end{equation}
where $r(i)$ denotes the stratum containing $i$. The coefficients are fixed. To bound them under the contribution-moment condition, use orthogonality of the centered Bernoulli monomials. Each coefficient is the expected product of its linkage-cell contribution with the corresponding monomial, divided by the monomial's second moment. The linkage-cell contribution is bounded by $CL_g$, and the second moments are bounded away from zero because treatment probabilities stay away from zero and one. Each coefficient is therefore bounded in absolute value by $CL_{g(i)}$. Bounded linkage-cell size and $\sum_gL_g^4=O(N)$ imply
\[
\sum_i|\ell_i|^4+
\sum_{\substack{i<j\\g(i)=g(j)}}|b_{ij}|^4=O(N).
\]
There are $O(N)$ quadratic terms, and each index belongs to a bounded number of them. Put $\bar\ell_r=(N_r^{\mathrm{str}})^{-1}\sum_{i\in\cS_r}\ell_i$.

\paragraph{Controlling the coupling error.}
We establish the $L^2$ approximation
\begin{equation}
F(\mathbf D)
=
F(\mathbf Z)-\sum_r\bar\ell_rK_r+R_N,
\qquad
\E[R_N^2]=O(N^{1/2})=o(N),
\label{eq:strat-coupling}
\end{equation}
where the expectation includes both stages of the coupling. We control separately the linear part and the two changes in the quadratic part.

For the linear part, the switched units form a simple random sample of size $|K_r|$ from the arm with excess units. This arm contains at least $\min(p_r,1-p_r)N_r^{\mathrm{str}}$ units. Conditional sampling variation of the sum of $\ell_i$ over switched units is therefore bounded by
\[
C\frac{|K_r|}{N_r^{\mathrm{str}}}\sum_{i\in\cS_r}\ell_i^2.
\]
Its expectation is $O(N^{1/2})$, because $\E|K_r|=O(N^{1/2})$ and $\sum_i\ell_i^2=O(N)$. The conditional mean differs from $-K_r\bar\ell_r$ by at most
\[
C\frac{|K_r|}{N_r^{\mathrm{str}}}
\left|\sum_{i\in\cS_r}(\ell_i-\bar\ell_r)(Z_i-p_r)\right|.
\]
The fourth moments of $K_r$ and of the centered sum in this display are $O(N^2)$. Cauchy--Schwarz thus bounds the second moment of the displayed error by $O(1)$. Summing over fixed $R$ proves the required bound for the linear part.

For the quadratic part, write $\delta_i=D_i-Z_i$. Its change is
\[
\sum_i\delta_i
\sum_{\substack{j\ne i\\g(j)=g(i)}}b_{ij}(Z_j-p_{r(j)})
+
\sum_{\substack{i<j\\g(i)=g(j)}}b_{ij}\delta_i\delta_j,
\]
where $b_{ji}=b_{ij}$. For the first sum, use the same conditional simple-random-sampling argument. The sum of squares of its inner coefficients is bounded by $C\sum_{i<j}b_{ij}^2=O(N)$. Its conditional sampling variance therefore has expectation $O(N^{1/2})$. To bound its conditional mean, note that the treated-arm and control-arm sums are each centered under the Bernoulli law: their indicators are $Z_i$ and $1-Z_i$, and each inner coefficient excludes $Z_i$. Each of these two sums has fourth moment $O(N^2)$, so the selected excess-arm sum has the same fourth-moment bound. This follows by expanding the fourth power, using independence across linkage cells and the bounded number of terms within each linkage cell. Multiplication by $K_r/N_r^{\mathrm{str}}$ again gives a second-moment bound of $O(1)$.

Finally, any specified $k\le4$ distinct units are all switched with probability $O(N^{-k/2})$. Conditional inclusion probabilities are bounded by a constant times products of $(|K_r|/N_r^{\mathrm{str}})^{k_r}$, and binomial moments give this unconditional bound. On expanding the square of the double-switch sum, identical edges contribute $O(N^{-1})\sum b_{ij}^2=O(1)$. Edges sharing one vertex contribute at most $O(N^{-3/2})$ times a sum of absolute coefficient products of order $N$. Disjoint edges contribute at most $O(N^{-2})(\sum|b_{ij}|)^2=O(1)$. Thus the double-switch sum has bounded second moment. These bounds establish \eqref{eq:strat-coupling}.

\paragraph{The normal limit.}
Subtracting $\sum_r\bar\ell_rK_r$ from \eqref{eq:strat-polynomial} replaces $\ell_i$ by $\ell_i-\bar\ell_{r(i)}$. The resulting expression is a sum of independent, centered linkage-cell contributions under the Bernoulli law, even when linkage cells cross strata. Its summed fourth moments are $O(N)$, since $\bar\ell_r=O(1)$ and linkage-cell sizes are bounded.

By the variance-limit condition evaluated at $(w,h)$, first-order equivalence, and the $L^2$ bound in \eqref{eq:strat-coupling}, the variance of this independent sum is asymptotic to $W_{\cS}^2V^{\mathrm{SR}}(w,h)/N$, which is of order $N$. Its Lyapunov ratio is therefore $O(N^{-1})$. The independent-array central limit theorem and \eqref{eq:strat-coupling} yield
\[
\sqrt N\{\widehat\tau^{\mathrm{SR}}(w,h)-\tau_{w,\cS}\}
\Rightarrow\mathcal N(0,V^{\mathrm{SR}}(w,h)).
\]
The already established equivalence transfers the limit to the corrected estimator. Setting $w_g=1$, $h_g=0$, and $W_{\cS}=N$ gives \eqref{eq:corr-first-order} and the normality conclusion in Proposition~\ref{prop:fixedcount-asymptotics}. \qed

\subsection{Proofs of Propositions~\ref{prop:strat-variance-a} and \ref{prop:strat-general-variance-a}}

We prove the general result in three steps. First, we express the fixed-count variance as the auxiliary Bernoulli variance net of its projection on the stratum treatment-count fluctuations. Second, we show that the proposed residual sum of squares reproduces this expression plus the heterogeneity gap. We then combine the pieces and obtain the unweighted and perfect-linkage cases.

For any assignment vector $\mathbf d$, write
\[
U_{gr}(\mathbf d)=\sum_{i\in\cC_g\cap\cS_r}(d_i-p_r),
\qquad U_{gr}=U_{gr}(\mathbf D).
\]
Use the auxiliary independent Bernoulli assignments $\mathbf Z$ from the proof of Proposition~\ref{prop:strat-general-large}, so $K_r=\sum_gU_{gr}(\mathbf Z)$. Define
\[
Q_g^w(\mathbf Z)=\sum_s\widetilde H_{gs}^{\,0}(w,h;\mathbf Z),
\qquad F(\mathbf Z)=\sum_gQ_g^w(\mathbf Z).
\]
Then $\E_ZQ_g^w=w_g\theta_g$. Define the coefficient in the linear projection of $F(\mathbf Z)$ on $K_r$ by
\begin{equation}
\beta_r^{\mathrm{pr}}
=\frac{\Cov_Z(F(\mathbf Z),K_r)}{N_r^{\mathrm{str}}p_r(1-p_r)}
=\frac{\sum_g\E_Z\{U_{gr}(\mathbf Z)Q_g^w(\mathbf Z)\}}
{N_r^{\mathrm{str}}p_r(1-p_r)}.
\label{eq:strat-projection-moment}
\end{equation}
The $K_r$ are mutually orthogonal. In \eqref{eq:strat-polynomial}, the quadratic terms are orthogonal to $K_r$, so
\[
\Cov_Z(F(\mathbf Z),K_r)=p_r(1-p_r)\sum_{i\in\cS_r}\ell_i,
\qquad \beta_r^{\mathrm{pr}}=\bar\ell_r.
\]
Thus the coupling in \eqref{eq:strat-coupling} removes precisely $\sum_r\beta_r^{\mathrm{pr}}K_r$ with $o(N)$ mean-squared error. Therefore
\begin{equation}
V_N^{\mathrm{SR}}(w,h)
=\frac{N}{W_{\cS}^2}\left[
\sum_g\Var_Z(Q_g^w)
-\sum_rN_r^{\mathrm{str}}p_r(1-p_r)(\beta_r^{\mathrm{pr}})^2
\right]+o(1).
\label{eq:strat-projection-variance}
\end{equation}

The fixed-count and Bernoulli probabilities of every bounded within-cell assignment pattern have ratio $1+O(N^{-1})$; the same is true for joint patterns from two distinct cells. Thus, for cell functions satisfying $|X_g|\le CL_g^q$, $q\in\{1,2\}$, the envelope bounds give
\[
\Var_D\left(\sum_gX_g\right)
\le C\left\{\sum_gL_g^{2q}+\frac1N\left(\sum_gL_g^q\right)^2\right\}=O(N).
\]
This supplies the laws of large numbers used below for cell functions and their products. In particular, exact unbiasedness gives
\[
\widehat\tau_r(h)-\tau_{\cS_r}=O_p(N^{-1/2})
\]
for each of the fixed number of strata. The decomposition
\[
Q_g^w-w_g\sum_rN_{gr}\tau_{\cS_r}
=(Q_g^w-w_g\theta_g)
+w_g\left(\theta_g-\sum_rN_{gr}\tau_{\cS_r}\right)
\]
separates a mean-zero random term from the deterministic heterogeneity term. Using independence of the $Q_g^w$ under the auxiliary design and then replacing the true centers by their estimates gives
\begin{equation}
\begin{aligned}
&\frac{N}{W_{\cS}^2}\sum_g
\left\{\widetilde H_g(w,h)-w_g\sum_rN_{gr}\widehat\tau_r(h)\right\}^2\\
&\quad=\frac{N}{W_{\cS}^2}\left[
\sum_g\Var_Z(Q_g^w)
+\sum_gw_g^2\left(\theta_g-\sum_rN_{gr}\tau_{\cS_r}\right)^2
\right]+o_p(1).
\end{aligned}
\label{eq:strat-centered-square-limit}
\end{equation}

The same laws of large numbers applied to $U_{gr}\widetilde H_g(w,h)$ imply
\[
\widehat b_r(w,h)=\beta_r^{\mathrm{pr}}+O_p(N^{-1/2}).
\]
They also give, for fixed $r,s$,
\[
\sum_gU_{gr}U_{gs}
=\1\{r=s\}N_r^{\mathrm{str}}p_r(1-p_r)+O_p(N^{1/2}),
\qquad
\sum_gU_{gr}w_gN_{gs}=O_p(N^{1/2}).
\]
By definition,
\[
\sum_gU_{gr}\widetilde H_g(w,h)
=N_r^{\mathrm{str}}p_r(1-p_r)\widehat b_r(w,h).
\]
Thus, on expanding \eqref{eq:strat-general-vhat-a}, the cross-product with $\widetilde H_g(w,h)$ removes two copies of the projection sum and the squared projection restores one. The preceding bounds make the products with the estimated stratum centers and the off-diagonal terms in the Gram matrix of the $U_{gr}$ negligible. Hence
\begin{equation}
\begin{aligned}
\widehat V_a(w,h)
={}&\frac{N}{W_{\cS}^2}\sum_g
\left\{\widetilde H_g(w,h)-w_g\sum_rN_{gr}\widehat\tau_r(h)\right\}^2\\
&-\frac{N}{W_{\cS}^2}\sum_rN_r^{\mathrm{str}}p_r(1-p_r)
\widehat b_r(w,h)^2+o_p(1).
\end{aligned}
\label{eq:strat-residual-expansion}
\end{equation}
Combining \eqref{eq:strat-projection-variance}--\eqref{eq:strat-residual-expansion} proves \eqref{eq:strat-general-vhat-a-limit}. The gap is a sum of squares. If effects are constant within strata, then $\theta_g=\sum_rN_{gr}\tau_{\cS_r}$, so it vanishes.

\paragraph{Special cases.}
Setting $w_g=1$ and $h_g=0$ gives \eqref{eq:strat-vhat-a-limit} and the estimator in \eqref{eq:strat-vhat-a}.
With perfect linkage in that unweighted special case, direct substitution gives $\widehat\tau_r=\overline Y_{1r}-\overline Y_{0r}$ and $\widehat b_r=\overline Y_{1r}/p_r+\overline Y_{0r}/(1-p_r)$. The residuals are the arm-centered outcomes divided by $p_r$ in treatment and by $-(1-p_r)$ in control, yielding the stratified Neyman expression stated in Appendix~\ref{app:fixed-count}. \qed

\subsection{Proofs of Propositions~\ref{prop:small-strata} and \ref{prop:strat-general-small}}

We prove the general result in two steps: use the covariance support to obtain the bias and root-$N$ rate, and then apply a bounded-degree dependency-graph CLT. The unweighted result follows at the end.

\paragraph{Consistency and the bias expansion.}
For a fixed intersection $(g,r)$, a potentially nonzero covariance with $(k,s)$ requires $N_{kr}>0$ and $N_{gs}>0$. There are at most $\overline N^{\mathrm{str}}$ choices of $k$ and $\overline M$ choices of $s$. The number of nonempty intersections is between $N/\min\{\overline N^{\mathrm{str}},\overline M\}$ and $N$. Therefore the covariance matrix has order $N$ potentially nonzero entries among order $N^2$ entries. To bound their sum without a uniform moment bound, use $2|\Cov_D(X,Y)|\leq\Var_D(X)+\Var_D(Y)$. Each contribution appears in only a bounded number of potentially nonzero entries, and Assumption~\ref{ass:strat-general-small} implies $\sum_{g,r}\E_D|\widetilde H_{gr}(w,h)|^2=O(N)$. Consequently,
\[
\Var_D\left\{\sum_{g,r}\widetilde H_{gr}(w,h)\right\}=O(N).
\]
Exact unbiasedness, $W_{\cS}\asymp N$, and Chebyshev's inequality prove the second line of \eqref{eq:strat-general-small-rate}. The uncorrected contributions are fixed bounded multiples of the corrected contributions, so they inherit the fourth-moment bound; the same covariance-support argument gives variance $O(N)$ for their sum. Combining this fact with \eqref{eq:strat-general-bias} proves the first line.

\paragraph{A normal limit with local dependence.}
Take the nonempty intersections as vertices and connect $(g,r)$ and $(k,s)$ whenever linkage cells $g$ and $k$ intersect a common stratum. Collections with no connecting edge depend on disjoint sets of independently randomized strata and are therefore independent. Each linkage cell intersects at most $\overline M$ strata; each such stratum meets at most $\overline N^{\mathrm{str}}$ linkage cells; and each of those cells supplies at most $\overline M$ contributions. The resulting dependency graph thus has degree bounded by $\overline N^{\mathrm{str}}\overline M^2$, independently of $N$.

It remains to verify a moment condition for a normal limit. Jensen's inequality and the average fourth-moment bound give
\[
\begin{aligned}
\sum_{g,r}\E_D|\widetilde H_{gr}(w,h)-w_g\theta_{gr}|^4&=O(N),\\
\sum_{g,r}\E_D|\widetilde H_{gr}(w,h)-w_g\theta_{gr}|^3&=O(N).
\end{aligned}
\]
The total variance is $W_{\cS}^2V_N^{\mathrm{SR}}(w,h)/N\asymp N$. Apply \citet[Theorem~2.4]{chenshao2004} with moment order three: closed graph neighborhoods of radii one, two, and three satisfy its local-dependence condition, and bounded graph degree bounds all required neighborhood sizes independently of $N$. Its normal-approximation error is therefore bounded by a constant times the sum of absolute third moments divided by the total variance to power $3/2$, which is $O(N^{-1/2})$. Hence
\[
\frac{\sum_{g,r}\{\widetilde H_{gr}(w,h)-w_g\theta_{gr}\}}
{\{W_{\cS}^2V_N^{\mathrm{SR}}(w,h)/N\}^{1/2}}
\Rightarrow\mathcal N(0,1).
\]
This proves Proposition~\ref{prop:strat-general-small}. Setting $w_g=1$, $h_g=0$, and $W_{\cS}=N$ gives \eqref{eq:small-strata-bias}--\eqref{eq:small-strata-rate} and \eqref{eq:small-strata-clt}, proving Proposition~\ref{prop:small-strata}. \qed

\subsection{Proofs of Propositions~\ref{prop:strat-variance-b} and \ref{prop:strat-general-variance-b}}

We prove the general result in four steps: derive the cross-stratum product identities, expose the heterogeneity gap under true centering, replace the true center by its estimate, and handle the nonnegative floor. The unweighted result follows at the end.

\begin{lemma}[Cross-stratum products with fixed weights and adjustment]
\label{lem:strat-products}
Under \eqref{eq:strat-nonempty}, for $r\ne s$,
\begin{equation}
\E_D\{\widetilde H_{gr}(w,h)\widetilde H_{gs}(w,h)\}
=2w_g^2\theta_{gr}\theta_{gs}.
\label{eq:strat-general-product-same}
\end{equation}
For $g\ne k$ and $r\ne s$, when all four intersections $(g,r),(k,r),(g,s),(k,s)$ are nonempty,
\begin{align}
\E_D\{\widetilde H_{gr}(w,h)\widetilde H_{ks}(w,h)\}
&=w_gw_k\{\theta_{gr}\theta_{ks}+a_{gr,ks}\theta_{kr}\theta_{gs}\},\notag\\
\E_D\{\widetilde H_{kr}(w,h)\widetilde H_{gs}(w,h)\}
&=w_gw_k\{\theta_{kr}\theta_{gs}+b_{gr,ks}\theta_{gr}\theta_{ks}\}.
\label{eq:strat-general-product-distinct}
\end{align}
In particular, under these respective conditions,
\[
C_{gr,gs}(w,h)=w_g^2\theta_{gr}\theta_{gs},
\qquad
C_{gr,ks}(w,h)=w_gw_ka_{gr,ks}\theta_{kr}\theta_{gs}.
\]
If $a_{gr,ks}b_{gr,ks}<1$, the nonsingular expression in \eqref{eq:strat-general-vhat-four} is unbiased for $C_{gr,ks}(w,h)$. Moreover, $a_{gr,ks}b_{gr,ks}=1$ exactly when both strata are covered exclusively by linkage cells $g$ and $k$.
\end{lemma}

\begin{proof}
We first compute the expected products of contributions from distinct
strata. For distinct cells, we then combine the two product identities
to recover the covariance whenever the resulting system is nonsingular,
and characterize when it is singular.
Use $Q_{gr}$ and $F_{gr}$ from the proof of Proposition~\ref{prop:strat-general-covariance}. In $\widetilde Y_g-h_g=\E_D(\widetilde Y_g-h_g)+\sum_tF_{gt}$, fixed $h_g$ changes only the constant. For distinct strata in the same cell, the only surviving term in the product is $2F_{gr}F_{gs}$, which gives \eqref{eq:strat-general-product-same}.
For distinct cells,
\[
\E_D(Q_{gr}F_{kr})
=-\frac{N_{gr}}{N_r^{\mathrm{str}}-N_{gr}}\theta_{kr}.
\]
Substitution into \eqref{eq:strat-cross-factorization} gives the first identity in \eqref{eq:strat-general-product-distinct}; interchanging the cells gives the second. Taking the second product minus $b_{gr,ks}$ times the first and multiplying by $a_{gr,ks}/(1-a_{gr,ks}b_{gr,ks})$ proves the inversion formula. Finally,
\[
a_{gr,ks}b_{gr,ks}
=\frac{N_{gr}N_{kr}}
{(N_r^{\mathrm{str}}-N_{gr})(N_r^{\mathrm{str}}-N_{kr})}
\frac{N_{gs}N_{ks}}
{(N_s^{\mathrm{str}}-N_{gs})(N_s^{\mathrm{str}}-N_{ks})}.
\]
Each factor is at most one, and equality holds precisely when the two cells exhaust the corresponding stratum.
\end{proof}

Use the groups $b$ defined in \eqref{eq:strat-general-gap-b}: strata covered exclusively by the same two linkage cells belong to the same group, and every remaining stratum forms its own group. Define the true-centered residual
\[
R_{gr}(w,h)=\widetilde H_{gr}(w,h)-w_gN_{gr}\tau_{w,\cS},
\qquad
\delta_{gr}(w)=w_g\{\theta_{gr}-N_{gr}\tau_{w,\cS}\},
\]
and let $\overline V_{b,0}(w,h)$ be $\widehat V_{b,0}(w,h)$ with each $\widehat R_{gr}(w,h)$ replaced by $R_{gr}(w,h)$. Since
\[
\E_D\{R_{gr}(w,h)R_{ks}(w,h)\}
=C_{gr,ks}(w,h)+\delta_{gr}(w)\delta_{ks}(w),
\]
the diagonal and within-stratum residual products contribute the squared stratum sums of the $\delta_{gr}(w)$. The singular cross-stratum residual products supply the ordered cross-products within each group $b$, while Lemma~\ref{lem:strat-products} makes every nonsingular uncentered product unbiased for its covariance. Consequently,
\begin{align}
\E_D\{\overline V_{b,0}(w,h)\}-V_N^{\mathrm{SR}}(w,h)
&=\frac{N}{W_{\cS}^2}\sum_b
\left\{\sum_{r\in b}\sum_g\delta_{gr}(w)\right\}^2\notag\\
&=\Delta_{b,N}(w)\ge0.
\label{eq:strat-small-gap-squares}
\end{align}
Thus the ordered-pair convention makes the relevant diagonal and cross-products form a sum of squares.

Each component contains only $O(N)$ products with uniformly bounded coefficients; positive values of $1-a_{gr,ks}b_{gr,ks}$ are bounded away from zero because the relevant counts are bounded integers. Joining two product terms when their operand-stratum sets overlap gives a bounded-degree dependency graph. Every contribution enters only a bounded number of products. Using $2X^2Y^2\leq X^4+Y^4$, the average fourth-moment condition bounds the sum of their second moments by $O(N)$; the same is true for residual products because the condition bounds the average fourth moments of the means and the overall target. The bounded-degree covariance bound then yields
\[
\overline V_{b,0}(w,h)-\E_D\{\overline V_{b,0}(w,h)\}=O_p(N^{-1/2}).
\]
It remains to replace the true center by its estimate. Let $e_N=\widehat\tau^{\mathrm{SR,corr}}(w,h)-\tau_{w,\cS}$. Every residual occurs in a bounded number of products, so, for a constant $C$ independent of $N$,
\[
\left|\widehat V_{b,0}(w,h)-\overline V_{b,0}(w,h)\right|
\le C\frac{N}{W_{\cS}^2}
\left\{|e_N|\sum_{g,r}|R_{gr}(w,h)|+Ne_N^2\right\}
=O_p(N^{-1/2}),
\]
because $N/W_{\cS}^2=O(N^{-1})$, $N^{-1}\sum_{g,r}|R_{gr}(w,h)|=O_p(1)$, and Proposition~\ref{prop:strat-general-small} gives $e_N=O_p(N^{-1/2})$. Combining these bounds with \eqref{eq:strat-small-gap-squares} proves \eqref{eq:strat-general-vhat-b-limit}.

Finally, write $\widehat X_N=\widehat V^{(3)}(w,h)+\widehat V^{(4)}(w,h)$. Since
\[
\widehat V^{(1)}(w,h)+\widehat V^{(2)}(w,h)
=\frac{N}{W_{\cS}^2}\sum_r
\left\{\sum_g\widehat R_{gr}(w,h)\right\}^2\ge0,
\]
$\widehat V_b(w,h)$ is nonnegative and
\[
\widehat V_b(w,h)-\widehat V_{b,0}(w,h)=(-\widehat X_N)_+.
\]
Because it cannot be smaller than $\widehat V_{b,0}(w,h)$, the preceding expansion also establishes asymptotic conservativeness.
Under homogeneous effects, $\delta_{gr}(w)=0$ and $\widehat X_N=X_N+O_p(N^{-1/2})$, where $X_N$ is the corresponding normalized sum of cross-stratum covariances. Those covariances are nonnegative because $w_g>0$, so $X_N\ge0$ and $(-\widehat X_N)_+=O_p(N^{-1/2})$. This proves homogeneous-effect consistency. Setting $w_g=1$, $h_g=0$, and $W_{\cS}=N$ gives Proposition~\ref{prop:strat-variance-b}. \qed

\subsection{Proof of Proposition~\ref{prop:strat-estimated-adjustment}}

We show that replacing the fixed adjustment $h^*$ by its estimate changes
the point estimators by $O_p(N^{-1})$ and the variance estimators by
$O_p(N^{-1/2})$. These bounds preserve the root-$N$ limits of the point
estimators and the probability limits of the variance estimators.
We first establish the point-estimator bound and then treat the two
variance estimators.

\paragraph{Point estimators.}
Write
\[
\dot h_g^*=\left.\frac{\partial h(x_g;\beta)}{\partial\beta}\right|_{\beta=\beta^*}.
\]
The smoothness condition gives the uniform expansion
\begin{equation}
\widehat h_g-h_g^*=\dot h_g^{*\prime}(\widehat\beta-\beta^*)+r_g,
\qquad
\sup_g|r_g|\le C\|\widehat\beta-\beta^*\|^2=O_p(N^{-1}).
\label{eq:strat-adjustment-taylor}
\end{equation}
Recall that $U_{gr}=T_{gr}-p_rN_{gr}$. Let
\[
d_{gr}^{\mathrm{corr}}
=\frac{N_r^{\mathrm{str}}-1}{N_r^{\mathrm{str}}-N_{gr}},
\qquad d_{gr}^{0}=1,
\]
and define
\[
G_N^q=\sum_{g,r}w_gd_{gr}^q
\frac{U_{gr}}{p_r(1-p_r)}\dot h_g^*,
\qquad q\in\{\mathrm{corr},0\}.
\]
For $i\in\cS_r$, the coefficient on $D_i-p_r$ in this sum is
\[
\alpha_{ir}^q=\frac{w_{g(i)}d_{g(i)r}^q}{p_r(1-p_r)}\dot h_{g(i)}^*.
\]
These coefficients are uniformly bounded under either asymptotic regime. The complete-randomization covariance formula therefore gives, with $\bar\alpha_r^q=(N_r^{\mathrm{str}})^{-1}\sum_{i\in\cS_r}\alpha_{ir}^q$,
\[
\E_D\|G_N^q\|^2
=\sum_rp_r(1-p_r)\frac{N_r^{\mathrm{str}}}{N_r^{\mathrm{str}}-1}
\sum_{i\in\cS_r}\|\alpha_{ir}^q-\bar\alpha_r^q\|^2
=O(N).
\]
Thus $G_N^q=O_p(\sqrt N)$. Substitution of \eqref{eq:strat-adjustment-taylor} into \eqref{eq:strat-general-estimators} yields
\[
\widehat\tau^{\mathrm{SR},q}(w,\widehat h)
-\widehat\tau^{\mathrm{SR},q}(w,h^*)
=-\frac1{W_{\cS}}(G_N^q)'(\widehat\beta-\beta^*)
+O_p(\|\widehat\beta-\beta^*\|^2)
=O_p(N^{-1}).
\]
For the Taylor remainder, we used $W_{\cS}\asymp N$ and $\sum_{g,r}|w_gd_{gr}^qU_{gr}/\{p_r(1-p_r)\}|=O(N)$. This proves \eqref{eq:strat-estimated-adjustment-rate} and transfers the fixed-$h^*$ root-$N$ limits by Slutsky's theorem.

\paragraph{Variance estimators.}
Let
\[
\Gamma_{gr}=\widetilde H_{gr}(w,\widehat h)-\widetilde H_{gr}(w,h^*).
\]
The uniform Taylor expansion and bounded contribution coefficients imply
\[
\sup_{g,r}|\Gamma_{gr}|=O_p(N^{-1/2}),
\qquad
\sum_{g,r}\Gamma_{gr}^2=O_p(1).
\]
For $\widehat V_a$, let $A_g(h)$ denote the bracketed cell residual in \eqref{eq:strat-general-vhat-a}. Since the number of strata is fixed, the definitions in \eqref{eq:strat-general-large-centers}--\eqref{eq:strat-general-large-projection} and the preceding bounds give
\[
\max_r|\widehat\tau_r(\widehat h)-\widehat\tau_r(h^*)|=O_p(N^{-1/2}),
\qquad
\max_r|\widehat b_r(w,\widehat h)-\widehat b_r(w,h^*)|=O_p(N^{-1/2}).
\]
Bounded cell size then implies
\[
\sum_g\{A_g(\widehat h)-A_g(h^*)\}^2=O_p(1),
\qquad
\sum_gA_g(h^*)^2=O_p(N).
\]
Cauchy--Schwarz and $N/W_{\cS}^2=O(N^{-1})$ therefore give
\[
\widehat V_a(w,\widehat h)-\widehat V_a(w,h^*)=O_p(N^{-1/2}).
\]
For $\widehat V_b$, the residual difference equals $\Gamma_{gr}$ minus $w_gN_{gr}$ times the point-estimator difference in \eqref{eq:strat-estimated-adjustment-rate}. Hence
\[
\sum_{g,r}\{\widehat R_{gr}(w,\widehat h)-\widehat R_{gr}(w,h^*)\}^2=O_p(1).
\]
Each component contains only $O(N)$ products with bounded coefficients, while the fixed-$h^*$ operands have summed squares $O_p(N)$. Another application of Cauchy--Schwarz gives an $O_p(N^{-1/2})$ change in each component of $\widehat V_{b,0}$ and in $\widehat V^{(1)}+\widehat V^{(2)}$. Since the maximum map is Lipschitz, the same rate holds for $\widehat V_b$. This proves \eqref{eq:strat-estimated-adjustment-variance}. \qed

\section{Alternative methods under ambiguous linkage}
\label{app:alternative-methods}

We use three settings to compare our estimators with selected linkage and post-linkage procedures. The examples are deliberately constructed to satisfy our assignment and candidate-set conditions while departing from particular data-structure or modeling restrictions maintained by some existing methods. We distinguish three layers of comparison: (i) direct implementations based on specifications presented in the original papers; (ii) alternative specifications within the original model classes that use additional analyst choices not prescribed in those papers; and (iii) extensions that would require broadening the original frameworks or developing guarantees outside their stated conditions. We focus on the first two layers and do not attempt to exhaust all admissible specifications within those frameworks or  possible extensions of the third type.

\subsection{Data-generating processes and simulations}

For all examples, the initial marginal probability of assigning experimental unit \(i\) its true outcome record is \(q_i=1/M_{g(i)}\). Outcome-informed methods may subsequently update these probabilities using outcomes, as specified by the method. In all three settings, \(\varepsilon_i\overset{\mathrm{iid}}{\sim}\mathcal N(0,1)\) is drawn before treatment assignment and held fixed across unit \(i\)'s two potential outcomes and across Monte Carlo replications. The same realized potential outcomes are used for every method. With bounded linkage-cell sizes and the fixed type-specific means and effects below, these Normal draws satisfy the average moment conditions almost surely along the sequence of finite populations; a uniform bound on each realized outcome is not required.

\setcounter{linkageexample}{0}
\refstepcounter{linkageexample}\label{ex:alt-main}
\paragraph{Example \thelinkageexample: effects vary with linkage ambiguity.}
All administrative records belong to experimental units. There are
$2J$ singleton linkage cells and $J$ doubleton linkage cells, with
independent Bernoulli$(1/2)$ treatment assignment:
\begin{center}
\small
\begin{tabular}{@{}lrrrrr@{}}
\toprule
Type & Linkage cells & $M_g=N_g^{\mathrm{cell}}$
& Experimental units & $Y_i(0)$ & $\tau_i$ \\
\midrule
$L$ & $2J$ & $1$ & $2J$ & $1+\varepsilon_i$ & $0$ \\
$H$ & $J$ & $2$ & $2J$ & $1+\varepsilon_i$ & $1$ \\
\bottomrule
\end{tabular}
\end{center}
Here $M=N=4J$ and $\sate=1/2$. Uniform allocation gives $q_i=1$
in type $L$ and $q_i=1/2$ in type $H$. Treatment effects therefore vary
with linkage ambiguity, and unambiguous matches contain only zero-effect
units.

\refstepcounter{linkageexample}\label{ex:alt-outsiders}
\paragraph{Example \thelinkageexample: nonparticipants among the candidates.}
Each of $J$ linkage cells contains two experimental units and one
nonparticipant:
\begin{center}
\small
\begin{tabular}{@{}lrrr@{}}
\toprule
Record & $Y_i(0)$ & $Y_i(1)$ & $\tau_i$ \\
\midrule
Participant $A$ & $\varepsilon_i$ & $\varepsilon_i$ & $0$ \\
Participant $B$ & $2+\varepsilon_i$ & $3+\varepsilon_i$ & $1$ \\
Nonparticipant & $10+\varepsilon_i$ & $10+\varepsilon_i$ & $0$ \\
\bottomrule
\end{tabular}
\end{center}
The record types are not observed by the analyst. The two participants
receive independent Bernoulli$(1/2)$ assignments, while the
nonparticipant is never treated. Thus $M=3J$, $N=2J$, $M_g=3$,
$N_g^{\mathrm{cell}}=2$, $q_i=1/3$, and $\sate=1/2$. False candidates
can now come from outside the experimental population and have a
different outcome distribution. Since all linkage cells have the same
size, conditioning on cell size does not distinguish the participant
types. 

\Needspace{13\baselineskip}
\refstepcounter{linkageexample}\label{ex:alt-stratified}
\paragraph{Example \thelinkageexample: fixed treatment counts in small strata.}
All administrative records belong to experimental units. There are
$J$ strata of each type below, and every linkage cell lies within a
stratum. Exactly half the units in each stratum receive treatment,
with independent randomization across strata:
\begin{center}
\small
\begin{tabular}{@{}lrrrlrr@{}}
\toprule
Type & Strata & $N_r^{\mathrm{str}}$ & $T_r^{\mathrm{str}}$
& Linkage cells per stratum & $Y_i(0)$ & $\tau_i$ \\
\midrule
$A$ & $J$ & $2$ & $1$ & Two singletons & $1+\varepsilon_i$ & $0$ \\
$B$ & $J$ & $4$ & $2$ & Two doubletons & $1+\varepsilon_i$ & $1$ \\
\bottomrule
\end{tabular}
\end{center}
Here $M=N=6J$ and $\sate=2/3$. Fixing treatment counts makes a
false candidate's assignment, and hence its outcome, depend on the
recipient's assignment. Our fixed-count corrected estimator is
unbiased and consistent for $2/3$.

\paragraph{Monte Carlo comparison.} We run Monte Carlo simulations with $J=5{,}000$ and $10{,}000$ replications for each example. We draw the potential outcomes once per example and hold that finite population fixed throughout the experiment. Each replication draws a new treatment assignment, shared by all methods, and the linkage or imputation draws required by each method. Table~\ref{tab:alt-methods} reports bias relative to the fixed-population SATE and the empirical standard error across replications. For methods that supply links rather than treatment-effect estimates, we use completed-file mean contrasts. 
We evaluate Bayesian linkage procedures using posterior effect summaries. 
Detailed descriptions of the implemented estimators, including specification choices and computational settings, are provided alongside the code in the replication archive \href{https://www.dropbox.com/scl/fi/rhwdelrw0z4t9fpay33jj/Replication.zip?rlkey=p7q6sssjygu2odyu392grz84z&dl=0}{(link)}.

\begin{table}[p]
\centering
\scriptsize
\setlength{\tabcolsep}{3pt}
\renewcommand{\arraystretch}{1.08}
\begin{tabular}{@{}>{\raggedright\arraybackslash}p{3.15in}*{6}{c}@{}}
\toprule
& \multicolumn{2}{c}{Example 1} & \multicolumn{2}{c}{Example 2} & \multicolumn{2}{c}{Example 3} \\
\cmidrule(lr){2-3}\cmidrule(lr){4-5}\cmidrule(l){6-7}
Estimator & Bias & \shortstack{Standard\\error} & Bias & \shortstack{Standard\\error} & Bias & \shortstack{Standard\\error} \\
\midrule
\multicolumn{7}{@{}l}{\textbf{Panel 1: Our estimators (fixed-count corrected in Example 3) }}  \\
Baseline estimator 
  & 0.0002 & 0.0369 & 0.0042 & 0.2469 & 0.0000 & 0.0221 \\
Weighted estimator
  & -0.1664 & 0.0295 & 0.0042 & 0.2469 & -0.1667 & 0.0186 \\
Adjusted estimator: intercept only
  & -0.0003 & 0.0236 & 0.0001 & 0.0351 & 0.0000 & 0.0221 \\
Adjusted estimator: intercept and $M_g$
  & -0.0001 & 0.0186 & 0.0001 & 0.0351 & 0.0000 & 0.0221 \\
Weighted-and-adjusted estimator: intercept and $M_g$
  & -0.1667 & 0.0170 & 0.0001 & 0.0351 & -0.1667 & 0.0186 \\
\midrule
\multicolumn{7}{@{}l}{\textbf{Panel 2: Direct implementations of procedures proposed in other papers}} \\
\citet{lahiri2005}: transformed regression
  & -0.1667 & 0.0173 & \multicolumn{2}{c}{N/A} & -0.2667 & 0.0160 \\
\citet{chambers2009}: ratio-adjusted regression
  & -0.1668 & 0.0211 & -0.1648 & 0.3152 & -0.2669 & 0.0236 \\
\citet{chipperfield2019}: exchangeable working matrix
  & -0.1667 & 0.0209 & \multicolumn{2}{c}{N/A} & -0.3336 & 0.0203 \\
\citet{kimchambers2012a}: corrected estimating equation
  & -0.1668 & 0.0209 & -0.1651 & 0.1739 & -0.2669 & 0.0233 \\
\citet{kimchambers2012b}: corrected estimating equation
  & -0.1669 & 0.0208 & -0.1660 & 0.1732 & -0.2667 & 0.0234 \\
\citet{hanlahiri2019}: linear estimating-equation framework
  & -0.1667 & 0.0173 & \multicolumn{2}{c}{N/A} & -0.2667 & 0.0160 \\
\citet{chamberssilva2020}: pooled linear equation
  & -0.1666 & 0.0208 & \multicolumn{2}{c}{N/A} & -0.2669 & 0.0233 \\
\citet{hofzwinderman2012}: unfiltered all-pairs WLS
  & -0.2501 & 0.0133 & -0.3333 & 0.0117 & -0.4445 & 0.0096 \\
\citet{hofzwinderman2015}: pair-mixture pseudo-likelihood
  & -0.1616 & 0.0178 & \multicolumn{2}{c}{No finite fit} & -0.3294 & 0.0144 \\
\citet{sadinle2017}: linkage and completed-file contrast
  & -0.2501 & 0.0133 & -0.3333 & 0.0117 & -0.4445 & 0.0096 \\
\citet{steorts2015}: linkage and completed-file contrast
  & -0.2501 & 0.0133 & -0.3333 & 0.0117 & -0.4445 & 0.0096 \\
\citet{shan2021}: treatment-label imputation
  & -0.2501 & 0.0133 & -0.3750 & 0.0088 & -0.4445 & 0.0096 \\
\citet{zhangtuoto2021}: pseudo-OLS correction
  & -0.1668 & 0.0210 & 0.0020 & 0.2589 & -0.3333 & 0.0202 \\
\citet{slawskietal2025}: response mixture
  & -0.2910 & 0.0198 & \multicolumn{2}{c}{N/A} & -0.4361 & 0.0181 \\
\citet{bukkeslawski2025}: donor mixture
  & -0.2911 & 0.0201 & \multicolumn{2}{c}{N/A} & -0.4363 & 0.0182 \\
\citet{hureiter2027}: post-linkage multiple imputation
  & -0.2500 & 0.0140 & \multicolumn{2}{c}{No finite fit} & -0.4428 & 0.0106 \\
\citet{goldstein2012}: standard multiple imputation
  & -0.4999 & 0.0210 & \multicolumn{2}{c}{N/A} & -0.6667 & 0.0210 \\
\citet{goldstein2012}: prior-informed multiple imputation
  & -0.3711 & 0.0167 & \multicolumn{2}{c}{N/A} & -0.5521 & 0.0147 \\
\citet{wortman2018}: MEV/MEDOV linkage selection
  & -0.2501 & 0.0154 & \multicolumn{2}{c}{N/A} & -0.4445 & 0.0122 \\
\midrule
\multicolumn{7}{@{}l}{\textbf{Panel 3: Within-framework specifications using appropriate analyst choices}} \\
Linkage-adjusted regression with control means and treatment effects varying by linkage-cell type
  & -0.0002 & 0.0244 & -0.1660 & 0.1732 & 0.0000 & 0.0333 \\
Post-linkage mixture with type-specific response and false-match models
  & -0.0021 & 0.0231 & \multicolumn{2}{c}{N/A} & -0.2143 & 0.0250 \\
\bottomrule
\end{tabular}
\caption{Monte Carlo bias and empirical standard errors.}
\label{tab:alt-methods}
\begin{minipage}{\linewidth}
\footnotesize
\textit{Note:} Results use $J=5{,}000$ and $R=10{,}000$ assignments to one fixed potential-outcome population per example. Bias is the mean estimate minus the SATE ($1/2$, $1/2$, and $2/3$); the empirical standard error is the sample standard deviation across replications. Panel 1 uses fixed-count corrections in Example 3. Panel 3 uses the interacted and standardized Kim--Chambers (2012b) regression ($G=X'$) and the type-specific integrated Slawski et al.\ mixture. N/A denotes inapplicability of the selected specification. Entries labeled ``No finite fit'' indicate that the specified likelihood is unbounded and therefore has no finite maximizer.
\citet{sadinle2017} and \citet{steorts2015} develop Bayesian linkage models rather than treatment-effect estimators. For these methods, we report an analyst-specified downstream analysis: the posterior average of the treated-minus-control mean difference in completed datasets. 
\end{minipage}
\end{table}

\subsection{Results}

Panel 1 of Table~\ref{tab:alt-methods} presents our proposed estimators. The unweighted options recover the SATE in all examples. The weighted ones fail to recover the SATE in Examples 1 and 3.\footnote{In Example 2, there is no variation in $M_g$, so the weighted and unweighted estimators are the same.} Consistent with our theoretical results from Section~\ref{sec:asymptotics}, those estimators may be consistent for a different weighted average of treatment effects in exchange for greater precision. Panel 2 of Table~\ref{tab:alt-methods} presents direct implementations of procedures proposed in other papers.  Overall, those implementations fail  to recover the SATE in at least one example. The reasons include (i) common treatment slopes and
pooled linkage corrections that incorrectly weight heterogeneous
effects, (ii) extrapolation from unambiguous matches that are not
representative of the experimental sample, (iii) mixing treatment
identities through incorrect matches, (iv) false candidates drawn
from a nonexperimental population with a different outcome
distribution, and (v) dependence between false-candidate outcomes
and the recipient's assignment under fixed treatment
counts.\footnote{We omit \citet{zhangtuoto2024} from the table
because their method addresses categorical outcomes. In the
binary version of Example 1, with $Y_i(0)=0$ and
$Y_i(1)=\mathbf 1\{M_g=2\}$, their pooled linkage-rate correction
yields an estimated risk difference converging to $1/3$, rather
than the SATE of $1/2$.}

\subsection{Alternative specifications within the original frameworks}

Some specifications within the frameworks considered above can recover the SATE when the analyst makes appropriate choices of regressors, response models, linkage-error components, or standardization. Panel 3 of  Table~\ref{tab:alt-methods} presents two such specifications. In each case, the particular choices that recover the SATE in some of our examples are permitted by the original framework but are not explicitly recommended or prescribed in the corresponding papers.

\paragraph{Linkage-adjusted linear regression methods.}
In Examples 1 and 3, the linkage-adjusted linear regression
approaches can recover the SATE when both the conditional mean
under control and the treatment effect are allowed to vary by
linkage-cell type
\citep{lahiri2005,chambers2009,kimchambers2012a,
kimchambers2012b,hanlahiri2019,chamberssilva2020}.
For these examples, consider
\[
Y_i=\alpha+\beta D_i+\gamma M_{g(i)}
+\delta D_i\bigl(M_{g(i)}-\bar M\bigr)+u_i,
\qquad
\bar M=\frac1N\sum_{i\in\cS}M_{g(i)}.
\]
The specification must be fitted using the appropriate linkage
transformation or corrected estimating equation, with the
linkage adjustment applied to all regressors, including the
interaction. Centering the interaction makes $\beta$ the average
fitted treatment effect across experimental units, so the
treatment coefficient itself recovers the SATE in both examples.
Merely controlling for linkage-cell size without the interaction
is insufficient. Although this specification lies within the
original regression frameworks, the choice of the relevant
interaction and the standardization needed to target the SATE
are not prescribed by those papers.

Those alternatives, however, do not generally solve the problem in Example 2. Some methods do not apply to settings with $N<M$.
The original formulation of \citet{lahiri2005} assumes that the
outcome and covariate files cover the same individuals, requiring
$N=M$ in our notation, and therefore does not apply directly.
The framework of \citet{hanlahiri2019} allows unequal file sizes,
but assumes that the units represented in the outcome file are
contained in the covariate register. This condition also fails
for the two files in Example 2, where some candidate outcomes
belong to nonparticipants without counterparts in the
experimental covariate file.

Other methods allow for $N<M$. However, their direct
implementations do not recover the SATE in Example~2.
The estimating equation in \citet{chambers2009} separates the expected outcome
contribution of nonparticipants from that of participants, but
the former contribution is unobserved. Section 6.2 proposes
approximating it using either the mean outcome of records left
unmatched or a probability-weighted combination of linked and
unmatched outcomes. In Example 2, both approximations, when used
with the linear estimator in equation (86) of that paper, yield
a treatment coefficient converging to $1/3$, rather than the
SATE of $1/2$. Records left unmatched are not necessarily
nonparticipants, and these approximations retain variation
induced by participants' treatment assignments.

The sample-to-register formulations of
\citet{kimchambers2012a,kimchambers2012b} use weighted sample
quantities to approximate population totals. In Example 2,
applying their sample-expansion corrections with weights
$M_g/N_g^{\mathrm{cell}}=3/2$ likewise yields a treatment
coefficient converging to $1/3$. The required representativeness
does not hold for the experimental participants relative to all
candidate records. Interactions with linkage-cell size alone
cannot resolve this issue, since all linkage cells have the
same size.

\paragraph{Post-linkage mixture models.}
The mixture framework of \citet{slawskietal2025}, particularly
its informative-linkage extension in \citet{bukkeslawski2025},
can accommodate the heterogeneity in Example 1 through correctly
specified type-specific response and false-match models. However, this
modeling flexibility does not extend automatically to the other
settings considered here. In Example 3, fixed treatment counts
make a false candidate's outcome depend on the recipient's
assignment, so validity would require a design-aware false-match
component. In Example 2, nonparticipant candidates may have a
different outcome distribution from the experimental population,
requiring additional information or modeling of that population.
Neither extension is supplied by the original papers.

\section{Primitive conditions for weighting by inverse linkage-cell size}
\label{app:weighting-primitives}

This section provides a simple stochastic benchmark under which the
asymptotic proportionality condition in
\eqref{eq:homoskedastic-benchmark} holds almost surely. The purpose is
not to impose a sampling model in the main analysis, but to illustrate
primitive conditions that can rationalize weighting by inverse linkage-cell size.

Consider an i.i.d. sequence of linkage cells indexed by $g$. Let $M_g$
have fixed finite support, with
\[
\Pr(M_g=m)>0
\]
for every $m$ in its support. Suppose each represented linkage cell
contains exactly one experimental unit, so that $N_g^{\mathrm{cell}}=1$. Conditional on
$M_g=m$, let the baseline potential outcomes
\[
Y_{g1}(0),\ldots,Y_{gm}(0)
\]
be i.i.d. draws from a common distribution that does not depend on $m$,
with
\[
\E[Y_{gj}(0)]=\mu,
\qquad
\Var\{Y_{gj}(0)\}=\sigma^2\in(0,\infty).
\]
Finally, suppose treatment effects are homogeneous:
\[
Y_{gi}(1)-Y_{gi}(0)=\tau
\]
for every experimental unit.

For a linkage cell of size $m$, define
\[
B_g=\sum_{j=1}^{m}Y_{gj}(0).
\]
When $N_g^{\mathrm{cell}}=1$, the linkage-cell-level variance contribution under a fixed
adjustment $h(m)$ is
\[
V_g\{h(m)\}
=
\frac{
\{B_g-h(m)+(1-p)\tau\}^2
}{
p(1-p)
}.
\]
The variance-minimizing adjustment conditional on linkage-cell size is
\[
h^*(m)
=
\E[B_g\mid M_g=m]+(1-p)\tau.
\]
Hence
\[
V_g\{h^*(m)\}
=
\frac{
\{B_g-\E[B_g\mid M_g=m]\}^2
}{
p(1-p)
}.
\]

\begin{proposition}[A primitive benchmark for weighting by inverse linkage-cell size]
\label{prop:inverse-k-primitives}
Under the conditions above, for every $m$ in the support of $M_g$,
\[
\frac{1}{N_m}
\sum_{g:M_g=m}
V_g\{h^*(m)\}
\overset{a.s.}{\longrightarrow}
\frac{m\sigma^2}{p(1-p)},
\]
where $N_m$ is the number of represented linkage cells with $M_g=m$.
Consequently, the asymptotic proportionality condition in
\eqref{eq:homoskedastic-benchmark} holds almost surely.
\end{proposition}

\begin{proof}
We compute the expected variance contribution within each linkage-cell-size
category and then apply the strong law of large numbers to obtain the
stated limit.
Conditional on $M_g=m$,
\[
\Var(B_g\mid M_g=m)
=
\Var\left\{\sum_{j=1}^{m}Y_{gj}(0)\right\}
=
m\sigma^2.
\]
Therefore,
\[
\E\left[
V_g\{h^*(m)\}
\mid M_g=m
\right]
=
\frac{m\sigma^2}{p(1-p)}.
\]
Because linkage cells are i.i.d. and $\Pr(M_g=m)>0$, the number of linkage cells with
$M_g=m$ diverges almost surely. The strong law of large numbers applied
within each linkage-cell-size category then gives
\[
\frac{1}{N_m}
\sum_{g:M_g=m}
V_g\{h^*(m)\}
\overset{a.s.}{\longrightarrow}
\frac{m\sigma^2}{p(1-p)}.
\]
\end{proof}

The proportionality result relies on the benchmark restrictions above.
If multiple experimental units can occupy the same linkage cell,
$V_g(h)$ generally contains additional terms depending on $N_g^{\mathrm{cell}}$. Likewise,
heterogeneous treatment effects introduce additional variation that need
not scale proportionally with $M_g$. The main results of the paper do not
require these restrictions; they are used here only to provide primitive
conditions under which the $1/M_g$ weighting rule emerges naturally.


\begin{thebibliography}{99}

\bibitem[Bukke and Slawski(2025)]{bukkeslawski2025}
Bukke, Priyanjali, and Martin Slawski. 2025. ``Relaxing the Assumption of Strongly Non-Informative Linkage Error in Secondary Regression Analysis of Linked Files.'' \emph{arXiv preprint arXiv:2510.17553}. \url{https://doi.org/10.48550/arXiv.2510.17553}.

\bibitem[Cameron et al.(2011)Cameron, Gelbach, and Miller]{camerongelbachmiller2011}
Cameron, A. Colin, Jonah B. Gelbach, and Douglas L. Miller. 2011. ``Robust Inference with Multiway Clustering.'' \emph{Journal of Business \& Economic Statistics} 29(2): 238--249. \url{https://doi.org/10.1198/jbes.2010.07136}.

\bibitem[Chambers(2009)]{chambers2009}
Chambers, Ray. 2009. ``Regression Analysis of Probability-Linked Data.'' \emph{Official Statistics Research Series} 4. Statistics New Zealand. \url{https://natlib.govt.nz/records/20918954}.

\bibitem[Chambers and Diniz da Silva(2020)]{chamberssilva2020}
Chambers, Ray, and Andrea Diniz da Silva. 2020. ``Improved Secondary Analysis of Linked Data: A Framework and an Illustration.'' \emph{Journal of the Royal Statistical Society Series A: Statistics in Society} 183(1): 37--59. \url{https://doi.org/10.1111/rssa.12477}.

\bibitem[Chen and Shao(2004)]{chenshao2004}
Chen, Louis H. Y., and Qi-Man Shao. 2004. ``Normal Approximation under Local Dependence.'' \emph{The Annals of Probability} 32(3A): 1985--2028. \url{https://doi.org/10.1214/009117904000000450}.

\bibitem[Chipperfield(2019)]{chipperfield2019}
Chipperfield, James. 2019. ``A Weighting Approach to Making Inference with Probabilistically Linked Data.'' \emph{Statistica Neerlandica} 73(3): 333--350. \url{https://doi.org/10.1111/stan.12172}.

\bibitem[Dasylva(2018)]{dasylva2018}
Dasylva, Abel. 2018. ``Design-Based Estimation with Record-Linked Administrative Files and a Clerical Review Sample.'' \emph{Journal of Official Statistics} 34(1): 41--54. \url{https://doi.org/10.1515/jos-2018-0003}.

\bibitem[Goldstein et al.(2012)Goldstein, Harron, and Wade]{goldstein2012}
Goldstein, Harvey, Katie Harron, and Angie Wade. 2012. ``The Analysis of Record-Linked Data Using Multiple Imputation with Data Value Priors.'' \emph{Statistics in Medicine} 31(28): 3481--3493. \url{https://doi.org/10.1002/sim.5508}.

\bibitem[Guha et al.(2022)Guha, Reiter, and Mercatanti]{guha2022}
Guha, Sharmistha, Jerome P. Reiter, and Andrea Mercatanti. 2022. ``Bayesian Causal Inference with Bipartite Record Linkage.'' \emph{Bayesian Analysis} 17(4): 1275--1299. \url{https://doi.org/10.1214/21-BA1297}.

\bibitem[Gutman et al.(2013)Gutman, Afendulis, and Zaslavsky]{gutman2013}
Gutman, Roee, Christopher C. Afendulis, and Alan M. Zaslavsky. 2013. ``A Bayesian Procedure for File Linking to Analyze End-of-Life Medical Costs.'' \emph{Journal of the American Statistical Association} 108(501): 34--47. \url{https://doi.org/10.1080/01621459.2012.726889}.

\bibitem[Gutman et al.(2016)Gutman, Sammartino, Green, and Montague]{gutman2016}
Gutman, R., C. J. Sammartino, T. C. Green, and B. T. Montague. 2016. ``Error Adjustments for File Linking Methods Using Encrypted Unique Client Identifier (eUCI) with Application to Recently Released Prisoners Who Are HIV+.'' \emph{Statistics in Medicine} 35(1): 115--129. \url{https://doi.org/10.1002/sim.6586}.

\bibitem[Han and Lahiri(2019)]{hanlahiri2019}
Han, Ying, and Partha Lahiri. 2019. ``Statistical Analysis with Linked Data.'' \emph{International Statistical Review} 87(S1): S139--S157. \url{https://doi.org/10.1111/insr.12295}.

\bibitem[Hof and Zwinderman(2012)]{hofzwinderman2012}
Hof, M. H. P., and A. H. Zwinderman. 2012. ``Methods for Analyzing Data from Probabilistic Linkage Strategies Based on Partially Identifying Variables.'' \emph{Statistics in Medicine} 31(30): 4231--4242. \url{https://doi.org/10.1002/sim.5498}.

\bibitem[Hof and Zwinderman(2015)]{hofzwinderman2015}
Hof, M. H. P., and A. H. Zwinderman. 2015. ``A Mixture Model for the Analysis of Data Derived from Record Linkage.'' \emph{Statistics in Medicine} 34(1): 74--92. \url{https://doi.org/10.1002/sim.6315}.

\bibitem[Horvitz and Thompson(1952)]{horvitz1952}
Horvitz, D. G., and D. J. Thompson. 1952. ``A Generalization of Sampling Without Replacement from a Finite Universe.'' \emph{Journal of the American Statistical Association} 47(260): 663--685. \url{https://doi.org/10.1080/01621459.1952.10483446}.

\bibitem[Hu and Reiter(2027)]{hureiter2027}
Hu, Xueyan, and Jerome P. Reiter. 2027. ``Regression Analysis after Bipartite Bayesian Record Linkage.'' \emph{Computational Statistics \& Data Analysis} 225: 108461. \url{https://doi.org/10.1016/j.csda.2026.108461}.

\bibitem[Kamat and Gutman(2026)]{kamatgutman2026}
Kamat, Gauri, and Roee Gutman. 2026. ``Analysis of Linked Files: A Missing Data Perspective.'' \emph{Statistical Science} 41(1): 28--48. \url{https://doi.org/10.1214/24-STS939}.

\bibitem[Kamat et al.(2023)Kamat, Shan, and Gutman]{kamatshan2023}
Kamat, Gauri, Mingyang Shan, and Roee Gutman. 2023. ``Bayesian Record Linkage with Variables in One File.'' \emph{Statistics in Medicine} 42(27): 4931--4951. \url{https://doi.org/10.1002/sim.9894}.

\bibitem[Kim and Chambers(2012a)]{kimchambers2012a}
Kim, Gunky, and Raymond Chambers. 2012a. ``Regression Analysis under Probabilistic Multi-Linkage.'' \emph{Statistica Neerlandica} 66(1): 64--79. \url{https://doi.org/10.1111/j.1467-9574.2011.00509.x}.

\bibitem[Kim and Chambers(2012b)]{kimchambers2012b}
Kim, Gunky, and Raymond Chambers. 2012b. ``Regression Analysis under Incomplete Linkage.'' \emph{Computational Statistics \& Data Analysis} 56(9): 2756--2770. \url{https://doi.org/10.1016/j.csda.2012.02.026}.

\bibitem[Lahiri and Larsen(2005)]{lahiri2005}
Lahiri, P., and Michael D. Larsen. 2005. ``Regression Analysis with Linked Data.'' \emph{Journal of the American Statistical Association} 100(469): 222--230. \url{https://doi.org/10.1198/016214504000001277}.

\bibitem[Lin(2013)]{lin2013}
Lin, Winston. 2013. ``Agnostic Notes on Regression Adjustments to Experimental Data: Reexamining Freedman's Critique.'' \emph{The Annals of Applied Statistics} 7(1): 295--318. \url{https://doi.org/10.1214/12-AOAS583}.

\bibitem[Sadinle(2017)]{sadinle2017}
Sadinle, Mauricio. 2017. ``Bayesian Estimation of Bipartite Matchings for Record Linkage.'' \emph{Journal of the American Statistical Association} 112(518): 600--612. \url{https://doi.org/10.1080/01621459.2016.1148612}.

\bibitem[Shan et al.(2021)Shan, Thomas, and Gutman]{shan2021}
Shan, Mingyang, Kali S. Thomas, and Roee Gutman. 2021. ``A Multiple Imputation Procedure for Record Linkage and Causal Inference to Estimate the Effects of Home-Delivered Meals.'' \emph{The Annals of Applied Statistics} 15(1): 412--436. \url{https://doi.org/10.1214/20-AOAS1397}.

\bibitem[Slawski et al.(2025)Slawski, West, Bukke, Wang, Diao, and Ben-David]{slawskietal2025}
Slawski, Martin, Brady T. West, Priyanjali Bukke, Zhenbang Wang, Guoqing Diao, and Emanuel Ben-David. 2025. ``A General Framework for Regression with Mismatched Data Based on Mixture Modelling.'' \emph{Journal of the Royal Statistical Society Series A: Statistics in Society} 188(3): 896--919. \url{https://doi.org/10.1093/jrsssa/qnae083}.

\bibitem[Steorts(2015)]{steorts2015}
Steorts, Rebecca C. 2015. ``Entity Resolution with Empirically Motivated Priors.'' \emph{Bayesian Analysis} 10(4): 849--875. \url{https://doi.org/10.1214/15-BA965SI}.

\bibitem[Tahamont et al.(2021)Tahamont, Jelveh, Chalfin, Yan, and Hansen]{tahamont2021}
Tahamont, Sarah, Zubin Jelveh, Aaron Chalfin, Shi Yan, and Benjamin Hansen. 2021. ``Dude, Where's My Treatment Effect? Errors in Administrative Data Linking and the Destruction of Statistical Power in Randomized Experiments.'' \emph{Journal of Quantitative Criminology} 37(3): 715--749. \url{https://doi.org/10.1007/s10940-020-09461-x}.

\bibitem[Thompson(2011)]{thompson2011}
Thompson, Samuel B. 2011. ``Simple Formulas for Standard Errors That Cluster by Both Firm and Time.'' \emph{Journal of Financial Economics} 99(1): 1--10. \url{https://doi.org/10.1016/j.jfineco.2010.08.016}.

\bibitem[Wortman and Reiter(2018)]{wortman2018}
Wortman, Joan Heck, and Jerome P. Reiter. 2018. ``Simultaneous Record Linkage and Causal Inference with Propensity Score Subclassification.'' \emph{Statistics in Medicine} 37(24): 3533--3546. \url{https://doi.org/10.1002/sim.7911}.

\bibitem[Zhang and Tuoto(2021)]{zhangtuoto2021}
Zhang, Li-Chun, and Tiziana Tuoto. 2021. ``Linkage-Data Linear Regression.'' \emph{Journal of the Royal Statistical Society Series A: Statistics in Society} 184(2): 522--547. \url{https://doi.org/10.1111/rssa.12630}.

\bibitem[Zhang and Tuoto(2024)]{zhangtuoto2024}
Zhang, Li-Chun, and Tiziana Tuoto. 2024. ``Categorical Linkage-Data Analysis.'' \emph{Statistics in Medicine} 43(18): 3463--3483. \url{https://doi.org/10.1002/sim.10134}.

\end{thebibliography}
\end{document}